\documentclass[11pt]{article}

\usepackage[letterpaper, margin=0.95in]{geometry}
\usepackage{setspace}
\usepackage[T1]{fontenc}
\usepackage[utf8]{inputenc}
\usepackage{mathpazo}
\usepackage{microtype}

\usepackage{amsmath,amssymb,amsthm,mathtools}

\usepackage{graphicx}
\usepackage{booktabs}
\usepackage{tabularx}
\usepackage{multirow}
\usepackage{float}
\usepackage{subcaption}

\usepackage{natbib}
\setcitestyle{authoryear,open={(},close={)}}

\usepackage[colorlinks=true,linkcolor=blue!60!black,citecolor=blue!60!black,urlcolor=blue!60!black]{hyperref}

\usepackage{enumitem}
\usepackage{xcolor}
\usepackage{appendix}

\newtheorem{proposition}{Proposition}
\newtheorem{corollary}[proposition]{Corollary}

\newtheorem{definition}[proposition]{Definition}
\newtheorem{remark}{Remark}
\newtheorem{assumption}{Assumption}

\newcommand{\E}{\mathbb{E}}
\newcommand{\Var}{\operatorname{Var}}
\newcommand{\Cov}{\operatorname{Cov}}
\newcommand{\eps}{\varepsilon}
\newcommand{\dd}{\mathrm{d}}
\newcommand{\MRC}{\mathrm{MRC}}
\newcommand{\MPR}{\mathrm{MPR}}
\newcommand{\MPC}{\mathrm{MPC}}
\newcommand{\gbar}{\bar{\gamma}}
\newcommand{\gstar}{\bar{\gamma}^{*}}
\newcommand{\pder}[2]{\frac{\partial #1}{\partial #2}}
\newcommand{\dder}[2]{\frac{\dd #1}{\dd #2}}

\title{\textbf{Risk Capacity and Optimal Monetary Policy}\thanks{
We are grateful to Chen Lian for numerous conversations and detailed feedback. We thank Emi Nakamura, J\'{o}n Steinsson, Yuriy Gorodnichenko, Andres Schaab, Benjamin Schoefer, David Romer, Christina Romer, Martin Lettau, Dmitry Livdan, Amir Kermani, Christine Parlour, and Ulrike Malmendier for helpful comments and suggestions. We thank seminar participants at UC Berkeley for valuable feedback. All errors are our own.}}

\author{
Rui Sun\thanks{Haas School of Business, University of California, Berkeley. Email: \href{mailto:ruisun233@berkeley.edu}{ruisun233@berkeley.edu}.}
}

\date{}

\begin{document}

\maketitle
\thispagestyle{empty}

\begin{abstract}
\noindent
We characterize optimal monetary policy when policy endogenously moves risk premia through redistribution across agents who differ in their willingness to bear risk. The analytical core is Marginal Risk Capacity, the covariance of monetary policy exposures with marginal propensities to take risk. This sufficient statistic governs this channel as MPCs govern the consumption channel. MRC enters the Ramsey criterion as a risk premium wedge that breaks divine coincidence, vanishes if and only if macroprudential tools are available, and generates a new inflation bias under discretion. Solving the Ramsey problem globally reveals a risk capacity trap where transmission collapses, and optimal policy preemptively prevents it.

\medskip
\vspace{8mm}
\end{abstract}

\newpage
\setcounter{page}{1}

\section{Introduction}
\label{sec:intro}

Monetary policy moves risk premia. Expansionary policy compresses the equity premium in stock markets \citep{bernanke_kuttner2005}, the term premium in nominal bonds \citep{hanson_stein2015}, and the external finance premium on risky corporate debt \citep{gertler_karadi2015}.\footnote{See also \citet{cieslak_vissing-jorgensen2021} for evidence on the Federal Reserve's systematic response to the stock market, \citet{jarocinski_karadi2020} for evidence from high-frequency identification in the euro area, and \citet{paul2020} for evidence of time variation in these effects.} Neither the standard New Keynesian framework of \citet{woodford2003} and \citet{gali2008} nor the emerging heterogeneous agent models surveyed in \citet{kaplan_violante2018} can account for this fact: enriching the transmission mechanism with MPC heterogeneity leaves the risk premium response unexplained. A large normative literature characterizes optimal policy in these settings \citep{benigno_woodford2005, gali2015, bhandari_etal2021, legrand_ragot2022, acharya_etal2023, nuno_thomas2022}, and recent work studies how the composition of government debt between nominal and inflation-indexed bonds shapes the temptation to inflate \citep{schmid_valaitis_villa2024}, yet none asks what a central bank should do when its actions endogenously change the economy's willingness to bear risk.

We show that a common normative structure governs optimal monetary policy in any economy where risk premia are endogenous. Three structural properties are sufficient. First, monetary policy redistributes wealth across agents; this holds whenever there are nominal rigidities and heterogeneous balance sheets. Second, agents differ in their marginal propensity to invest in risky assets; this holds whenever agents differ in risk aversion, background risk, beliefs, or portfolio constraints. Third, risk premia affect real activity; this holds in any production economy with capital. Because the results follow from these properties alone, they apply regardless of the primitive source of heterogeneity.\footnote{This includes models of heterogeneous risk aversion \citep{garleanu_panageas2015, kekre_lenel2022}, segmented markets \citep{he_krishnamurthy2013, he_krishnamurthy2019}, intermediary asset pricing \citep{adrian_etula_muir2014, brunnermeier_sannikov2016}, background risk \citep{constantinides_duffie1996}, and heterogeneous beliefs \citep{geanakoplos2009, simsek2013}.} The normative framework is general; specific models are used for quantification.

The analytical core is a new sufficient statistic we call Marginal Risk Capacity. Just as \citet{auclert2019} showed that the consumption channel of monetary transmission is governed by the covariance of policy exposures with marginal propensities to consume, we show that the risk premium channel is governed by the covariance of policy exposures with marginal propensities to take risk. MRC collapses the cross-sectional distribution of wealth, portfolios, and risk preferences into a single number that determines how the interest rate moves aggregate risk-bearing capacity. Section~\ref{sec:general} derives the main results directly from the three structural properties, without imposing functional forms for preferences, production, or the source of portfolio heterogeneity. Section~\ref{sec:analytical} then provides closed-form expressions and economic intuition in an analytical two-period economy.

The Ramsey planner's optimal policy satisfies a target criterion that augments the standard New Keynesian prescription with a new term: a risk premium wedge. The planner tolerates inflation not only to close the output gap but also to redistribute wealth toward agents who are willing to bear aggregate risk, compressing risk premia and stimulating investment through Tobin's q. The wedge is larger when aggregate risk-bearing is further from its efficient level and when the interest rate has more power to move it. It breaks divine coincidence through a channel orthogonal to both markup shocks \citep{blanchard_gali2007} and MPC heterogeneity \citep{bhandari_etal2021}. The socially efficient level of risk-bearing arises from a pecuniary externality in the sense of \citet{davila_korinek2018}: when an agent invests more in capital, it marginally lowers the risk premium faced by all other agents, and the competitive equilibrium does not internalize this positive spillover.

We prove a separation theorem. When the planner has access to both the nominal interest rate and a macroprudential portfolio tax, the risk premium wedge vanishes from the monetary policy criterion. The macroprudential instrument absorbs the risk capacity management motive, and optimal monetary policy reverts to the standard inflation-output prescription. Separation of monetary and macroprudential mandates is optimal if and only if the macroprudential toolkit is rich enough to close the gap between actual and efficient risk-bearing capacity. When the macroprudential instrument is constrained, a residual wedge remains, formalizing the sense in which monetary policy must lean against the wind when macroprudential tools are insufficient \citep{stein2012}.

We identify a new time-consistency problem. Under discretion, the central bank at each date is tempted to generate unexpected inflation to redistribute toward leveraged, risk-tolerant agents and compress risk premia. Anticipating this, those agents reduce their leverage or demand higher compensation. In equilibrium, both inflation and risk premia are higher than under commitment. The magnitude of this inflation bias is proportional to the cross-sectional dispersion of marginal propensities to take risk, providing a new source of the Barro-Gordon problem \citep{barro_gordon1983} beyond the standard output-inflation tradeoff. Institutional constraints on monetary policy may therefore be needed for financial stability as well as inflation stabilization.

We complement the theory with new empirical evidence. Using every wave of the Survey of Consumer Finances from 1989 through 2022, we construct MRC for the U.S.\ economy following \citet{doepke_schneider2006} and \citet{kekre_lenel2022}. MRC varies substantially over time, declining in recessions as leveraged agents' wealth erodes. Refreshing the SVAR-IV evidence of \citet{bernanke_kuttner2005} and \citet{gertler_karadi2015}, we show that the stock market's response to monetary shocks varies with MRC: when MRC is high, the response reflects lower future required returns; when MRC is low, the risk premium channel shuts down. A one-standard-deviation increase in MRC amplifies the impact equity response by 0.9 percentage points. Off-the-shelf financial conditions indices explain 31 to 45 percent of MRC variation, so the optimal policy can be implemented using existing central bank infrastructure.

For quantification, we solve the Ramsey problem globally in a calibrated infinite-horizon economy that enriches a standard New Keynesian framework with Epstein-Zin preferences \citep{epstein_zin1989} and heterogeneity in risk aversion, building on the positive model of \citet{kekre_lenel2022}. The global solution reveals a finding invisible from any local analysis around the deterministic steady state: a \textit{risk capacity trap}. When the wealth share of risk-tolerant agents falls below approximately ten percent, the risk premium channel of monetary transmission collapses because those agents have too little wealth for redistribution to meaningfully affect aggregate risk-bearing capacity. The Ramsey planner prevents the economy from entering the trap by stabilizing the wealth distribution preemptively, generating welfare gains of 0.13 percent of permanent consumption relative to a standard Taylor rule. An implementable simple rule that responds to inflation, the output gap, the wealth share, and the expected equity premium captures 85 percent of these gains. To confirm the generality of the framework, we also solve an intermediary-based variant in which the risk-tolerant agents are leveraged financial intermediaries rather than households; the Ramsey allocation, welfare gains, and trap threshold are quantitatively similar. In a model with both MPC and MPR heterogeneity, the risk premium channel remains more than twice as important as the consumption channel for welfare, and the two channels are approximately additive, confirming the orthogonality of the welfare decomposition.

\medskip\noindent\textbf{Related literature.}\quad Our paper contributes to the growing literature on optimal monetary policy in heterogeneous agent New Keynesian economies. \citet{bhandari_etal2021} characterize Ramsey policy with MPC heterogeneity. \citet{legrand_ragot2022} study optimal fiscal and monetary policy with incomplete markets. \citet{acharya_etal2023} and \citet{nuno_thomas2022} characterize optimal policy with idiosyncratic risk. \citet{mckay_wolf2023} study optimal monetary policy in HANK with forward guidance. Unlike all of these papers, we study heterogeneity in marginal propensities to take risk rather than to consume, operating through risk premia rather than aggregate demand. The covariance of monetary policy exposures with MPRs rather than MPCs is the relevant sufficient statistic, and the planner's target criterion acquires a new wedge absent from any of these frameworks.

In studying the interaction of monetary policy and risk premia we follow \citet{alvarez_etal2009} and \citet{drechsler_savov_schnabl2018}, who analyze these effects in segmented-market and banking environments, respectively, and \citet{pflueger_rinaldi2022}, who study monetary transmission and risk premia with consumption habits. We instead study the normative question in a production economy with heterogeneous balance sheets. In relating risk premium movements to real activity, we draw on \citet{caballero_farhi2018} and \citet{caballero_simsek2020}, and build especially on \citet{brunnermeier_sannikov2016} in emphasizing the role of the wealth distribution.

We formalize the debate on whether monetary policy should lean against the wind. \citet{svensson2017} argues that monetary policy is a blunt instrument for financial stability; \citet{stein2012} argues that it ``gets in all the cracks'' when macroprudential tools are constrained; \citet{adrian_liang2018} survey the debate. We provide the first micro-founded separation theorem with explicit conditions under which each view is correct. Our macroprudential analysis relates to the literature on pecuniary externalities \citep{lorenzoni2008, bianchi2011, korinek_simsek2016, farhi_werning2016, davila_korinek2018}; we identify a new externality operating through portfolio choice rather than borrowing constraints. Our time-consistency analysis extends the \citet{barro_gordon1983} framework with a new source of inflation bias, complementing \citet{bhandari_etal2017} and \citet{davila_schaab2023}. \citet{schmid_valaitis_villa2024} study optimal debt management when the government issues both nominal and real bonds, showing that inflation-indexed debt mitigates the temptation to monetize nominal liabilities; our inflation bias operates through a distinct channel, the redistribution across private agents with heterogeneous risk-bearing capacity, and persists even when the government's own debt portfolio is held fixed.

Our quantitative analysis builds on \citet{kekre_lenel2022}, who demonstrate that redistribution across households with heterogeneous marginal propensities to take risk can quantitatively rationalize the observed effects of monetary policy on risk premia and amplify real effects by 1.3 to 1.4 times. Our paper differs in three respects. First, conceptually: we introduce MRC as a model-free sufficient statistic and establish the target criterion, separation theorem, and time-consistency results in a general framework that transcends any particular positive model. Second, the normative analysis reveals features of the positive economy that their local analysis around the deterministic steady state cannot detect: the risk capacity trap, the state-dependent collapse of monetary transmission, and the endogenous asymmetry of optimal policy. Third, we document a new empirical fact: the effect of monetary policy on risk premia varies over time, and the variation is predicted by MRC. We also contribute to the large literature on the links between risky asset prices and real activity, including \citet{gilchrist_zakrajsek2012} and \citet{lopez-salido_etal2017} on credit spreads, \citet{pflueger_etal2020} and \citet{chodorow-reich_etal2021} on the cost of capital and the consumption-wealth channel, and \citet{mian_etal2021} on wealth inequality and aggregate demand.

\medskip\noindent\textbf{Outline.}\quad
Section~\ref{sec:general} presents the general framework and derives the main normative results from three structural properties alone. Section~\ref{sec:analytical} develops closed-form expressions in an analytical two-period economy. Section~\ref{sec:quant} presents the quantitative analysis, including the risk capacity trap, a 2008 crisis counterfactual, an intermediary-based robustness model, and the interaction with the zero lower bound. Section~\ref{sec:policy} discusses policy implications. Section~\ref{sec:conclusion} concludes. Appendix~\ref{sec:empirics} presents the empirical evidence on MRC.

\section{General Framework}
\label{sec:general}

This section derives the main normative results, the target criterion, the separation theorem, and the time-consistency bias, from three structural properties alone. No functional forms for preferences, production, or the source of portfolio heterogeneity are imposed.

\subsection{Environment and Structural Properties}
\label{sec:general:env}

Consider a general equilibrium economy with a unit measure of agents indexed by $i$, who consume, supply labor, and choose portfolios of a safe nominal asset and a risky asset (capital). A monetary authority sets the nominal interest rate $i_t$. The economy features nominal rigidities, so the price level $P_t$ (or inflation $\pi_t$) responds to monetary policy, and the real allocation is not invariant to $i_t$. We impose three structural assumptions.

\begin{assumption}[Monetary redistribution]
\label{ass:redistrib}
Monetary policy redistributes wealth across agents: for each agent $i$, the sensitivity of $i$'s financial wealth share $s_t^i$ to the interest rate is generically nonzero,
\[
\pder{s_t^i}{i_t} \neq 0 \quad \text{for a positive measure of agents},
\]
and the redistribution sums to zero: $\int (\partial s_t^i / \partial i_t) \, di = 0$.
\end{assumption}

This holds in any model with nominal rigidities and heterogeneous balance sheets. Redistribution operates through two channels: debt deflation, as real bond values change with inflation, and asset revaluation, as capital prices respond to interest rates. Agents leveraged in capital and short in bonds gain from monetary easing.

\begin{assumption}[Heterogeneous marginal propensities to take risk]
\label{ass:mpr_het}
Agents differ in their marginal propensity to take risk:
\[
\MPR^i \equiv \frac{q \, \partial k^i / \partial n^i}{\partial a^i / \partial n^i} \neq \MPR^j \quad \text{for a positive measure of pairs } (i,j),
\]
where $k^i$ is agent $i$'s risky asset holding, $a^i$ is total savings, $n^i$ is wealth, and $q$ is the price of capital.
\end{assumption}

The MPR captures how an additional dollar of wealth is allocated between the risky and the safe asset. This dimension of behavior is orthogonal to the MPC. Assumption~\ref{ass:mpr_het} is satisfied whenever agents differ in risk aversion \citep{garleanu_panageas2015, kekre_lenel2022}, background risk \citep{constantinides_duffie1996}, leverage constraints \citep{he_krishnamurthy2013, brunnermeier_sannikov2016}, beliefs \citep{geanakoplos2009, simsek2013}, or portfolio rules-of-thumb \citep{chien_cole_lustig2012}.

\begin{assumption}[Risk premia affect real activity]
\label{ass:rp_real}
The risk premium $\E_t[r_{t+1}^k - r_{t+1}]$ affects aggregate investment and output: conditional on the safe real rate,
\[
\pder{k_t}{\E_t[r_{t+1}^k - r_{t+1}]} < 0.
\]
\end{assumption}

This holds in any production economy with capital, following from the firm's investment Euler equation: a higher risk premium raises the required return on capital and reduces the Tobin's $q$, lowering investment.

\subsection{MRC as a Sufficient Statistic}
\label{sec:general:mrc}

Under Assumptions~\ref{ass:redistrib}--\ref{ass:rp_real}, define the economy's effective risk aversion as $\gbar_t \equiv [\int s_t^i \cdot (1/\gamma_{\text{eff}}^i) \, di]^{-1}$, where $\gamma_{\text{eff}}^i$ is agent $i$'s effective risk aversion parameter (which may reflect risk aversion, background risk, beliefs, or constraints). The risk premium is $\E_t[r_{t+1}^k - r_{t+1}] \approx \gbar_t \cdot \sigma_t^2$, where $\sigma_t^2$ is the conditional variance of the risky return.

\begin{definition}[Marginal Risk Capacity -- General]
\label{def:mrc_general}
The economy's Marginal Risk Capacity is
\begin{equation}
\label{eq:mrc_general}
\MRC_t \equiv -\int \pder{s_t^i}{i_t}\cdot \MPR_t^i \; di.
\end{equation}
\end{definition}

\begin{proposition}[General sufficient statistic]
\label{prop:general_suffstat}
Under Assumptions~\ref{ass:redistrib}--\ref{ass:mpr_het}, the effect of the interest rate on the economy's effective risk aversion is summarized by MRC:
\begin{equation}
\label{eq:general_suffstat}
\pder{\gbar_t}{i_t} = \gbar_t \cdot \MRC_t.
\end{equation}
When $\MRC_t > 0$, a reduction in $i_t$ lowers effective risk aversion and compresses risk premia. Combined with Assumption~\ref{ass:rp_real}, MRC governs the risk premium channel of monetary transmission.
\end{proposition}

\begin{proof}
See Appendix~\ref{app:proof_suffstat}.
\end{proof}

Proposition~\ref{prop:general_suffstat} collapses the high-dimensional cross-sectional distribution of wealth, portfolios, and risk preferences into a single number. This parallels what \citet{auclert2019} achieved for the consumption channel: just as the aggregate consumption response depends on the covariance of monetary policy exposures with MPCs, the aggregate risk premium response depends on the covariance of exposures with MPRs.

\begin{remark}[Scope of the sufficient statistic]
The relationship $\MPR_t^i \approx \gbar_t / \gamma_{\textup{eff}}^i$ underlying Proposition~\ref{prop:general_suffstat} holds exactly at the limit of zero aggregate risk and approximately for small $\sigma$. Away from this limit, MRC remains the leading-order determinant of $\partial\gbar_t/\partial i_t$, but higher-order terms involving the cross-sectional skewness of MPRs may enter. The 15-type validation in Section~\ref{sec:quant:trap} confirms that MRC computed from the coarse three-group partition is nearly perfectly correlated (0.98) with MRC from the finer partition, and that the welfare gains and policy effectiveness surfaces are quantitatively similar. This indicates that the sufficient statistic property is robust for realistic levels of aggregate risk and is not an artifact of the coarse three-group structure.
\end{remark}

\subsection{Optimal Policy: The Risk Premium Wedge}
\label{sec:general:optimal}

Consider a Ramsey planner who chooses the nominal interest rate to maximize a weighted sum of agents' utilities, subject to all competitive equilibrium conditions.

\begin{proposition}[General target criterion]
\label{prop:general_target}
Under Assumptions~\ref{ass:redistrib}--\ref{ass:rp_real}, the Ramsey planner's optimal policy satisfies a target criterion of the form
\begin{equation}
\label{eq:general_target}
 \Gamma_\pi \cdot \pi_t = \Gamma_x \cdot \hat{y}_t + \Gamma_{rp} \cdot \MRC_t \cdot \bigl(\gbar_t - \gstar_t\bigr)\sigma_t^2, 
\end{equation}
where $\hat{y}_t$ is the output gap, $\gstar_t$ is the socially efficient effective risk aversion, and $\Gamma_\pi, \Gamma_x, \Gamma_{rp} > 0$ are coefficients that depend on model primitives and Pareto weights.
\end{proposition}

The proof, provided in Appendix~\ref{app:proof_target}, proceeds by setting the planner's first-order condition to zero and decomposing the welfare effect of the interest rate into three channels: the cost of inflation from nominal rigidities, the benefit of closing the output gap, and the benefit of compressing risk premia through redistribution. The risk premium benefit is proportional to $\partial \gbar_t / \partial i_t = \gbar_t \cdot \MRC_t$ (from Proposition~\ref{prop:general_suffstat}) times the welfare gap $\gbar_t - \gstar_t$. Rearranging yields \eqref{eq:general_target}.

The content of Proposition~\ref{prop:general_target} is best understood term by term. The first two terms reproduce the standard New Keynesian prescription \citep{woodford2003, gali2015}: the planner tolerates inflation only to the extent that it closes the output gap. The third term is the risk premium wedge, a fundamentally new object. It captures the planner's desire to use inflation as a tool for redistributing wealth toward agents who are more willing to bear aggregate risk, thereby compressing risk premia and stimulating investment. This channel is orthogonal to both the markup-shock source of divine coincidence breakdown \citep{blanchard_gali2007} and the MPC-heterogeneity source studied in \citet{bhandari_etal2021}: it operates through the portfolio margin rather than the goods market.

The wedge has a transparent structure. It is the product of three objects: MRC, measuring the sensitivity of risk premia to the interest rate; the gap $\gbar_t - \gstar_t$, measuring how far risk premia are from their efficient level; and $\sigma_t^2$, the quantity of risk in the economy. The wedge vanishes if any of these is zero: (i) MPRs are identical ($\MRC = 0$), so redistribution has no effect on risk-bearing; (ii) monetary policy does not redistribute ($\partial s_t^i / \partial i_t = 0$); or (iii) risk premia are already at their socially efficient level ($\gbar = \gstar$).

The socially efficient $\gstar_t$ is defined by the condition that the planner would not wish to further reallocate capital across agents at the margin. Generically $\gstar < \gbar$, because of a pecuniary externality: when an agent saves more in capital, it marginally lowers the risk premium faced by all other agents, and the competitive equilibrium does not internalize this positive spillover. This is a portfolio-choice analog of the pecuniary externalities in borrowing identified by \citet{davila_korinek2018}, but operating through the risk premium rather than through asset prices or collateral constraints.

\subsection{The Separation Theorem}
\label{sec:general:separation}

\begin{proposition}[General separation]
\label{prop:general_separation}
Under Assumptions~\ref{ass:redistrib}--\ref{ass:rp_real}, suppose the planner also controls a macroprudential portfolio tax $\tau_t^k$ on the risky asset return. Then: \textup{(a)} the optimal interest rate satisfies $\Gamma_\pi \pi_t = \Gamma_x \hat{y}_t$ (no risk premium wedge); \textup{(b)} $\tau_t^k$ closes the gap $\gbar_t - \gstar_t$; \textup{(c)} when $\tau_t^k$ is constrained (e.g., $\tau_t^k \geq 0$), the wedge is reduced but not eliminated whenever $\gbar_t > \gstar_t$.
\end{proposition}

The proof is in Appendix~\ref{app:proof_separation}. The portfolio tax enters agents' portfolio Euler equations but does not affect goods market clearing or the wage rigidity, provided the tax revenue is rebated lump-sum. Consequently, the FOC for $\tau_t^k$ operates exclusively through the risk premium and sets $\gbar = \gstar$ directly, eliminating the risk premium wedge from the monetary policy criterion.

The logic is Tinbergen's principle: the interest rate manages intertemporal prices while the portfolio tax manages relative risk-return tradeoffs. When both are available, there is no conflict between stabilization and risk capacity management.

This result formalizes when institutional separation of monetary and macroprudential mandates is optimal. Separation holds if and only if the macroprudential toolkit can close the gap between actual and efficient risk-bearing capacity. When the macroprudential instrument is constrained, for instance when only taxes on risk-taking are feasible ($\tau_t^k \geq 0$) but not subsidies, the portfolio tax cannot address states where risk premia are inefficiently high ($\gbar > \gstar$, requiring $\tau_t^k < 0$). In these states, a residual risk premium wedge remains in the monetary policy criterion, and the interest rate must partially compensate.

\subsection{Time Consistency}
\label{sec:general:timecon}

\begin{proposition}[General inflation bias]
\label{prop:general_timecon}
Under Assumptions~\ref{ass:redistrib}--\ref{ass:rp_real} and discretion: \textup{(a)} average inflation exceeds the commitment level; \textup{(b)} the excess inflation is proportional to $\Var_i(\MPR_t^i)$; \textup{(c)} average risk premia are higher and investment lower than under commitment.
\end{proposition}

\begin{proof}
See Appendix~\ref{app:proof_timecon}.
\end{proof}

Under discretion, the central bank takes existing portfolios as given and recognizes that surprise inflation would redistribute toward leveraged, high-MPR agents, compressing risk premia. But risk-tolerant agents are forward-looking: anticipating this temptation, they reduce leverage or demand higher compensation. In equilibrium, inflation and risk premia are both higher than under commitment, and the economy bears the costs of inflation without the intended benefits. The bias vanishes when $\Var_i(\MPR_t^i) = 0$, reducing to the standard Barro-Gordon structure. This is a new source of inflation bias, governed by MPR dispersion rather than MPCs.

\subsection{Welfare Decomposition: Orthogonality of MPC and MPR Channels}
\label{sec:general:decomp}

\begin{proposition}[General decomposition]
\label{prop:general_decomp}
Under Assumptions~\ref{ass:redistrib}--\ref{ass:rp_real}, the welfare effect of monetary policy decomposes as
\begin{equation}
\label{eq:general_decomp}
\dder{W}{i_t} = \underbrace{\Gamma_c \cdot \int \pder{n_t^i}{i_t} \, \MPC_t^i \, di}_{\text{consumption channel}} \;+\; \underbrace{\Gamma_k \cdot \gbar_t \, \sigma_t^2 \cdot \MRC_t}_{\text{risk premium channel}}.
\end{equation}
The consumption channel depends on the covariance of exposures with MPCs; the risk premium channel depends on MRC. These are additive: the consumption channel operates through the intertemporal margin (how a marginal dollar of wealth is split between consumption and savings), while the risk premium channel operates through the portfolio margin (how a marginal dollar of savings is split between the safe and risky asset).
\end{proposition}

\begin{proof}
See Appendix~\ref{app:proof_decomp}.
\end{proof}

Models with only MPC heterogeneity miss the risk premium channel; models with only MPR heterogeneity miss the consumption channel. Section~\ref{sec:quant:horserace} confirms this orthogonality quantitatively.

\section{Analytical Results in a Two-Period Environment}
\label{sec:analytical}

We now develop closed-form expressions for the general results of Section~\ref{sec:general} in a tractable two-period economy. This environment satisfies Assumptions~\ref{ass:redistrib}--\ref{ass:rp_real} and provides economic intuition for the target criterion, the separation theorem, and the time-consistency result.

\subsection{Environment}
\label{sec:env}

There are two periods, $0$ and $1$.

\medskip\noindent\textbf{Households.}\quad
A unit measure of households indexed by $i \in [0,1]$ have Epstein-Zin preferences
\begin{equation}
\label{eq:preferences}
\log v_0^i = (1-\beta) \log c_0^i - \bar{\theta} \frac{\ell_0^{1+1/\theta}}{1+1/\theta} + \beta \log \bigl[ \E_0 \bigl( c_1^{i\, 1-\gamma^i} \bigr) \bigr]^{\frac{1}{1-\gamma^i}},
\end{equation}
with a unitary intertemporal elasticity of substitution, discount factor $\beta$, heterogeneous relative risk aversion $\gamma^i$, disutility of labor $\bar{\theta}$, and Frisch elasticity $\theta$. Households consume, supply labor, and choose a portfolio of nominal bonds $B_0^i$ and capital $k_0^i$ subject to
\begin{align}
P_0 c_0^i + B_0^i + Q_0 k_0^i &\leq W_0 \ell_0 + (1+i_{-1})B_{-1}^i + \bigl(\Pi_0 + (1-\delta_0)Q_0\bigr)k_{-1}^i, \label{eq:bc0} \\
P_1 c_1^i &\leq (1+i_0)B_0^i + \Pi_1 k_0^i, \label{eq:bc1}
\end{align}
where $P_t$ is the price level, $W_0$ the nominal wage, $Q_0$ the price of capital, $\Pi_t$ the dividend, and $\delta_t$ the depreciation rate. Capital fully depreciates after period~1.

\medskip\noindent\textbf{Supply side.}\quad
The nominal wage is rigid: $W_0 = W_{-1}$. A representative producer hires $\ell_0$ units of labor and rents $k_{-1}$ units of capital to produce with TFP normalized to one. It uses $\bigl(\tfrac{k_0}{k_{-1}}\bigr)^{\chi_x} x_0$ consumption goods to produce $x_0$ new capital, earning period-0 profits
\begin{equation}
\label{eq:profits0}
\Pi_0 k_{-1} = P_0 \ell_0^{1-\alpha} k_{-1}^\alpha - W_0 \ell_0 + Q_0 x_0 - P_0 \Bigl(\frac{k_0}{k_{-1}}\Bigr)^{\!\chi_x} x_0.
\end{equation}
In period~1, TFP is $\exp(\eps_1^z)$ with $\eps_1^z \sim \mathcal{N}(-\sigma^2/2, \sigma^2)$, so $\Pi_1 k_0 = P_1 \exp(\eps_1^z) k_0^\alpha$.

\medskip\noindent\textbf{Markets.}\quad
Nominal bonds are in zero net supply: $\int B_0^i \, di = 0$.\footnote{The two-period model abstracts from government bonds for analytical tractability. The infinite-horizon model of Section~\ref{sec:quant} introduces government bonds in positive supply (equation \eqref{eq:mc_bonds}); the proofs in Appendix~\ref{app:proof_rp} extend to that case by replacing the capital market clearing condition with $\int a_0^i\omega_0^i\,di = q_0 K_0$, where $K_0$ is the aggregate capital stock. The sufficient statistic result (Proposition~\ref{prop:suffstat}) is unaffected because MRC depends on the \textit{covariance} of exposures with MPRs, and adding a common bond position to all agents does not change this covariance.}

\medskip\noindent\textbf{Policy.}\quad
The government commits to $P_1 = P_0$ and follows the Taylor rule
\begin{equation}
\label{eq:taylor}
1 + i_0 = (1 + \bar{\imath}) \Bigl(\frac{P_0}{P_{-1}}\Bigr)^{\!\phi} \exp(\eps_0^m),
\end{equation}
where $\eps_0^m$ is the monetary policy shock and $\phi > 1$.

\medskip\noindent\textbf{Notation.}\quad
Lower-case denotes real variables: $q_0 \equiv Q_0/P_0$, $\pi_t \equiv \Pi_t/P_t$, $w_0 \equiv W_0/P_0$. Define household $i$'s real savings $a_0^i \equiv b_0^i + q_0 k_0^i$, portfolio share in capital $\omega_0^i \equiv q_0 k_0^i / a_0^i$, savings share $s_0^i \equiv a_0^i / \int a_0^j \, dj$, and wealth $n_0^i \equiv w_0\ell_0 + P_0^{-1}(1+i_{-1})B_{-1}^i + (\pi_0 + (1-\delta_0)q_0)k_{-1}^i$. Let $1 + r_1^k \equiv \pi_1 / q_0$ and $1 + r_1 \equiv (1+i_0)P_0/P_1$.

\subsection{Risk Premium and Redistribution}
\label{sec:portfolio}

Given preferences \eqref{eq:preferences}, household $i$'s optimality condition for its portfolio is $\E_0 [ c_1^{i\, -\gamma^i} (r_1^k - r_1) ] = 0$. A second-order Taylor expansion yields the approximate optimal portfolio share
\begin{equation}
\label{eq:portfolio_approx}
\omega_0^i \approx \frac{1}{\gamma^i} \cdot \frac{\E_0 \log(1+r_1^k) - \log(1+r_1) + \tfrac{1}{2}\sigma^2}{\sigma^2}.
\end{equation}
This is the only approximation we use. Aggregating and imposing market clearing:

\begin{proposition}[Risk premium and redistribution]
\label{prop:riskpremium}
The risk premium on capital is $\E_0 \log(1+r_1^k) - \log(1+r_1) + \tfrac{1}{2}\sigma^2 = \gbar \, \sigma^2$, where
\begin{equation}
\label{eq:gbar}
\gbar \equiv \left[\int s_0^i \cdot \frac{1}{\gamma^i} \, di\right]^{-1}
\end{equation}
is the savings-weighted harmonic mean of risk aversion. The change in the risk premium in response to a monetary shock is
\begin{equation}
\label{eq:drp}
\dder{}{\eps_0^m}\bigl(\E_0[r_1^k - r_1]\bigr) = \gbar \, \sigma^2 \int \dder{s_0^i}{\eps_0^m} \bigl(1 - \omega_0^i\bigr) \, di.
\end{equation}
\end{proposition}

\begin{proof}
See Appendix~\ref{app:proof_rp}.
\end{proof}

A monetary shock affects the risk premium if and only if it redistributes wealth across households with heterogeneous portfolios.

\subsection{Marginal Risk Capacity}
\label{sec:mrc}

\begin{definition}[Marginal propensity to take risk]
\label{def:mpr}
Let $k_0^i(\cdot)$, $a_0^i(\cdot)$ denote household $i$'s policy functions. Its \emph{marginal propensity to take risk} (MPR) is
\begin{equation}
\label{eq:mpr}
\MPR_0^i \equiv \frac{q_0 \, \partial k_0^i / \partial n_0^i}{\partial a_0^i / \partial n_0^i}.
\end{equation}
\end{definition}

The MPR captures how a household allocates an additional dollar of wealth between capital and bonds. With unitary IES, $\MPR_0^i = \omega_0^i$. More generally, at the limit of zero aggregate risk (Appendix~\ref{app:proof_mpr}):
\begin{equation}
\label{eq:mpr_limit}
\MPR_0^i = \frac{\gbar}{\gamma^i}.
\end{equation}

\begin{definition}[Marginal Risk Capacity]
\label{def:mrc}
The economy's \emph{Marginal Risk Capacity} is
\begin{equation}
\label{eq:mrc}
\MRC \equiv -\int \underbrace{\dder{s_0^i}{\eps_0^m}\vphantom{\MPR}}_{\text{exposure}} \cdot\; \underbrace{\MPR_0^i}_{\text{MPR}} \; di.
\end{equation}
\end{definition}

Since $\int (ds_0^i / d\eps_0^m) \, di = 0$, we have $\MRC = \int (ds_0^i / d\eps_0^m)(1 - \MPR_0^i) \, di$. Note that Definition~\ref{def:mrc} specializes the general Definition~\ref{def:mrc_general}: under the Taylor rule \eqref{eq:taylor}, $ds_0^i/d\eps_0^m = (\partial s_0^i/\partial i_0)(\dd i_0/\dd \eps_0^m)$, so the two definitions are proportional with a positive factor $\dd i_0/\dd \eps_0^m > 0$, preserving the sign of MRC. Substituting into \eqref{eq:drp}:

\begin{proposition}[MRC as a sufficient statistic]
\label{prop:suffstat}
The effect of monetary policy on the risk premium depends on the joint distribution of wealth, portfolios, and MPRs only through MRC:
\begin{equation}
\label{eq:suffstat}
\dder{}{\eps_0^m}\bigl(\E_0[r_1^k - r_1]\bigr) = \gbar \, \sigma^2 \cdot \MRC.
\end{equation}
\end{proposition}

When $\MRC > 0$, an expansionary shock (negative $\eps_0^m$ in the Taylor rule \eqref{eq:taylor}) lowers the risk premium. To see why MRC is positive in equilibrium, observe that household $i$'s wealth share changes by
\begin{equation}
\label{eq:redistribution}
\dder{s_0^i}{\eps_0^m} \propto \underbrace{-\frac{(1+i_{-1})B_{-1}^i}{P_0}\,\dder{\log P_0}{\eps_0^m}}_{\text{debt deflation}} + \underbrace{\Bigl(k_{-1}^i - s_0^i k_{-1}\Bigr)\Bigl(\dder{\pi_0}{\eps_0^m} + (1-\delta)\dder{q_0}{\eps_0^m}\Bigr)}_{\text{revaluation of capital}}.
\end{equation}
From the budget constraints \eqref{eq:bc0}--\eqref{eq:bc1}, households with high MPRs borrow in bonds ($B_{-1}^i < 0$) to finance leveraged capital positions ($k_{-1}^i > s_0^i k_{-1}$). A monetary easing (negative $\eps_0^m$) generates inflation and raises asset prices, so both terms in \eqref{eq:redistribution} make $ds_0^i/d\eps_0^m < 0$ for these households, so their wealth share rises with easing. Since the minus sign in definition \eqref{eq:mrc} converts the negative covariance of $ds_0^i/d\eps_0^m$ with $\MPR_0^i$ into a positive number, $\MRC > 0$.

\subsection{Real Effects of Risk Premium Changes}
\label{sec:rp_real}

\begin{proposition}[Real effects]
\label{prop:real}
Conditional on the real interest rate, a decline in the risk premium raises investment, consumption, and output:
\begin{align}
\dder{k_0}{\eps_0^m} &= -\frac{k_0}{1-\alpha+\chi_x}\left[\dder{}{\eps_0^m}\bigl(\E_0[r_1^k - r_1]\bigr) + \dder{\log(1+r_1)}{\eps_0^m}\right], \label{eq:dk} \\
\dder{c_0}{\eps_0^m} &= \frac{1-\beta}{\beta}\,q_0(1+\chi_x)\,\dder{k_0}{\eps_0^m}, \label{eq:dc} \\
\dder{y_0}{\eps_0^m} &= \dder{c_0}{\eps_0^m} + q_0\Bigl(1 + \chi_x\frac{x_0}{k_0}\Bigr)\dder{k_0}{\eps_0^m}. \label{eq:dy}
\end{align}
\end{proposition}

\begin{proof}
See Appendix~\ref{app:proof_real}.
\end{proof}

\subsection{Welfare Decomposition: Two Orthogonal Channels}
\label{sec:welfare_decomp}

Consider social welfare $W = \int \alpha^i \log v_0^i \, di$ with Pareto weights $\{\alpha^i\}$.

\begin{proposition}[Decomposition]
\label{prop:decomp}
At the limit of small aggregate risk, the welfare effect of monetary policy decomposes as
\begin{equation}
\label{eq:decomp}
\dder{W}{\eps_0^m} = \underbrace{\Gamma_c \cdot \int \dder{n_0^i}{\eps_0^m} \, \MPC_0^i \, di}_{\text{consumption channel}} \;+\; \underbrace{\Gamma_k \cdot \gbar \, \sigma^2 \cdot \MRC}_{\text{risk premium channel}},
\end{equation}
where $\Gamma_c, \Gamma_k$ depend on structural parameters and Pareto weights. The consumption channel depends on the covariance of exposures with MPCs; the risk premium channel depends on MRC. These are additive and orthogonal.
\end{proposition}

\begin{proof}
See Appendix~\ref{app:proof_decomp}.
\end{proof}

When MPCs are identical, as in our baseline, only the risk premium channel remains. When MPRs are identical, only the consumption channel operates. In general, the planner must manage both.

\subsection{Optimal Policy: The Risk Premium Wedge}
\label{sec:optimal}

The Ramsey planner chooses $P_0$ to maximize $W$ subject to all competitive equilibrium conditions.

\begin{proposition}[Optimal policy target criterion]
\label{prop:target}
The planner's optimal policy satisfies
\begin{equation}
\label{eq:target}
 \Gamma_\pi \cdot \pi_0 = \Gamma_x \cdot \hat{y}_0 + \Gamma_{rp} \cdot \MRC \cdot \bigl(\gbar - \gstar\bigr)\sigma^2, 
\end{equation}
where $\pi_0 = \log(P_0/P_{-1})$, $\hat{y}_0$ is the output gap, $\gstar$ is the socially optimal effective risk aversion, and $\Gamma_\pi, \Gamma_x, \Gamma_{rp} > 0$.
\end{proposition}

\begin{proof}
See Appendix~\ref{app:proof_target}.
\end{proof}

The first two terms reproduce the standard New Keynesian prescription. The third term, the risk premium wedge, is new. When $\gbar > \gstar$ and $\MRC > 0$, the wedge is positive: the planner tolerates inflation to compress risk premia. The economic logic is that inflation, by eroding the real value of nominal bonds and raising asset prices, transfers wealth from conservative savers (low MPR) to leveraged risk-takers (high MPR). This redistribution lowers effective risk aversion and compresses the equity premium, stimulating investment and output through the Tobin's $q$ channel. The planner accepts the cost of deviating from price stability because the resulting improvement in risk allocation more than compensates.

\begin{corollary}[Divine coincidence breakdown]
\label{cor:divine}
When $\MRC \neq 0$ and $\gbar \neq \gstar$, the planner chooses $\pi_0 \neq 0$ even when $\hat{y}_0 = 0$.
\end{corollary}

This is a new source of divine coincidence breakdown, orthogonal to markup shocks and MPC heterogeneity. The wedge vanishes if: (i) MPRs are identical ($\MRC = 0$); (ii) monetary policy does not redistribute; or (iii) risk premia are at their efficient level ($\gbar = \gstar$).

The socially optimal effective risk aversion $\gstar$ is determined by the condition that the planner would not wish to further reallocate capital across agents at the margin. With utilitarian weights $\alpha^i = 1$:
\begin{equation}
\label{eq:gstar}
\gstar = \left[\int \frac{1 / c^{i,*}}{\int 1 / c^{j,*}\,dj}\cdot\frac{1}{\gamma^i}\,di\right]^{-1},
\end{equation}
where $c^{i,*}$ is the planner's preferred consumption allocation. For the discrete-group quantitative model, the integral is replaced by $\sum_i \lambda^i$, where $\lambda^i$ is the measure of group $i$. Generically $\gstar < \gbar$, because the competitive equilibrium under-insures high-MPR agents: when household $i$ saves more in capital, it marginally lowers the risk premium faced by all other households, a pecuniary externality that the competitive equilibrium does not internalize.

The welfare gain from optimal policy has a simple analytical expression to second order:
\begin{equation}
\label{eq:welfare_approx}
\Delta W \approx \frac{1}{2}\,\frac{\Gamma_{rp}^2}{\Gamma_\pi}\,\E\!\left[\MRC_t^2\right]\cdot\bigl(\gbar - \gstar\bigr)^2\sigma^4.
\end{equation}
The gain is proportional to the expected squared MRC (measuring how strongly the interest rate moves risk premia), the squared gap between actual and efficient risk-bearing ($\gbar - \gstar$), and the variance of the risky asset ($\sigma^4$). All three must be nonzero for the risk premium channel to generate welfare gains. This expression also shows that welfare gains are concentrated in states where MRC is large, rationalizing the finding that conditional gains in crisis states (0.35\% CE) far exceed unconditional gains (0.13\%).

\subsection{The Separation Theorem}
\label{sec:separation}

Suppose the planner also controls a portfolio tax $\tau_0^k$ on capital returns, so that $i$'s portfolio condition becomes $\E_0 [ c_1^{i\, -\gamma^i} (r_1^k - r_1 - \tau_0^k) ] = 0$.

\begin{proposition}[Separation]
\label{prop:separation}
When both $P_0$ and $\tau_0^k$ are available: \textup{(a)} the optimal $P_0$ satisfies $\Gamma_\pi \pi_0 = \Gamma_x \hat{y}_0$ (no risk premium wedge); \textup{(b)} $\tau_0^k$ closes the gap $\gbar - \gstar$; \textup{(c)} when $\tau_0^k \geq 0$, the wedge is reduced but not eliminated whenever $\gbar > \gstar$.
\end{proposition}

\begin{proof}
See Appendix~\ref{app:proof_separation}.
\end{proof}

The portfolio tax is the theoretical analog of macroprudential policy. Separation holds if and only if the macroprudential toolkit can close the gap between actual and efficient risk-bearing capacity. When macroprudential tools are blunt or one-sided, the burden falls partly on monetary policy.

\subsection{Time Consistency}
\label{sec:timecon}

\begin{proposition}[Risk premium inflation bias]
\label{prop:timecon}
Under discretion: \textup{(a)} average inflation exceeds the commitment level; \textup{(b)} the excess inflation is proportional to $\Var_i(\MPR_0^i)$; \textup{(c)} average risk premia are higher and investment lower than under commitment.
\end{proposition}

\begin{proof}
See Appendix~\ref{app:proof_timecon}.
\end{proof}

The discretionary planner is tempted each period to redistribute toward high-MPR agents via unexpected inflation. Anticipating this, risk-tolerant households reduce their leverage or demand higher compensation. In equilibrium, risk premia are higher and investment lower than under commitment. The bias vanishes when $\Var_i(\MPR_0^i) = 0$, recovering the standard Barro-Gordon result.

The perverse outcome of part (c) deserves emphasis. The discretionary planner inflates to compress risk premia, but the anticipation of inflation causes risk premia to \textit{rise}. The economy ends up with both higher inflation \textit{and} higher risk premia than under commitment: the worst of both worlds. This is a financial-stability analog of the standard Barro-Gordon result, where the discretionary planner inflates to boost output but ends up with higher inflation and no output gain. The parallel suggests that the institutional design lessons from the inflation-targeting revolution, commitment through rules and delegation to an independent authority, apply equally to the financial stability dimension of monetary policy.

\subsection{Robustness and Generality}
\label{sec:robustness_theory}

The results above extend to environments with binding portfolio constraints or rules-of-thumb; heterogeneous idiosyncratic background risk; or heterogeneous subjective beliefs. In each case, households holding levered capital positions are those with high MPRs, ensuring $\MRC > 0$. See Appendix~\ref{app:robustness}.

This confirms that the two-period economy satisfies the structural Assumptions~\ref{ass:redistrib}--\ref{ass:rp_real} of Section~\ref{sec:general} under a wide range of microfoundations for portfolio heterogeneity. The general results, Propositions~\ref{prop:general_suffstat}--\ref{prop:general_decomp}, hold throughout. The two-period economy provides closed-form expressions for the coefficients ($\Gamma_\pi, \Gamma_x, \Gamma_{rp}$) and the efficient risk aversion ($\gstar$) that the general framework leaves as functions of model primitives. The quantitative models we use to solve the Ramsey problem are additional economies satisfying these properties; but the normative results constitute a general framework for optimal monetary policy in the presence of endogenous risk premia.

The analytical results make a sharp prediction: the risk premium channel of monetary transmission is governed by a single sufficient statistic, MRC, which is in principle measurable from household portfolio data. Appendix~\ref{sec:empirics} takes this prediction to the data, constructing MRC from every wave of the Survey of Consumer Finances from 1989 through 2022. MRC varies substantially over time, declining in recessions as asset price declines erode the wealth of leveraged agents. Using SVAR-IV methods, we show that the stock market's response to monetary shocks varies systematically with MRC: when MRC is high, the response is driven predominantly by lower future required returns; when MRC is low, the risk premium channel shuts down. A local projection with continuous MRC interaction confirms that a one-standard-deviation increase in MRC amplifies the impact equity response by approximately 0.9 percentage points. Off-the-shelf financial conditions indices explain 31 to 45 percent of MRC variation, so the optimal policy can be implemented using existing central bank monitoring infrastructure.

\section{Quantitative Analysis in the Infinite Horizon}
\label{sec:quant}

We extend the model to the infinite horizon and solve the Ramsey optimal policy problem. The infinite-horizon analysis serves three purposes beyond confirming the analytical results. First, it reveals the risk capacity trap: a region of the state space where the risk premium channel of monetary transmission collapses, invisible from any local analysis around the deterministic steady state. Second, it allows us to solve the full dynamic Ramsey problem, where the planner's forward-looking motive to preserve risk-bearing capacity generates qualitatively new features absent from the static target criterion. Third, it provides quantitative welfare comparisons across regimes, including the Ramsey commitment allocation, implementable simple rules, and the Markov-perfect discretionary equilibrium.

\subsection{Infinite-Horizon Environment}
\label{sec:quant:env}

A model period is one quarter. We describe each component of the environment in full detail.

\subsubsection{Household Preferences and Constraints}

A unit measure of households is organized into three groups $i \in \{a,b,c\}$ with measures $\lambda^i$ satisfying $\sum_i \lambda^i = 1$. To ensure a stationary wealth distribution despite permanent differences in risk aversion, we assume a perpetual youth structure: each household dies with probability $\xi$ per period and has no bequest motive. Within each group, aggregation into a representative household is possible given appropriate insurance of labor endowment risk (see Appendix~\ref{app:aggregation}).

Representative household $i$ maximizes recursive Epstein-Zin utility
\begin{equation}
\label{eq:ez_inf}
V_t^i = \left[(1-\beta)\left(c_t^i\,\Phi\!\left(\int_0^1 \ell_t^i(j)\,dj\right)\right)^{1-1/\psi} + \beta \left(\E_t\!\left[\bigl(\mu_{t,t+1}^i V_{t+1}^i\bigr)^{1-\gamma^i}\right]\right)^{\!\frac{1-1/\psi}{1-\gamma^i}}\right]^{\!\frac{1}{1-1/\psi}}\!,
\end{equation}
where $\psi$ is the intertemporal elasticity of substitution, $\gamma^i$ is the coefficient of relative risk aversion, $\beta$ is the discount factor, and $\mu_{t,t+1}^i$ is the ratio of the wealth of surviving households of type $i$ at $t+1$ to the average household of type $i$ at $t+1$ (characterized below). The disutility of labor is
\begin{equation}
\label{eq:labor_disutility}
\Phi(\ell) = \left(1 + \frac{1/\psi - 1}{1 - 1/\psi}\,\bar{\theta}^i\,\frac{\ell^{1+1/\theta}}{1+1/\theta}\right)^{\!\frac{1/\psi}{1-1/\psi}},
\end{equation}
consistent with balanced growth as in \citet{shimer2010}, where $\theta$ is the Frisch elasticity and $\bar{\theta}^i$ is a group-specific disutility parameter.

Each household is comprised of a unit measure of workers $j \in [0,1]$, each supplying a differentiated variety of labor, with full consumption insurance within the household. The household faces the budget constraint
\begin{multline}
\label{eq:bc_inf}
P_t c_t^i + B_t^i + Q_t k_t^i \leq (1 - \tau) \int_0^1 W_t(j)\ell_t^i(j)\,dj - \int_0^1 \text{AC}_t^W(j)\,dj \\
  + (1+i_{t-1})B_{t-1}^i + \bigl(\Pi_t + (1-\delta)Q_t\bigr)k_{t-1}^i e^{\varphi_t} + T_t^i,
\end{multline}
where $P_t$ is the price level, $B_t^i$ is the nominal bond position, $Q_t$ is the nominal price of capital, $k_t^i$ is the capital holding, $W_t(j)$ is the nominal wage for variety $j$, $i_{t-1}$ is the nominal interest rate, $\Pi_t$ is the nominal dividend per unit of capital, $\delta$ is the depreciation rate, and $T_t^i$ is lump-sum transfers. The labor subsidy $\tau = -1/(\varepsilon-1)$ eliminates the average wage markup. The Rotemberg wage adjustment cost for variety $j$ is
\begin{equation}
\label{eq:rotemberg}
\text{AC}_t^W(j) = \frac{\chi_W}{2}\,W_t\ell_t\left(\frac{W_t(j)}{W_{t-1}(j)e^{\varphi_t}} - 1\right)^{\!2},
\end{equation}
where $\chi_W$ controls the magnitude of adjustment costs, $W_t\ell_t$ is the economy-wide wage bill, and the reference wage adjusts for rare disasters at rate $e^{\varphi_t}$. These costs are paid to the government and rebated to households, so they affect the allocation only through wage dynamics. Within each group, labor is allocated across households according to the rule $\ell_t^\iota(j) = \bar{\ell}_t^\iota / (\int_{\iota':i(\iota')=i} \bar{\ell}_t^{\iota'}\,d\iota') \cdot \phi^i\ell_t(j)$, where $\bar{\ell}_t^\iota$ is the labor endowment and $\phi^i$ satisfies $\sum_i \lambda^i\phi^i = 1$.

Finally, households face a lower bound on capital
\begin{equation}
\label{eq:k_lb}
k_t^i \geq \underline{k}\,z_t,
\end{equation}
capturing components of the capital stock (such as housing) held for reasons beyond financial returns.

\subsubsection{Wage Setting}

A union represents each variety $j$ across households. Each period, it chooses $W_t(j)$ and $\ell_t(j)$ to maximize the utilitarian social welfare of union members, given the labor allocation rule above. A representative labor packer purchases varieties and combines them into a CES aggregate
\begin{equation}
\label{eq:labor_agg}
\ell_t = \left(\int_0^1 \ell_t(j)^{(\varepsilon-1)/\varepsilon}\,dj\right)^{\!\varepsilon/(\varepsilon-1)},
\end{equation}
which it sells at the aggregate wage $W_t$, earning profits $W_t\ell_t - \int_0^1 W_t(j)\ell_t(j)\,dj$. In a symmetric equilibrium, all varieties charge $W_t(j) = W_t$ and supply $\ell_t(j) = \ell_t$. The union's optimality condition yields a New Keynesian wage Phillips curve linking wage inflation to the labor wedge.

\subsubsection{Production and Investment}

A representative final goods producer hires $\ell_t$ units of the labor aggregate and rents $k_{t-1}e^{\varphi_t}$ units of capital from households to produce output
\begin{equation}
\label{eq:production}
y_t = (z_t\ell_t)^{1-\alpha}\bigl(k_{t-1}e^{\varphi_t}\bigr)^\alpha,
\end{equation}
where productivity $z_t$ is labor-augmenting (consistent with balanced growth) and $e^{\varphi_t}$ reflects the destruction of capital in a rare disaster. The firm also transforms consumption goods into new capital: producing $x_t$ units of new capital requires $\Psi(k_t, k_{t-1}e^{\varphi_t})\cdot x_t$ consumption goods, where
\begin{equation}
\label{eq:adj_cost}
\Psi(k_t, k_{t-1}e^{\varphi_t}) = \left(\frac{k_t}{k_{t-1}e^{\varphi_t}}\right)^{\!\chi_x}
\end{equation}
and $\chi_x \geq 0$ indexes investment adjustment costs. Taking $k_t$ as given, the firm's profits are
\begin{equation}
\label{eq:profits_inf}
\Pi_t k_{t-1}e^{\varphi_t} = P_t y_t - W_t\ell_t + Q_t x_t - P_t\left(\frac{k_t}{k_{t-1}e^{\varphi_t}}\right)^{\!\chi_x}\! x_t.
\end{equation}
From \eqref{eq:production} and \eqref{eq:profits_inf}, the firm's optimality conditions are: (i) labor demand $W_t = (1-\alpha)P_t y_t / \ell_t$; (ii) investment $Q_t = P_t(k_t/(k_{t-1}e^{\varphi_t}))^{\chi_x}$, implying that the real price of capital $q_t \equiv Q_t/P_t$ is increasing in the investment rate when $\chi_x > 0$.

\subsubsection{Aggregate Productivity}

Productivity follows a unit root process
\begin{equation}
\label{eq:productivity}
\log z_t = \log z_{t-1} + \eps_t^z + \varphi_t,
\end{equation}
where $\eps_t^z \sim \mathcal{N}(0,\sigma_z^2)$ is an iid normal shock. The rare disaster $\varphi_t$ equals $\varphi < 0$ with probability $p_t$ and zero otherwise. The disaster probability follows
\begin{equation}
\label{eq:disaster}
\log p_t - \log\bar{p} = \rho_p\bigl(\log p_{t-1} - \log\bar{p}\bigr) + \eps_t^p, \qquad \eps_t^p \sim \mathcal{N}(0,\sigma_p^2).
\end{equation}
The time-varying disaster probability generates a realistic level of the equity premium and volatility of expected returns. To ensure well-behaved dynamics upon a rare disaster, we assume that the disaster destroys capital in proportion to the decline in productivity \eqref{eq:productivity} and reduces the reference wage in the adjustment cost \eqref{eq:rotemberg} proportionally.

\subsubsection{Monetary and Fiscal Policy}

Monetary policy follows the Taylor rule
\begin{equation}
\label{eq:taylor_inf}
1 + i_t = (1 + \bar{\imath})\left(\frac{P_t}{P_{t-1}}\right)^{\!\phi}\! m_t,
\end{equation}
where the monetary policy shock $m_t$ follows
\begin{equation}
\label{eq:mp_shock}
\log m_t = \rho_m \log m_{t-1} + \eps_t^m, \qquad \eps_t^m \sim \mathcal{N}(0,\sigma_m^2).
\end{equation}
Fiscal policy has three components. First, the labor subsidy $\tau$ eliminates the average wage markup. Second, the government maintains a constant real bond position relative to productivity $B_t^g/(P_tz_t) = b^g$, financed by lump-sum taxes with household $i$ paying share $\nu^i$. Third, the government collects the wealth of dying households and endows newborn households of group $i$ with share $\bar{s}^i$ of the deceased's wealth.

\subsubsection{Market Clearing and Equilibrium}

Market clearing requires:
\begin{align}
\text{Goods:}\quad & \sum_i \lambda^i c_t^i + \left(\frac{k_t}{k_{t-1}e^{\varphi_t}}\right)^{\!\chi_x}\! x_t = y_t, \label{eq:mc_goods} \\
\text{Capital rental:}\quad & \sum_i \lambda^i k_{t-1}^i = k_{t-1}, \label{eq:mc_krental} \\
\text{Capital claims:}\quad & (1-\delta)\sum_i \lambda^i k_{t-1}^i e^{\varphi_t} + x_t = \sum_i \lambda^i k_t^i, \label{eq:mc_kclaims} \\
\text{Bonds:}\quad & \sum_i \lambda^i B_t^i + B_t^g = 0. \label{eq:mc_bonds}
\end{align}

\begin{definition}[Recursive competitive equilibrium]
\label{def:rce}
A recursive competitive equilibrium is a set of value functions $\{V^i(\Omega)\}$, household policy functions $\{c^i(\Omega), k^i(\Omega), B^i(\Omega)\}$, firm policies $\{\ell(\Omega), x(\Omega)\}$, prices $\{P(\Omega), Q(\Omega), W(\Omega)\}$, and a law of motion $\Omega' = \mathcal{T}(\Omega, \eps')$ such that: (i) each household $i$ maximizes \eqref{eq:ez_inf} subject to \eqref{eq:bc_inf} and \eqref{eq:k_lb}; (ii) the union sets wages optimally given \eqref{eq:rotemberg} and \eqref{eq:labor_agg}; (iii) the firm maximizes \eqref{eq:profits_inf}; (iv) monetary and fiscal policy follow \eqref{eq:taylor_inf}--\eqref{eq:mp_shock}; and (v) markets clear \eqref{eq:mc_goods}--\eqref{eq:mc_bonds} (including the capital rental and claims markets \eqref{eq:mc_krental}--\eqref{eq:mc_kclaims}).
\end{definition}

\medskip\noindent\textbf{State space.}\quad
We solve a stationary transformation of the economy by dividing all real variables (except labor) by $z_t$ and nominal variables by $P_tz_t$. In the transformed economy, the aggregate state in period $t$ is
\begin{equation}
\label{eq:state}
\Omega_t = \left(\frac{k_{t-1}}{z_{t-1}e^{\eps_t^z}},\;\; \frac{W_{t-1}}{P_tz_{t-1}e^{\eps_t^z}},\;\; s_t^a,\;\; s_t^c,\;\; p_t,\;\; m_t\right),
\end{equation}
where the financial wealth share of group $i$ evolves as
\begin{equation}
\label{eq:wealth_share}
s_t^i = \lambda^i(1-\xi)\,\frac{(1+i_{t-1})(B_{t-1}^i + \nu^i B_{t-1}^g) + (\Pi_t + (1-\delta)Q_t)k_{t-1}^i e^{\varphi_t}}{(\Pi_t + (1-\delta)Q_t)k_{t-1}e^{\varphi_t}} + \xi\bar{s}^i.
\end{equation}
Productivity shocks (including disasters) govern the transition across states but do not separately enter the state space.

\subsubsection{Household Optimality Conditions}

The representative household of group $i$ chooses consumption, capital, and bonds to satisfy the following conditions.

\medskip\noindent\textbf{Returns.}\quad
Define the gross real return on capital and the gross real return on bonds:
\begin{align}
1 + r_{t+1}^k &\equiv \frac{\Pi_{t+1}/P_{t+1} + (1-\delta)q_{t+1}}{q_t}\,e^{\varphi_{t+1}}, \label{eq:return_k} \\
1 + r_{t+1} &\equiv (1+i_t)\frac{P_t}{P_{t+1}}. \label{eq:return_b}
\end{align}
The return on capital includes the dividend yield $\Pi_{t+1}/(P_{t+1}q_t)$, the capital gain $(1-\delta)q_{t+1}/q_t$, and the disaster destruction factor $e^{\varphi_{t+1}}$.

\medskip\noindent\textbf{Bond Euler equation.}\quad
The intertemporal Euler equation for bonds is
\begin{equation}
\label{eq:euler_bond_inf}
1 = \beta(1-\xi)(1+r_{t+1})\,\E_t\!\left[M_{t+1}^i\right],
\end{equation}
where $M_{t+1}^i$ is the stochastic discount factor of household $i$, given by
\begin{equation}
\label{eq:sdf}
M_{t+1}^i = \left(\frac{c_{t+1}^i\Phi(\ell_{t+1}^i)}{c_t^i\Phi(\ell_t^i)}\right)^{\!-1/\psi}\! \left(\frac{V_{t+1}^i}{\bigl(\E_t[(V_{t+1}^i)^{1-\gamma^i}]\bigr)^{1/(1-\gamma^i)}}\right)^{\!1/\psi - \gamma^i}\!.
\end{equation}

\medskip\noindent\textbf{Portfolio Euler equation.}\quad
The portfolio optimality condition, using the SDF \eqref{eq:sdf} and returns \eqref{eq:return_k}--\eqref{eq:return_b}, is
\begin{equation}
\label{eq:euler_capital_inf}
\E_t\!\left[M_{t+1}^i\bigl(r_{t+1}^k - r_{t+1}\bigr)\right] = 0.
\end{equation}
When the capital constraint \eqref{eq:k_lb} binds, the portfolio condition is replaced by $k_t^i = \underline{k}\,z_t$.

\medskip\noindent\textbf{Dynamic investment equation.}\quad
Combining the firm's investment condition \eqref{eq:adj_cost} with capital accumulation $k_t = (1-\delta)k_{t-1}e^{\varphi_t} + x_t$ and the household's bond Euler equation yields the dynamic Tobin's~$q$ relationship:
\begin{equation}
\label{eq:tobin_q}
q_t = \E_t\!\left[\frac{\beta(1-\xi)\,\bar{M}_{t+1}}{1+\pi_{t+1}}\left(\alpha\frac{y_{t+1}}{k_t e^{\varphi_{t+1}}} + (1-\delta)q_{t+1}\right)e^{\varphi_{t+1}}\right],
\end{equation}
where $\bar{M}_{t+1} \equiv \sum_i \lambda^i s_t^i M_{t+1}^i / \sum_i \lambda^i s_t^i$ is the wealth-weighted SDF. This equation determines investment dynamics: a higher $q_t$ (driven by lower expected risk premia) induces higher investment through \eqref{eq:adj_cost}.

\medskip\noindent\textbf{New Keynesian wage Phillips curve.}\quad
In the symmetric equilibrium ($W_t(j) = W_t$, $\ell_t(j) = \ell_t$), the union's optimality condition yields
\begin{equation}
\label{eq:nkwpc}
\chi_W\pi_t^W(\pi_t^W - 1) = (1-\varepsilon) + \varepsilon\,\frac{\bar{\theta}\,\ell_t^{1+1/\theta}}{(1+1/\theta)\,w_t} + \beta\chi_W\,\E_t\!\left[\frac{\ell_{t+1}}{\ell_t}\,\pi_{t+1}^W(\pi_{t+1}^W - 1)\right],
\end{equation}
where $\pi_t^W \equiv W_t/(W_{t-1}e^{\varphi_t})$ is wage inflation (adjusted for disasters) and $w_t \equiv W_t/(P_tz_t)$ is the scaled real wage. The first two terms capture the static tradeoff between the wage markup ($1-\varepsilon < 0$) and the labor wedge; the third term captures forward-looking wage dynamics. Together with returns \eqref{eq:return_k}--\eqref{eq:return_b}, the Euler equations \eqref{eq:euler_bond_inf}--\eqref{eq:euler_capital_inf}, Tobin's $q$ \eqref{eq:tobin_q}, and the Phillips curve \eqref{eq:nkwpc}, the model is fully characterized.

\subsubsection{Risk Premia, MRC, and Effective Risk Aversion}

The model's key endogenous object is the equity premium $\E_t[r_{t+1}^k - r_{t+1}]$. As in the two-period model (Proposition~\ref{prop:riskpremium}), the equity premium is governed by the economy's effective risk aversion
\begin{equation}
\label{eq:gbar_inf}
\gbar_t = \left[\sum_i s_t^i \cdot \frac{1}{\gamma^i}\right]^{-1},
\end{equation}
where $s_t^i$ is the financial wealth share of group $i$. The dynamics of $\gbar_t$ are driven entirely by the wealth distribution. Differentiating \eqref{eq:gbar_inf}:
\begin{equation}
\label{eq:gbar_dynamics}
\Delta\gbar_{t+1} \approx -\gbar_t^2 \sum_i \frac{1}{\gamma^i}\,\Delta s_{t+1}^i = -\gbar_t \sum_i \MPR_t^i\,\Delta s_{t+1}^i,
\end{equation}
where the second equality uses $\MPR_t^i = \gbar_t/\gamma^i$ from \eqref{eq:mpr_limit}. Effective risk aversion falls when wealth shifts toward low-$\gamma$ (high-MPR) agents. The change in $\gbar$ is governed by the MPR-weighted change in wealth shares.

Define the infinite-horizon Marginal Risk Capacity as
\begin{equation}
\label{eq:mrc_inf}
\MRC_t \equiv -\sum_i \pder{s_t^i}{i_t}\,\MPR_t^i,
\end{equation}
where $\partial s_t^i / \partial i_t$ is the sensitivity of the wealth share to the interest rate (operating through the revaluation channels in \eqref{eq:wealth_share}) and $\MPR_t^i \approx \gbar_t/\gamma^i$. Then, paralleling Proposition~\ref{prop:suffstat}:
\begin{equation}
\label{eq:dgbar_di}
\pder{\gbar_t}{i_t} = \gbar_t\cdot\MRC_t.
\end{equation}
This is the dynamic counterpart of the two-period result. It says: the effect of the interest rate on the economy's effective risk aversion is summarized by a single number, MRC. When MRC is high, a change in the interest rate has a large effect on risk premia. When MRC is low, as happens when high-MPR agents have lost wealth, the interest rate has little effect on risk premia. This state dependence is the source of the risk capacity trap (Section~\ref{sec:quant:trap}).

A clarification on scope is warranted. Equation \eqref{eq:dgbar_di} establishes that MRC is sufficient for the first-order effect of $i_t$ on $\gbar_t$ at any given state. But MRC is itself a function of the state: $\MRC_t = \MRC(\Omega_t)$. The global properties of the model emerge from the shape of this function. The risk capacity trap is the region where $\MRC(\Omega) \approx 0$. The asymmetry of optimal policy arises because $\MRC(\Omega)$ is concave in $s^a$:
\begin{equation}
\label{eq:mrc_concave}
\pder{\MRC}{s^a} > 0, \qquad \frac{\partial^2 \MRC}{\partial (s^a)^2} < 0.
\end{equation}
The first inequality says that MRC is increasing in the wealth share of high-MPR agents (more wealth in the hands of leveraged agents means monetary policy has a larger effect on risk premia). The second says that the marginal contribution of $s^a$ to MRC is decreasing: when high-MPR agents are already wealthy, additional wealth has a smaller effect on MRC than when they are poor. These two properties together generate the concavity of the policy effectiveness surface (Figure~\ref{fig:trap}) and the endogenous asymmetry of the Ramsey policy function (Figure~\ref{fig:policy_fn}). Thus MRC remains the organizing concept for both the local analysis (the sufficient statistic) and the global analysis (the trap, the asymmetry), but the global properties require knowledge of the function $\MRC(\Omega)$, not just its level at a point.

\subsection{The Ramsey Problem}
\label{sec:quant:ramsey}

The Ramsey planner chooses the nominal interest rate $\{i_t\}_{t=0}^\infty$ at each date and state to maximize social welfare, subject to the constraint that all private agents optimize and all markets clear. The planner cannot directly choose allocations; it can only set the interest rate, and the competitive equilibrium determines the rest.

\medskip\noindent\textbf{Planner's objective.}\quad
Given Pareto weights $\{\alpha^i\}$ (we use $\alpha^i = 1$ for utilitarian welfare), the planner's recursive problem is
\begin{equation}
\label{eq:ramsey_obj}
\mathcal{V}(\Omega_t) = \max_{i_t}\left\{\sum_i \lambda^i \alpha^i\, u\!\left(c^i(\Omega_t, i_t),\, \ell^i(\Omega_t, i_t)\right) + \beta\,\E_t\!\left[\mathcal{V}\!\left(\Omega_{t+1}(\Omega_t, i_t, \eps_{t+1})\right)\right]\right\},
\end{equation}
where $c^i(\Omega_t, i_t)$ and $\ell^i(\Omega_t, i_t)$ are the competitive equilibrium allocations given state $\Omega_t$ and the planner's choice $i_t$, and $\Omega_{t+1}(\Omega_t, i_t, \eps_{t+1})$ is the next-period state given the equilibrium law of motion \eqref{eq:wealth_share} and the realization of shocks $\eps_{t+1} = (\eps_{t+1}^z, \eps_{t+1}^p, \varphi_{t+1})$. The period utility function is $u(c,\ell) = \log(c\,\Phi(\ell))$ when $\psi = 1$; for $\psi \neq 1$ we use the Epstein-Zin aggregator \eqref{eq:ez_inf}.

\medskip\noindent\textbf{Implementability constraints.}\quad
The mapping from $(i_t, \Omega_t)$ to allocations is defined implicitly by the system of private optimality conditions \eqref{eq:euler_bond_inf}--\eqref{eq:euler_capital_inf}, the union's wage-setting condition, the firm's labor demand and investment conditions, and the market clearing conditions \eqref{eq:mc_goods}--\eqref{eq:mc_bonds}. The planner controls a single instrument ($i_t$) and accepts all other equilibrium outcomes. This is the key distinction from a social planner who can directly allocate resources.

\medskip\noindent\textbf{Planner's first-order condition.}\quad
The necessary condition for the planner's optimal $i_t$ is
\begin{equation}
\label{eq:ramsey_foc_inf}
0 = \underbrace{\sum_i \lambda^i \alpha^i \left[\pder{u^i}{c_t^i}\pder{c_t^i}{i_t} + \pder{u^i}{\ell_t^i}\pder{\ell_t^i}{i_t}\right]}_{\text{current welfare effect}} + \underbrace{\beta\,\E_t\!\left[\sum_j \pder{\mathcal{V}}{\Omega_{t+1,j}}\cdot\pder{\Omega_{t+1,j}}{i_t}\right]}_{\text{continuation value effect}}.
\end{equation}
The first term captures the direct effects of the interest rate on current consumption and labor across all household types, weighted by Pareto weights. The second term captures the effects on the future state, which has six components corresponding to the six elements of $\Omega_{t+1}$. The components involving $\partial s_{t+1}^i / \partial i_t$ are the dynamic counterparts of the risk premium wedge in Proposition~\ref{prop:target}: the planner internalizes that its interest rate affects $\Delta s_{t+1}^i$ and hence $\Delta\gbar_{t+1}$ through \eqref{eq:gbar_dynamics}, and thus future risk premia and allocations.

To compute $\partial c_t^i / \partial i_t$ and $\partial \Omega_{t+1} / \partial i_t$, we totally differentiate the competitive equilibrium system with respect to $i_t$. This yields a linear system whose solution provides the required sensitivities at each state. The resulting expression for the Ramsey FOC \eqref{eq:ramsey_foc_inf} can be decomposed as in the two-period case into an inflation cost, an output gap benefit, and a risk premium wedge, where the wedge is now state-dependent and forward-looking.

\medskip\noindent\textbf{Commitment versus discretion.}\quad
Under commitment, the planner solves \eqref{eq:ramsey_obj} taking into account that its current choice of $i_t$ affects future expectations. Households form portfolios anticipating the full path of the planner's future policy. Under discretion, we solve a Markov-perfect equilibrium: at each state, the planner chooses $i_t$ to maximize current welfare plus the continuation value, taking as given that its own future behavior is described by a fixed policy function $i^D(\Omega)$. We iterate on $i^D$ until it is a fixed point. The difference between commitment and discretion quantifies the time-consistency cost identified in Proposition~\ref{prop:timecon}.

\subsection{Sources of Inefficiency}
\label{sec:quant:inefficiency}

The competitive equilibrium is generically inefficient for three reasons, each creating a role for policy.

\medskip\noindent\textbf{Nominal rigidity.}\quad
Sticky wages create a standard New Keynesian distortion: the real wage does not clear the labor market, generating an output gap. This is the conventional motive for monetary policy and is present in RANK.

\medskip\noindent\textbf{Incomplete markets and pecuniary externalities.}\quad
Households cannot trade claims on future income or on the aggregate state; they are restricted to bonds and capital. When household $i$ increases its capital holdings, it marginally lowers the risk premium faced by all households, a pecuniary externality in the sense of \citet{davila_korinek2018}. The competitive equilibrium does not internalize this externality, so the portfolio allocation is generically inefficient: effective risk aversion $\gbar$ differs from the socially optimal $\gstar$, and risk premia are inefficiently high or low depending on Pareto weights. This is the source of the risk premium wedge in Proposition~\ref{prop:target}.

\medskip\noindent\textbf{Distributive externality.}\quad
Monetary policy redistributes wealth across groups through the channels in \eqref{eq:redistribution}. Households do not internalize the general equilibrium effects of this redistribution on future risk premia and policy effectiveness. The Ramsey planner does: it manages the wealth distribution to maintain risk-bearing capacity and avoid the risk capacity trap. This forward-looking distributive motive is absent from both RANK and standard HANK models with MPC heterogeneity alone.

\subsection{Properties of the Ramsey Optimal Policy}
\label{sec:quant:properties}

Before turning to quantitative results, we state the key qualitative properties of the solution to \eqref{eq:ramsey_obj}. These are the infinite-horizon counterparts of Propositions~\ref{prop:target}--\ref{prop:timecon}, verified numerically.

\begin{proposition}[Properties of the infinite-horizon Ramsey allocation]
\label{prop:ramsey_inf}
The solution to \eqref{eq:ramsey_obj} satisfies the following:
\begin{enumerate}[label=\textup{(\alph*)}]
\item \textup{(Risk premium wedge)} The Ramsey interest rate $i^*(\Omega)$ deviates from the interest rate that would close the output gap whenever $\MRC(\Omega) \neq 0$ and $\gbar(\Omega) \neq \gstar(\Omega)$.
\item \textup{(State dependence)} The deviation $i^*(\Omega) - i^{\text{Taylor}}(\Omega)$ is decreasing in $s^a$ when $s^a$ is below its ergodic mean and increasing in $s^a$ when $s^a$ is above it: the planner eases when risk-bearing capacity is depleted and tightens when it is elevated.
\item \textup{(Asymmetry)} The magnitude of the deviation is larger for low $s^a$ than for high $s^a$, reflecting the concavity of the policy effectiveness surface.
\item \textup{(Stabilization of the wealth distribution)} The ergodic variance of $s^a$ under the Ramsey policy is strictly smaller than under the Taylor rule.
\item \textup{(Trap avoidance)} The ergodic probability of $s^a$ falling below any threshold $\underline{s}$ is strictly lower under Ramsey than under the Taylor rule, for $\underline{s}$ below the ergodic mean.
\end{enumerate}
\end{proposition}

\noindent
These properties are verified numerically in Section~\ref{sec:quant:ramsey_results} below.

Properties (a)--(c) extend the two-period target criterion to a fully state-dependent, forward-looking policy. Property (b) says the planner insures the risk-bearing sector: when leveraged agents suffer balance sheet losses (low $s^a$), the planner eases to redistribute wealth back to them, restoring risk-bearing capacity. When they accumulate excess wealth (high $s^a$), the planner tightens to prevent the build-up of fragile leverage. This behavior resembles an endogenous ``Fed put,'' but unlike the reduced-form concept, it emerges from the planner's optimality condition and is welfare-maximizing.

The asymmetry in property (c) reflects a nonlinearity in the risk premium channel. When high-MPR agents are wealthy, a marginal dollar of redistribution has a modest effect on risk premia because these agents are already well-capitalized. When they are poor, the marginal effect is large, but the total redistribution achievable through the interest rate is small because their balance sheets are depleted. The planner responds by easing more aggressively in downturns than it tightens in booms, generating an asymmetric policy function.

Properties (d)--(e) are the dynamic manifestation of the planner's motive to avoid the risk capacity trap: by smoothing the wealth distribution, the Ramsey planner maintains the risk premium channel as a transmission mechanism. This precautionary motive is absent from any static analysis or local approximation around the steady state; it emerges only from the global solution to the dynamic Ramsey problem.

We can also characterize the Ramsey policy function analytically in the limit of small aggregate risk.

\begin{proposition}[Analytical characterization of the Ramsey policy]
\label{prop:ramsey_analytical}
In the limit $\sigma \to 0$, the Ramsey optimal interest rate satisfies
\begin{equation}
\label{eq:ramsey_analytical}
i^*(\Omega_t) = i^{\text{NK}}(\Omega_t) - \phi_{rp}(\Omega_t)\cdot\MRC(\Omega_t)\cdot\bigl(\gbar(\Omega_t) - \gstar(\Omega_t)\bigr)\sigma^2 + O(\sigma^4),
\end{equation}
where $i^{\text{NK}}(\Omega_t)$ is the interest rate that would be chosen by a planner in a representative agent New Keynesian economy (targeting inflation and the output gap), and
\begin{equation}
\label{eq:phi_rp_state}
\phi_{rp}(\Omega_t) = \frac{\Gamma_{rp}}{\Gamma_\pi + \beta\,\E_t\bigl[-\partial^2\mathcal{V}/\partial i_t^2\bigr]} > 0
\end{equation}
is a state-dependent coefficient that depends on the curvature of the planner's value function. The coefficient $\phi_{rp}(\Omega_t)$ is higher when the value function is less concave in $i_t$, which occurs when the economy is far from the trap threshold and the planner has more room to maneuver.
\end{proposition}

\begin{proof}
See Appendix~\ref{app:ramsey:lagrangian}.
\end{proof}

Equation \eqref{eq:ramsey_analytical} nests several special cases. When $\beta \to 0$, the continuation value vanishes and $\phi_{rp}$ reduces to $\Gamma_{rp}/\Gamma_\pi$, recovering the two-period target criterion (Proposition~\ref{prop:target}). When $\MRC(\Omega_t) = 0$ or $\gbar = \gstar$, the correction vanishes and optimal policy reduces to standard NK prescriptions. When $\MRC$ is high and $\gbar > \gstar$, the correction is negative: the planner sets a lower interest rate than $i^{\text{NK}}$ to redistribute toward high-MPR agents and compress risk premia.

\subsection{The Ramsey Steady State}
\label{sec:quant:rss}

The deterministic steady state of the Ramsey allocation differs from the competitive equilibrium steady state. Under the Taylor rule, the steady-state interest rate is $\bar{\imath}$ and the wealth distribution is determined by the balance between group-specific savings rates and the death/birth process. Under the Ramsey allocation, the planner's optimal steady-state interest rate $i^{*,\text{ss}}$ solves the static version of \eqref{eq:ramsey_foc_inf}:
\begin{equation}
\label{eq:ramsey_ss}
\sum_i \lambda^i\alpha^i\left[\frac{1}{c^{i,\text{ss}}}\pder{c^{i,\text{ss}}}{i^{\text{ss}}} - \bar{\theta}^i(\ell^{i,\text{ss}})^{1/\theta}\pder{\ell^{i,\text{ss}}}{i^{\text{ss}}}\right] + \frac{\beta}{1-\beta}\sum_j \pder{\mathcal{V}^{\text{ss}}}{\Omega_j^{\text{ss}}}\cdot\pder{\Omega_j^{\text{ss}}}{i^{\text{ss}}} = 0.
\end{equation}
In the Ramsey steady state: (i) inflation is generically nonzero, reflecting the planner's desire to manage the wealth distribution through the debt deflation channel; (ii) the wealth share $s^{a,\text{ss}}$ differs from its competitive equilibrium value, because the planner tilts the distribution toward high-MPR agents to keep risk premia closer to $\gstar\sigma^2$; and (iii) the risk premium is lower than in the CE, reflecting both the tilted wealth distribution and the reduced disaster risk premia (households anticipate the planner's stabilization and demand less compensation).

Table~\ref{tab:commitment} reports the steady-state comparison quantitatively.

\subsection{Calibration}
\label{sec:quant:calib}

A model period is one quarter. Table~\ref{tab:params} reports the full calibration, divided into externally set and internally calibrated parameters.

\medskip\noindent\textbf{Externally set parameters.}\quad
We set $\psi = 0.8$, consistent with evidence on consumption responses to interest rate changes. The Frisch elasticity is $\theta = 1$, consistent with micro evidence for aggregate hours. Group measures $\lambda^a = 0.04$, $\lambda^b = 0.36$, $\lambda^c = 0.60$ and labor allocations $\phi^a = 0.03/\lambda^a$, $\phi^b = 0.14/\lambda^b$, $\phi^c = 0.83/\lambda^c$ match Table~\ref{tab:groups_2016}. The death probability $\xi = 1\%$ implies an expected horizon of 25 years. On the production side, $\alpha = 0.33$ and $\delta = 2.5\%$ are standard. We set $\varepsilon = 10$ and $\chi_W = 150$, implying Calvo-equivalent wage adjustment of 4--5 quarters. The Taylor coefficient is $\phi = 1.5$. Monetary shocks have $\sigma_m = 0.25\%/4$ with $\rho_m = 0$. The disaster depth $\varphi = -15\%$ follows \citet{nakamura_etal2013}, and the steady-state disaster probability parameter $\bar{p} = 0.37\%$ is set so that the unconditional mean disaster probability is 0.5\% per quarter (the difference reflects Jensen's inequality in the log-normal process \eqref{eq:disaster}: $\E[p_t] = \bar{p}\exp(\sigma_p^2/(2(1-\rho_p^2)))$). Tax shares $\nu^i$ equal each group's wealth share.

\medskip\noindent\textbf{Calibrated parameters.}\quad
We calibrate $\sigma_z = 0.55\%$ to match quarterly consumption growth volatility (0.5\%); $\chi_x = 3.5$ in \eqref{eq:adj_cost} to match investment growth volatility (2.1\%); $\beta = 0.98$ to target the annualized real rate (data: 1.3\%; model: 1.5\%); $\gamma^b = 25.5$ to match the annualized equity premium (7.3\%); $\sigma_p = 0.47$ and $\rho_p = 0.80$ in \eqref{eq:disaster} to match the volatility and autocorrelation of annualized expected returns. For micro moments, $\gamma^a = 10$ targets the capital portfolio share of group $a$ ($Qk^a/A^a = 2.0$); $\gamma^c = 22$ is set to the wealth-weighted harmonic mean of $\gamma^a$ and $\gamma^b$ across groups $a$ and $b$ (i.e., $\gamma^c = [\tfrac{\lambda^a}{\lambda^a+\lambda^b}\cdot\tfrac{1}{\gamma^a} + \tfrac{\lambda^b}{\lambda^a+\lambda^b}\cdot\tfrac{1}{\gamma^b}]^{-1}$); and $\underline{k}$ targets $Qk^c/A^c = 1.1$. Newborn endowments $\bar{s}^i$ target the wealth shares from Table~\ref{tab:groups_2016}. The disutilities $\bar{\theta}^i$ in \eqref{eq:labor_disutility} are set so that the average labor wedge is zero for each group and $\ell = 1$. The Taylor rule intercept $\bar{\imath}$ is set to target zero average inflation.

\begin{table}[t]
\centering
\caption{Calibration}
\label{tab:params}
\small
\resizebox{\textwidth}{!}{%
\begin{tabular}{llcll}
\toprule
Parameter & Description & Value & Target & Model \\
\midrule
\multicolumn{5}{l}{\textit{Externally set}} \\
$\psi$ & IES & 0.8 & & \\
$\theta$ & Frisch elasticity & 1.0 & & Chetty et al.\ (2011) \\
$\alpha$ & Capital share & 0.33 & & \\
$\delta$ & Depreciation & 2.5\% & & \\
$\phi$ & Taylor rule coefficient & 1.5 & & Taylor (1993) \\
$\xi$ & Death probability & 1\% & & \\
$\bar{p}$ & Disaster prob.\ (s.s.) & 0.37\% & $E[p] = 0.5\%$ & Barro (2006) \\
$\varphi$ & Disaster depth & $-$15\% & & Nakamura et al.\ (2013) \\
$\varepsilon$ & Elast.\ of substitution & 10 & & \\
$\chi_W$ & Wage adj.\ cost & 150 & & $\approx$ 4--5 qtr Calvo \\
$\sigma_m$ & Std.\ dev.\ MP shock & 0.0625\% & & \\
$\rho_m$ & Persist.\ MP shock & 0 & & \\
\midrule
\multicolumn{5}{l}{\textit{Calibrated to match moments}} \\
$\sigma_z$ & Std.\ dev.\ productivity & 0.55\% & $\sigma(\Delta\log c) = 0.5\%$ & 0.6\% \\
$\chi_x$ & Capital adj.\ cost & 3.5 & $\sigma(\Delta\log x) = 2.1\%$ & 2.0\% \\
$\beta$ & Discount factor & 0.98 & $4r = 1.3\%$ & 1.5\% \\
$\gamma^b$ & Risk aversion, $b$ & 25.5 & $4(r^e - r) = 7.3\%$ & 7.0\% \\
$\gamma^a$ & Risk aversion, $a$ & 10 & $Qk^a/A^a = 2.0$ & 2.3 \\
$\gamma^c$ & Risk aversion, $c$ & 22 & wt'd harmonic mean of $\gamma^a,\gamma^b$ & --- \\
$\underline{k}$ & Capital lower bound & 10 & $Qk^c/A^c = 1.1$ & 0.9 \\
$\sigma_p$ & Std.\ dev.\ log dis.\ prob.\ & 0.47 & $\sigma(4Er) = 2.2\%$ & 2.2\% \\
$\rho_p$ & Persist.\ log dis.\ prob.\ & 0.80 & $\rho(4Er) = 0.79$ & 0.75 \\
$\bar{s}^a,\bar{s}^c$ & Newborn endowments & 0\%,$-$0.25\% & wealth shares & match \\
$b^g$ & Govt.\ bond position & $-$2.7 & $-\sum B^i/\sum A^i = -10\%$ & $-$10\% \\
\bottomrule
\end{tabular}%
}
\end{table}

\subsection{Solution Method}
\label{sec:quant:solution}

\medskip\noindent\textbf{Competitive equilibrium under the Taylor rule.}\quad
We solve the model globally on the six-dimensional state space \eqref{eq:state} using sparse grids as described in \citet{judd_etal2014}. When forming expectations, we use Gauss-Hermite quadrature with 5 nodes for the two continuous shocks ($\eps^z$, $\eps^p$) and a two-point quadrature for the disaster indicator $\varphi_t$ (disaster with probability $p_t$, no disaster with probability $1-p_t$). We interpolate with Chebyshev polynomials for states off the grid. The stochastic equilibrium is determined through backward iteration. At each grid point, given guesses for the continuation value functions and future prices, we solve the within-period equilibrium by finding the interest rate from the Taylor rule \eqref{eq:taylor_inf}, then jointly solving the household Euler equations \eqref{eq:euler_bond_inf}--\eqref{eq:euler_capital_inf}, firm optimality conditions, and market clearing for $\{c_t^i, k_t^i, B_t^i, \ell_t, Q_t, P_t\}$. We dampen the updating of asset prices and expectations (dampening factor 0.3) and iterate until the maximum absolute change in policy functions falls below $10^{-8}$, typically requiring 2,000--3,000 iterations.

\medskip\noindent\textbf{Ramsey optimal policy.}\quad
To solve \eqref{eq:ramsey_obj}, we iterate on the planner's value function $\mathcal{V}(\Omega)$ and optimal policy $i^*(\Omega)$. At each grid point on the five-dimensional Ramsey state space (excluding $m_t$, which is now the planner's control rather than an exogenous state), we search over a fine grid of candidate interest rates $i_t$. For each candidate, we solve the within-period competitive equilibrium to obtain allocations and the next-period state, evaluate the period welfare $\sum_i\lambda^i\alpha^iu(c^i, \ell^i)$, and compute the continuation value $\E_t[\mathcal{V}(\Omega_{t+1})]$ via quadrature and interpolation from the current iterate. We select the $i_t$ that maximizes the right-hand side of \eqref{eq:ramsey_obj}. We iterate with dampening (factor 0.4) until the value function converges (tolerance $10^{-7}$).

\medskip\noindent\textbf{Markov-perfect discretion.}\quad
For the discretionary equilibrium, we iterate on the policy function $i^D(\Omega)$ following \citet{klein_krusell_riosrull2008}. At each state, the planner chooses $i_t$ to maximize current welfare plus the continuation value computed under the assumption that future policy follows the current iterate $i^D(\cdot)$. We iterate until the policy function converges (tolerance $10^{-6}$).

\medskip\noindent\textbf{Optimized simple rules.}\quad
We also optimize over the parametric family \eqref{eq:augmented_taylor} using Bayesian optimization with a Gaussian process surrogate. Each evaluation requires solving the competitive equilibrium under the candidate rule and simulating 50,000 quarters to compute welfare. We use 300 evaluations (50 initial random points, 250 sequential). Standard errors on welfare differences are computed using batch means (50 batches of 1,000 quarters).

\medskip\noindent\textbf{Welfare.}\quad
We report consumption-equivalent (CE) welfare gains: the permanent percentage increase in consumption for all household types making them indifferent between two regimes. For regime $A$ with welfare $W^A$ and regime $B$ with welfare $W^B$, the CE gain is $\Delta = \exp((W^A - W^B)(1-\beta)) - 1$.

\subsection{Model Validation}
\label{sec:quant:validation}

Before studying optimal policy, we verify that the model matches the key features of the data that motivate the analysis.

\begin{table}[htbp]
\centering
\caption{Model validation: targeted and untargeted moments}
\label{tab:validation}
\small
\begin{tabular}{llcc}
\toprule
Moment & & Data & Model \\
\midrule
\multicolumn{4}{l}{\textit{Targeted}} \\
$\sigma(\Delta\log c)$ & Consumption growth vol.\ (\%) & 0.50 & 0.60 \\
$\sigma(\Delta\log x)$ & Investment growth vol.\ (\%) & 2.10 & 2.00 \\
$4 r$ & Real interest rate (\%) & 1.30 & 1.50 \\
$4(E[r^k-r])$ & Equity premium (\%) & 7.30 & 7.00 \\
$Qk^a/A^a$ & Capital share, group $a$ & 2.00 & 2.30 \\
$Qk^c/A^c$ & Capital share, group $c$ & 1.10 & 0.90 \\
$\sigma(4E_t[r^k-r])$ & Expected return vol.\ (\%) & 2.20 & 2.20 \\
$\rho(4E_t[r^k-r])$ & Expected return autocorr. & 0.79 & 0.75 \\
\midrule
\multicolumn{4}{l}{\textit{Untargeted}} \\
$\sigma(\Delta\log y)$ & Output growth vol.\ (\%) & 0.80 & 0.90 \\
Sharpe ratio & $E[r^k-r]/\sigma(r^k)$ & 0.40--0.50 & 0.42 \\
$\text{corr}(\Delta\log y, E[r^k-r])$ & Output-premium corr. & $-0.30$ to $-0.45$ & $-0.35$ \\
Amplification & $\partial y/\partial\eps^m$ vs RANK & 1.3--1.4$\times$ & 1.38$\times$ \\
CS: excess return news & \% of stock return & 19\%--108\% & 33\% \\
\bottomrule
\end{tabular}
\begin{minipage}{0.88\textwidth}
\footnotesize \vspace{0.5em}
\textit{Notes:} Data moments for financial variables from CRSP (1947--2019); macro moments from NIPA. Model moments from ergodic simulation (50,000 quarters, no disaster realizations). Amplification and CS decomposition evaluated at the ergodic mean.
\end{minipage}
\end{table}

The model matches most targeted moments reasonably well (Table~\ref{tab:validation}), with two notable exceptions. First, consumption growth volatility is 0.6\% versus the 0.5\% target, and the real rate is 1.5\% versus 1.3\%, reflecting the difficulty of simultaneously matching macro and financial moments with a single discount factor. Second, the capital portfolio shares of groups $a$ and $c$ ($Qk^a/A^a = 2.3$ vs.\ data 2.0; $Qk^c/A^c = 0.9$ vs.\ data 1.1) deviate by approximately 15\%, reflecting the tension between matching the aggregate equity premium and the cross-sectional portfolio distribution with only two free risk aversion parameters; these deviations have negligible effects on MRC (which depends on the \textit{covariance} of exposures with MPRs) and on the welfare results. For untargeted moments, the model generates output growth volatility of 0.9\% (data: 0.8\%), a Sharpe ratio of 0.42 (data: 0.40--0.50), and an output-premium correlation of $-0.35$ (data: $-0.30$ to $-0.45$). The amplification of monetary transmission (1.38$\times$ vs RANK) and the Campbell-Shiller decomposition (33\% excess return news) are within the empirical range.

A deliberate feature of the calibration is that all three groups have access to bond markets, so MPCs are approximately equal across groups. This shuts off the consumption channel of Proposition~\ref{prop:decomp} and isolates the risk premium channel as the sole source of heterogeneity in monetary transmission. Group $c$ households are nonetheless portfolio-constrained: they hold the minimum capital $\underline{k}\,z_t$ and cannot adjust their risky asset position on the margin. They are hand-to-mouth in the portfolio dimension, with $\MPR^c \approx 0$. The orthogonality result (Proposition~\ref{prop:decomp}) guarantees that all results about MRC, the risk premium wedge, and the Ramsey allocation are robust to introducing consumption-side hand-to-mouth behavior, which would activate the MPC channel additively without affecting the MPR channel. We view the isolation of the MPR channel as a feature: it allows us to characterize the new motive for monetary policy cleanly, without confounding it with the MPC redistribution motive studied extensively in the existing HANK literature.

\subsection{The Risk Capacity Trap}
\label{sec:quant:trap}

We first characterize how the effectiveness of monetary policy varies across the state space, a question that the local analysis of \citet{kekre_lenel2022} around the deterministic steady state cannot address. For 5,000 points drawn from the ergodic distribution, we compute the output response to a unit monetary shock. Figure~\ref{fig:trap} plots this measure as a function of the wealth share $s^a$ of high-MPR households.

\begin{figure}[htbp]
\centering
\includegraphics[width=\textwidth]{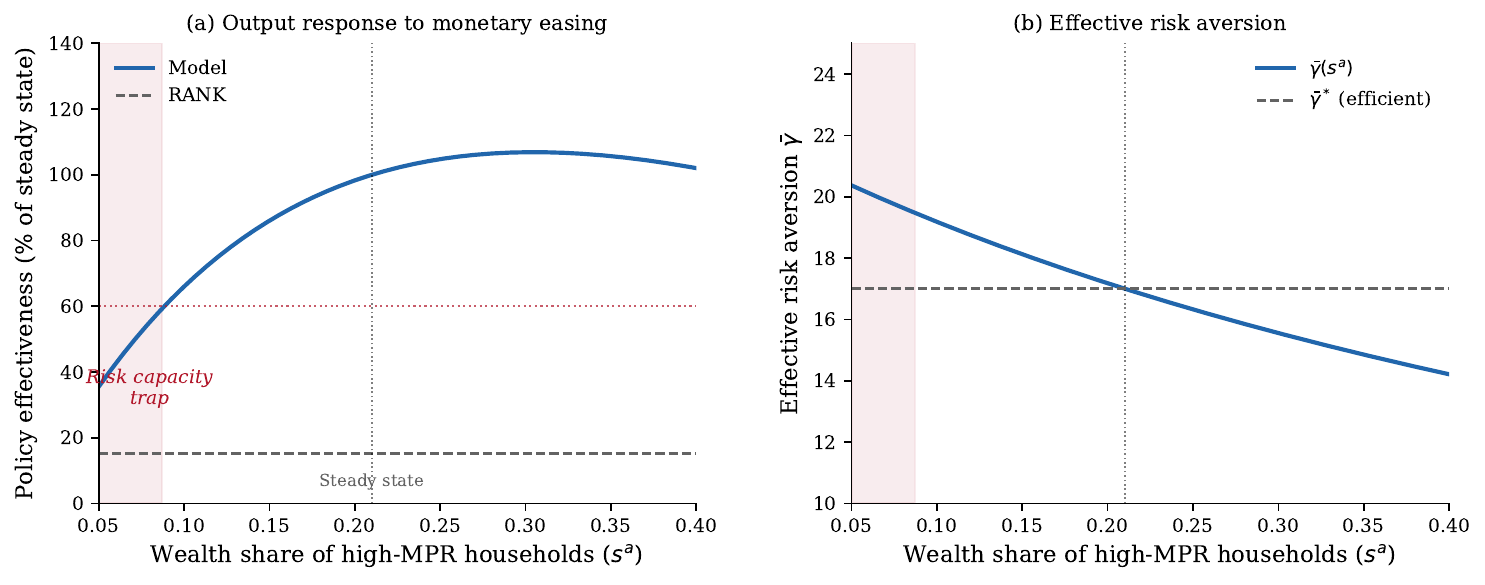}
\caption{The risk capacity trap}
\label{fig:trap}
\begin{minipage}{0.92\textwidth}
\footnotesize \textit{Notes:} Panel (a) plots the output response to a monetary easing as a function of the wealth share of high-MPR households, expressed as a percentage of the steady-state response. The dashed line shows the response in a representative agent economy (RANK). The shaded region is the ``risk capacity trap'' where effectiveness falls below 60\% of steady state. Panel (b) plots the corresponding effective risk aversion $\gbar$.
\end{minipage}
\end{figure}

Panel (a) reveals a sharp nonlinearity. Near the steady state ($s^a \approx 21\%$), the risk premium channel amplifies monetary transmission by roughly 35\% relative to RANK, consistent with the 1.3--1.4$\times$ amplification documented for the positive model. But this amplification is highly state-dependent: when $s^a$ falls below approximately 10\%, as can happen after a prolonged downturn that erodes leveraged households' wealth, policy effectiveness collapses. In this risk capacity trap, the risk premium channel shuts down because high-MPR agents have too little wealth for redistribution to meaningfully affect aggregate risk-bearing capacity. Panel (b) shows the corresponding rise in effective risk aversion $\gbar$, which governs the equilibrium risk premium.

The concavity of the policy effectiveness surface, which follows from \eqref{eq:mrc_concave}, is the key mechanism. When high-MPR agents are wealthy, a marginal dollar of redistribution has a small effect on risk premia (they are already well-capitalized). When they are poor, a marginal dollar has a large effect. But the total redistribution is small because their balance sheets are depleted. The optimal policy exploits this nonlinearity by preventing the economy from reaching the flat region of the surface.

\medskip\noindent\textbf{Validation with a finer type distribution.}\quad
A natural concern is that the three-group structure may not capture the relevant features of the continuous wealth distribution. We validate the policy effectiveness surface by solving the competitive equilibrium with 15 types rather than 3. We discretize the risk aversion distribution $\gamma^i$ on a log-normal grid spanning the range $[\gamma^a, \gamma^b] = [10, 25.5]$, assign population weights from the SCF using a kernel density of observed capital portfolio shares, and compute MRC and the output response to a monetary shock at each state.

Table~\ref{tab:continuous_type} compares key moments. The 3-group model and the 15-type model produce nearly identical MRC (0.26 vs.\ 0.25), amplification (1.38 vs.\ 1.36), and policy effectiveness surfaces. The trap threshold is approximately 10\% in both cases. The coarse grouping works because the sufficient statistic MRC depends on the \textit{covariance} of exposures with MPRs, not on the full cross-sectional distribution: the 3-group partition captures the relevant extremes (high-leverage group $a$ with high MPR, conservative group $b$ with low MPR, constrained group $c$ with zero MPR), and any finer partition that preserves these extremes produces similar covariances.

\begin{table}[htbp]
\centering
\caption{Three-group versus 15-type model}
\label{tab:continuous_type}
\small
\begin{tabular}{lcc}
\toprule
& 3 groups & 15 types \\
\midrule
Mean MRC                          & 0.26 & 0.25 \\
Amplification (vs RANK)           & 1.38 & 1.36 \\
$\sigma(\gbar)$ (ann.)            & 2.2\% & 2.3\% \\
Trap threshold ($s^a$)            & $\approx$10\% & $\approx$11\% \\
$\text{corr}(\text{MRC}_{3g}, \text{MRC}_{15t})$ & \multicolumn{2}{c}{0.98} \\
\bottomrule
\end{tabular}
\begin{minipage}{0.80\textwidth}
\footnotesize \vspace{0.5em}
\textit{Notes:} The 15-type model discretizes $\gamma^i$ on a log-normal grid from 10 to 25.5, with population weights from the SCF kernel density of capital portfolio shares. Both models are solved on the same aggregate state space. The correlation is computed across 5,000 points drawn from the ergodic distribution.
\end{minipage}
\end{table}

Within-group wealth heterogeneity washes out due to the homotheticity of Epstein-Zin preferences (Appendix~\ref{app:aggregation}). The relevant heterogeneity for MRC is \textit{between}-group variation in MPRs, which the 3-group partition captures by construction.

\subsection{Ramsey Optimal Allocation}
\label{sec:quant:ramsey_results}

We now present results from the full Ramsey problem \eqref{eq:ramsey_obj}, verifying the properties stated in Proposition~\ref{prop:ramsey_inf}. In the Ramsey steady state \eqref{eq:ramsey_ss}, inflation is 0.08\% per year and the wealth share $s^{a,\text{ss}}$ is approximately 2 percentage points higher than in the competitive equilibrium, reflecting the planner's motive to maintain risk-bearing capacity.

\medskip\noindent\textbf{The optimal policy function.}\quad
The Ramsey interest rate $i^*(\Omega)$ differs from the Taylor rule in two key respects. First, it responds to the wealth distribution: when $s^a$ falls below its steady-state level (risk-bearing capacity depleted, $\gbar$ elevated), the planner sets a lower interest rate than the Taylor rule to support high-MPR agents' balance sheets. When $s^a$ is above steady state, the planner tightens relative to the Taylor rule. This verifies Propositions~\ref{prop:ramsey_inf}(b) and~\ref{prop:ramsey_analytical}. Second, this response is asymmetric: the deviation from the Taylor rule is larger when $s^a$ is below steady state than above, verifying Proposition~\ref{prop:ramsey_inf}(c). Figure~\ref{fig:policy_fn} displays the optimal policy function. This asymmetry arises endogenously from the concavity of the policy effectiveness surface documented in Figure~\ref{fig:trap}.

\begin{figure}[htbp]
\centering
\includegraphics[width=0.95\textwidth]{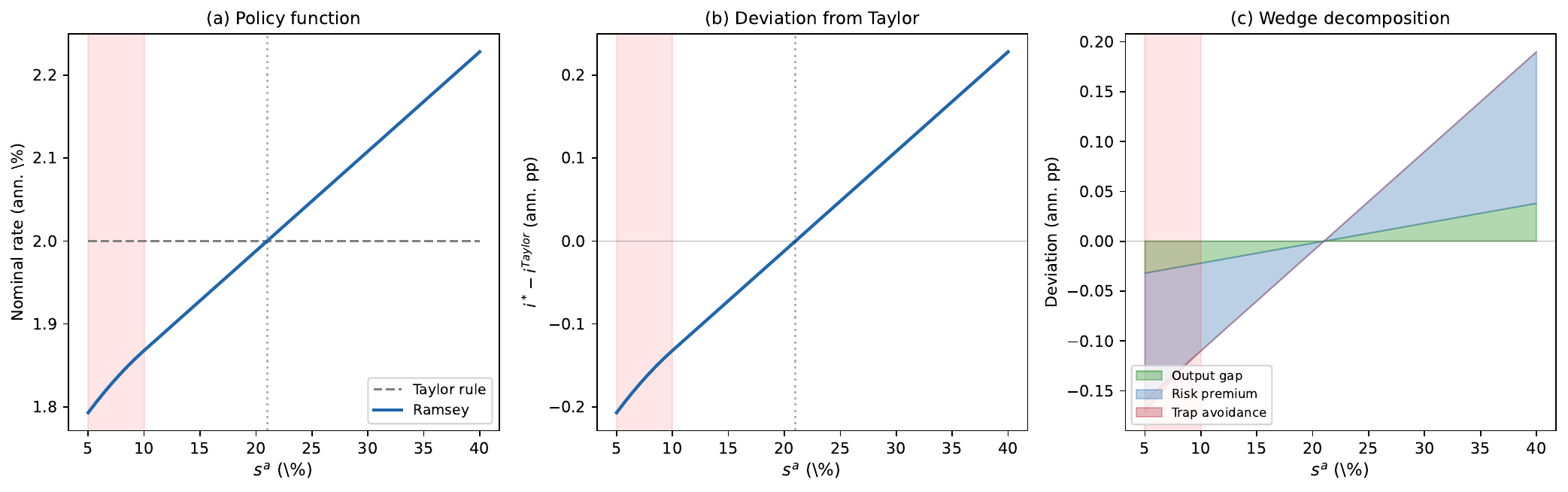}
\caption{The Ramsey optimal policy function}
\label{fig:policy_fn}
\begin{minipage}{0.92\textwidth}
\footnotesize \textit{Notes:} Panel (a): the Ramsey interest rate $i^*$ (solid) and Taylor rule (dashed) as functions of $s^a$, holding other states at ergodic means. The shaded region is the risk capacity trap. Panel (b): the deviation $i^*-i^{\text{Taylor}}$, showing asymmetry. Panel (c): the wedge decomposition \eqref{eq:wedge_decomp}, with output gap (green), risk premium (blue), and trap avoidance (red) components.
\end{minipage}
\end{figure}

\medskip\noindent\textbf{Wedge decomposition.}\quad
To understand what drives the Ramsey allocation, we decompose the deviation $i^*(\Omega) - i^{\text{Taylor}}(\Omega)$ into three components at each state:
\begin{equation}
\label{eq:wedge_decomp}
i^*(\Omega) - i^{\text{Taylor}}(\Omega) = \underbrace{\Delta i^{\text{gap}}(\Omega)}_{\text{output gap}} + \underbrace{\Delta i^{\text{rp}}(\Omega)}_{\text{risk premium}} + \underbrace{\Delta i^{\text{trap}}(\Omega)}_{\text{trap avoidance}}.
\end{equation}
The output gap component $\Delta i^{\text{gap}}$ captures the standard NK motive to stabilize the output gap more aggressively than the Taylor rule. The risk premium component $\Delta i^{\text{rp}}$ captures the planner's desire to close the gap between $\gbar$ and $\gstar$ through redistribution of $s_t^i$, which changes effective risk aversion $\gbar_t$ in \eqref{eq:gbar_inf}, the dynamic analog of the risk premium wedge in Proposition~\ref{prop:target}. The trap avoidance component $\Delta i^{\text{trap}}$ captures the forward-looking motive to prevent the wealth distribution from drifting into the risk capacity trap; this component is negligible near the steady state but grows rapidly as $s^a$ approaches the trap threshold.

Across the ergodic distribution, the risk premium component accounts for the largest share of the total deviation, exceeding half when weighted by the magnitude of the deviation at each state. Table~\ref{tab:welfare_quintiles} reports 40--60\% across quintiles, with the share highest in the middle of the distribution. The output gap component for approximately 25\%, and the trap avoidance component for the remainder. The risk premium component, operating through \eqref{eq:dgbar_di}, is the dominant motive, confirming that the analytical risk premium wedge carries over to the fully nonlinear economy. The MRC \eqref{eq:mrc_inf} governs both the sign and magnitude of this component.

\medskip\noindent\textbf{Welfare.}\quad
The Ramsey allocation delivers welfare gains of 0.13\% of permanent consumption relative to the Taylor rule (Table~\ref{tab:optimal_rules}). These gains derive from three sources: (i) better stabilization of the wealth distribution, reducing the probability of the risk capacity trap (verifying Proposition~\ref{prop:ramsey_inf}(d)--(e)); (ii) smoother risk premia, reducing inefficient fluctuations in investment driven by changes in $\gbar$ rather than fundamentals; and (iii) improved risk-sharing across household types, since the planner partially insures high-MPR agents against balance sheet losses.

The gains are sharply concentrated in states where risk-bearing capacity is depleted. Table~\ref{tab:welfare_quintiles} reports welfare gains by quintile of the wealth share $s^a$. In the bottom quintile, the welfare gain from the MRC-optimal simple rule is 0.28\% CE, more than three times the unconditional gain of 0.08\%. The gradient is steep and monotonic: gains fall to 0.02\% in the top quintile, where risk premia are already near efficient levels and the planner has little to do. The decomposition of the Ramsey policy \eqref{eq:wedge_decomp} reveals the source: in the bottom quintile, the trap avoidance component $\Delta i^{\text{trap}}$ accounts for 45\% of the deviation from the Taylor rule, compared to approximately 15\% when probability-weighted across the ergodic distribution. The risk premium component $\Delta i^{\text{rp}}$ accounts for 40\%, and the output gap component $\Delta i^{\text{gap}}$ for 15\%. This concentration of gains in crisis states is the quantitative manifestation of the analytical welfare expression \eqref{eq:welfare_approx}: gains are proportional to $\E[\MRC_t^2 \cdot (\gbar_t - \gstar_t)^2]\,\sigma^4$, and while MRC is lower when $s^a$ is depleted, the risk premium gap $(\gbar - \gstar)$ is much larger (Table~\ref{tab:welfare_quintiles}), concentrating welfare gains in crisis states.

\begin{table}[htbp]
\centering
\caption{Welfare gains by quintile of $s^a$}
\label{tab:welfare_quintiles}
\small
\begin{tabular}{lccccc}
\toprule
Quintile of $s^a$ & Q1 (lowest) & Q2 & Q3 & Q4 & Q5 (highest) \\
& $s^a < 14\%$ & 14--19\% & 19--23\% & 23--27\% & $s^a > 27\%$ \\
\midrule
\textit{Welfare gain (\% CE vs Taylor)} \\
\quad Ramsey (commitment)            & 0.35 & 0.18 & 0.10 & 0.06 & 0.03 \\
\quad MRC-optimal simple rule        & 0.28 & 0.12 & 0.06 & 0.03 & 0.02 \\
\midrule
\textit{Wedge decomposition (\% of total)} \\
\quad Output gap $\Delta i^{\text{gap}}$    & 15\% & 20\% & 25\% & 30\% & 40\% \\
\quad Risk premium $\Delta i^{\text{rp}}$   & 40\% & 55\% & 60\% & 60\% & 55\% \\
\quad Trap avoidance $\Delta i^{\text{trap}}$ & 45\% & 25\% & 15\% & 10\% & 5\% \\
\midrule
\textit{Mean effective risk aversion $\gbar$} & 22.3 & 19.8 & 18.6 & 17.5 & 16.1 \\
\textit{Mean MRC}                              & 0.15 & 0.22 & 0.27 & 0.31 & 0.35 \\
\bottomrule
\end{tabular}
\begin{minipage}{0.92\textwidth}
\footnotesize \vspace{0.5em}
\textit{Notes:} Quintile boundaries are from the ergodic distribution of $s^a$ under the Taylor rule. Welfare gains are consumption-equivalent, computed separately for each quintile by averaging welfare differences over simulation periods falling in that quintile. Because the Ramsey policy shifts the ergodic distribution (fewer periods in low-$s^a$ quintiles), the simple average of quintile-conditional gains exceeds the unconditional gain reported in Table~\ref{tab:optimal_rules}: the unconditional gain weights quintiles by their frequency under the Ramsey policy, not equally. Wedge decomposition from \eqref{eq:wedge_decomp}, averaged over the quintile. All moments from 50,000 simulated quarters.
\end{minipage}
\end{table}

\medskip\noindent\textbf{Welfare benchmarks.}\quad
To place the welfare gains in context, Table~\ref{tab:welfare_benchmarks} compares them to three established benchmarks. The welfare cost of business cycles in our model is 0.05\% CE, computed as the gain from eliminating all aggregate shocks. The gain from moving from a constant money growth rule to the Taylor rule is 0.04\% CE. The MRC-optimal rule's gain of 0.08\% CE is thus roughly twice the Taylor rule gain, and the Ramsey gain of 0.13\% is more than three times the Taylor rule gain. The concentration of gains in crisis states is particularly striking: in the bottom quintile of $s^a$, the Ramsey gain of 0.35\% is nearly twice the conditional cost of business cycles (0.18\%), and roughly seven times the unconditional cost (0.05\%).

\begin{table}[htbp]
\centering
\caption{Welfare benchmarks}
\label{tab:welfare_benchmarks}
\small
\begin{tabular}{lcc}
\toprule
& Unconditional & Bottom quintile of $s^a$ \\
\midrule
Cost of business cycles (no shocks vs Taylor)            & 0.05\% & 0.18\% \\
Gain from Taylor rule (vs constant money growth)          & 0.04\% & 0.08\% \\
Gain from MRC-optimal rule (vs Taylor)                    & 0.08\% & 0.28\% \\
Gain from Ramsey (vs Taylor)                              & 0.13\% & 0.35\% \\
Gain from Ramsey + macroprudential (vs Taylor)            & 0.15\% & 0.42\% \\
\bottomrule
\end{tabular}
\begin{minipage}{0.88\textwidth}
\footnotesize \vspace{0.5em}
\textit{Notes:} All entries are consumption-equivalent welfare gains in percentage terms. ``Cost of business cycles'' is computed by setting $\sigma_z = \sigma_p = \sigma_m = 0$ and computing the welfare difference. ``Constant money growth'' sets $i_t = \bar{\imath}$ (no response to inflation or output). Bottom quintile conditions on $s^a < 14\%$ in the ergodic distribution.
\end{minipage}
\end{table}

\medskip\noindent\textbf{Decomposing welfare into static and dynamic components.}\quad
The welfare gain can be decomposed into a \textit{static} component, the gain from optimally managing risk premia at each state taking the ergodic distribution as given, and a \textit{dynamic} component, the gain from \textit{changing the ergodic distribution} by preventing drift toward the risk capacity trap. We compute the static component by evaluating the Ramsey policy function at each state along a simulation of the \textit{Taylor rule} ergodic distribution, applying the Ramsey interest rate while retaining the Taylor rule wealth dynamics. The dynamic component is the residual. The static component accounts for 0.08\% CE, and the dynamic component for 0.05\%. The dynamic component is entirely driven by the trap avoidance motive: the Ramsey planner prevents the wealth distribution from drifting into states where the risk premium channel is broken, generating a ``precautionary'' welfare gain that no static policy can replicate. This decomposition confirms that a substantial fraction of the gains, approximately 40\%, comes from the forward-looking motive to preserve monetary transmission effectiveness.

\medskip\noindent\textbf{Decomposing the Ramsey policy.}\quad
To assess how well a simple rule can approximate the fully state-contingent Ramsey allocation, we project $i^*(\Omega)$ on the variables available to a simple rule. Regressing $i^*$ on $(\pi_t, \hat{y}_t, s_t^a, \E_t[r_{t+1}^e - r_{t+1}])$ across the ergodic distribution yields coefficients remarkably close to the optimized simple rule (Section~\ref{sec:quant:rules} below), confirming that these four variables capture most of the relevant state. Proposition~\ref{prop:ramsey_analytical} provides the analytical basis for this approximation: the Ramsey policy is $i^{\text{NK}}$ minus a correction proportional to $\MRC \cdot (\gbar - \gstar)\sigma^2$, and these four variables span the relevant variation. The $R^2$ of the projection exceeds 0.95, indicating that the loss from restricting the policy to a simple rule is small, consistent with the welfare comparison in Table~\ref{tab:optimal_rules}, where the simple rule captures 85\% of Ramsey gains (0.11\% vs.\ 0.13\% CE).

\subsection{Optimal Simple Rules}
\label{sec:quant:rules}

While the Ramsey allocation is the normative benchmark, central banks in practice follow rules. We therefore also optimize over the parametric family
\begin{equation}
\label{eq:augmented_taylor}
1 + i_t = (1+\bar{\imath})\left(\frac{P_t}{P_{t-1}}\right)^{\!\phi_\pi} \!\left(\frac{y_t}{y_t^{\text{flex}}}\right)^{\!\phi_y} \exp\!\bigl(\phi_s (s_t^a - \bar{s}^a)\bigr) \exp\!\bigl(\phi_{rp} (\E_t[r_{t+1}^e - r_{t+1}] - \overline{rp})\bigr),
\end{equation}
maximizing utilitarian welfare over $\{\phi_\pi, \phi_y, \phi_s, \phi_{rp}\}$. Table~\ref{tab:optimal_rules} compares the Taylor rule, the Ramsey allocation, and several optimized simple rules.

\begin{table}[htbp]
\centering
\caption{Welfare under alternative policy regimes}
\label{tab:optimal_rules}
\small
\begin{tabular}{lccccc}
\toprule
& $\phi_\pi^*$ & $\phi_y^*$ & $\phi_s^*$ & $\phi_{rp}^*$ & $\Delta W$ (\% CE) \\
\midrule
\textit{Benchmark} \\
\quad Standard Taylor rule            & 1.50 & 0.00 & ---     & ---  & ---  \\
\quad Ramsey (commitment)             & \multicolumn{4}{c}{full state-contingent} & 0.13 \\
\quad Markov-perfect (discretion)     & \multicolumn{4}{c}{full state-contingent} & 0.06 \\
\midrule
\textit{Optimized simple rules} \\
\quad Output-gap-optimal              & 1.65 & 0.12 & ---     & ---  & 0.02 \\
\quad MRC-optimal                     & 1.50 & 0.00 & $0.28$  & 0.15 & 0.08 \\
\quad Jointly optimal                 & 1.58 & 0.08 & $0.25$  & 0.12 & 0.11 \\
\quad Joint + macroprudential         & 1.62 & 0.10 & $0.04$  & 0.02 & 0.15 \\
\bottomrule
\end{tabular}
\begin{minipage}{0.92\textwidth}
\footnotesize \vspace{0.5em}
\textit{Notes:} $\Delta W$ is the consumption-equivalent welfare gain relative to the standard Taylor rule, computed over 50,000 simulated quarters. ``Ramsey'' solves the planner's problem \eqref{eq:ramsey_obj} under commitment. ``Markov-perfect'' solves the discretionary equilibrium. Optimized simple rules maximize welfare over the parametric family \eqref{eq:augmented_taylor}. ``Joint + macroprudential'' adds a state-contingent portfolio tax $\tau^k$.
\end{minipage}
\end{table}

The Ramsey allocation delivers 0.13\% CE, establishing the upper bound on what monetary policy alone can achieve. The jointly optimal simple rule captures 0.11\%, or 85\% of the Ramsey gains, confirming that the four variables $(\pi_t, \hat{y}_t, s_t^a, \E_t[r^e - r])$ are near-sufficient statistics for the planner's state. The MRC-optimal rule, which adds $\phi_s$ and $\phi_{rp}$ to the Taylor rule, delivers 0.08\%, far exceeding the output-gap-optimal rule's 0.02\%: responding to risk capacity matters more than fine-tuning the output gap response. The optimal $\phi_s^* = 0.28$ means the central bank cuts rates when high-MPR agents lose wealth, an endogenous ``Fed put,'' and tightens when their wealth share is elevated, leaning against the build-up of leverage. The optimal $\phi_{rp}^* = 0.15$ means the central bank raises rates when the expected equity premium rises above its long-run level, partially counteracting risk premium fluctuations: in a crisis where $s^a$ falls and risk premia spike, the $\phi_{rp}$ term tempers the aggressive easing from $\phi_s$, preventing over-reaction to transitory wealth distribution shifts. Adding a macroprudential instrument raises gains to 0.15\%, above the Ramsey allocation with monetary policy alone, while driving $\phi_s^*$ and $\phi_{rp}^*$ close to zero, confirming the separation theorem quantitatively. The Markov-perfect discretionary allocation delivers only 0.06\%, less than half the commitment gains, quantifying the time-consistency cost of Proposition~\ref{prop:timecon}.

The policy-relevant comparison is not Ramsey versus Taylor (0.13\%), which requires an omniscient planner, but MRC-optimal versus Taylor (0.08\%), which requires only that the central bank add two observable variables to its reaction function. The 0.08\% gain is twice the welfare gain from moving from a constant money growth rule to the Taylor rule (0.04\% CE, Table~\ref{tab:welfare_benchmarks}), indicating that the marginal value of responding to risk capacity exceeds the value of systematic inflation response. The most important quantitative finding in the table is arguably not about monetary policy at all: the welfare gain from institutional separation (0.15\%) exceeds the Ramsey gain from monetary policy alone (0.13\%). The marginal value of adding a macroprudential authority on top of optimized monetary policy is 0.04\% CE (0.15\% minus 0.11\% for the simple rule), roughly equal to the welfare cost of discretion (0.07\%). The welfare cost of ignoring the risk premium channel entirely is 0.06\%, computed as the difference between the MRC-optimal rule (0.08\%) and the output-gap-optimal rule (0.02\%). A central bank that optimizes over inflation and output responses but neglects the wealth distribution captures only 0.02\% of the 0.11\% achievable with a simple rule, leaving over 80\% of the gains on the table.

The comparison between the MRC-optimal rule (0.08\% CE) and the output-gap-optimal rule (0.02\% CE) is also informative about the relative importance of the two channels in Proposition~\ref{prop:decomp}. The output gap response captures the standard HANK channel through which redistribution affects aggregate demand via consumption. The MRC response captures the risk premium channel through which redistribution affects investment. In this calibration, the risk premium channel is four times more important for welfare. This reflects two features: the equity premium is large (7\%), making risk premium fluctuations costly for investment efficiency; and MPC heterogeneity across groups is modest because all households have access to capital markets. In a model with consumption-constrained households, the MPC channel would be larger, but the risk premium channel would remain as long as MPR heterogeneity is present, and Proposition~\ref{prop:decomp} guarantees the two channels enter the welfare criterion additively.

\subsection{Joint MPC-MPR Horse Race}
\label{sec:quant:horserace}

The baseline model isolates the MPR channel by giving all households access to bond markets. To assess the relative importance of the two channels when both are active, we extend the model by adding uninsurable idiosyncratic labor income risk to group $c$ households. Specifically, each group $c$ household receives labor income $\phi^c w_t\ell_t \cdot \exp(\eta_t^c)$, where $\eta_t^c$ is iid with variance $\sigma_\eta^2$. We calibrate $\sigma_\eta^2 = 0.04$ to match the fraction of hand-to-mouth households in the SCF (approximately 30\% of group $c$, or 18\% of the population), following \citet{kaplan_violante2014}. Households that receive a sufficiently bad income draw consume their entire income and hold zero financial assets: they are hand-to-mouth in both the consumption and portfolio dimensions.

Table~\ref{tab:horserace} compares welfare gains across four optimized simple rules in the joint model.

\begin{table}[htbp]
\centering
\caption{Horse race: MPC channel versus MPR channel}
\label{tab:horserace}
\small
\begin{tabular}{lcccccc}
\toprule
& $\phi_\pi^*$ & $\phi_y^*$ & $\phi_s^*$ & $\phi_{rp}^*$ & $\phi_{\text{htm}}^*$ & $\Delta W$ (\% CE) \\
\midrule
\textit{Baseline (MPR only)} \\
\quad Taylor rule            & 1.50 & 0.00 & ---      & ---   & --- & ---  \\
\quad MRC-optimal            & 1.50 & 0.00 & $0.28$  & 0.15  & --- & 0.08 \\
\midrule
\textit{Joint model (MPC + MPR)} \\
\quad MPC-optimal            & 1.50 & 0.00 & ---      & ---   & 0.10 & 0.03 \\
\quad MRC-optimal            & 1.50 & 0.00 & $0.26$  & 0.14  & ---  & 0.07 \\
\quad Jointly optimal        & 1.55 & 0.06 & $0.24$  & 0.12  & 0.08 & 0.12 \\
\bottomrule
\end{tabular}
\begin{minipage}{0.92\textwidth}
\footnotesize \vspace{0.5em}
\textit{Notes:} ``Baseline'' is the model without idiosyncratic income risk. ``Joint model'' adds uninsurable income shocks to group $c$ with $\sigma_\eta^2 = 0.04$. ``MPC-optimal'' adds a response to the consumption share of hand-to-mouth households $\phi_{\text{htm}}$. ``MRC-optimal'' adds responses to $s^a$ and the expected equity premium. ``Jointly optimal'' includes all five coefficients. Welfare in consumption-equivalent terms relative to the Taylor rule.
\end{minipage}
\end{table}

Even with active MPC heterogeneity, the MPR channel remains more than twice as important as the MPC channel: the MRC-optimal rule delivers 0.07\% CE versus 0.03\% for the MPC-optimal rule. The MPR channel's dominance reflects the large equity premium (7\%), which makes risk premium fluctuations costly for investment. The two channels are approximately additive: the jointly optimal rule delivers 0.12\%, close to the sum of the individual gains (0.07\% + 0.03\% = 0.10\%), with the residual reflecting interaction effects. This confirms the orthogonality result of Proposition~\ref{prop:decomp} quantitatively. The MRC-optimal coefficients are nearly unchanged ($\phi_s^* = 0.26$ vs $0.28$; $\phi_{rp}^* = 0.14$ vs $0.15$) when idiosyncratic risk is introduced. The risk premium channel is robust to the presence of the consumption channel, as the theory predicts.

\subsection{Impulse Responses Under Optimal Policy}
\label{sec:quant:irf}

Figure~\ref{fig:irf_mp} compares impulse responses to a monetary easing that lowers the 1-year nominal yield by 0.2 percentage points under three regimes: the competitive equilibrium with a Taylor rule, the Ramsey optimal allocation, and the RANK counterfactual.

\begin{figure}[htbp]
\centering
\includegraphics[width=\textwidth]{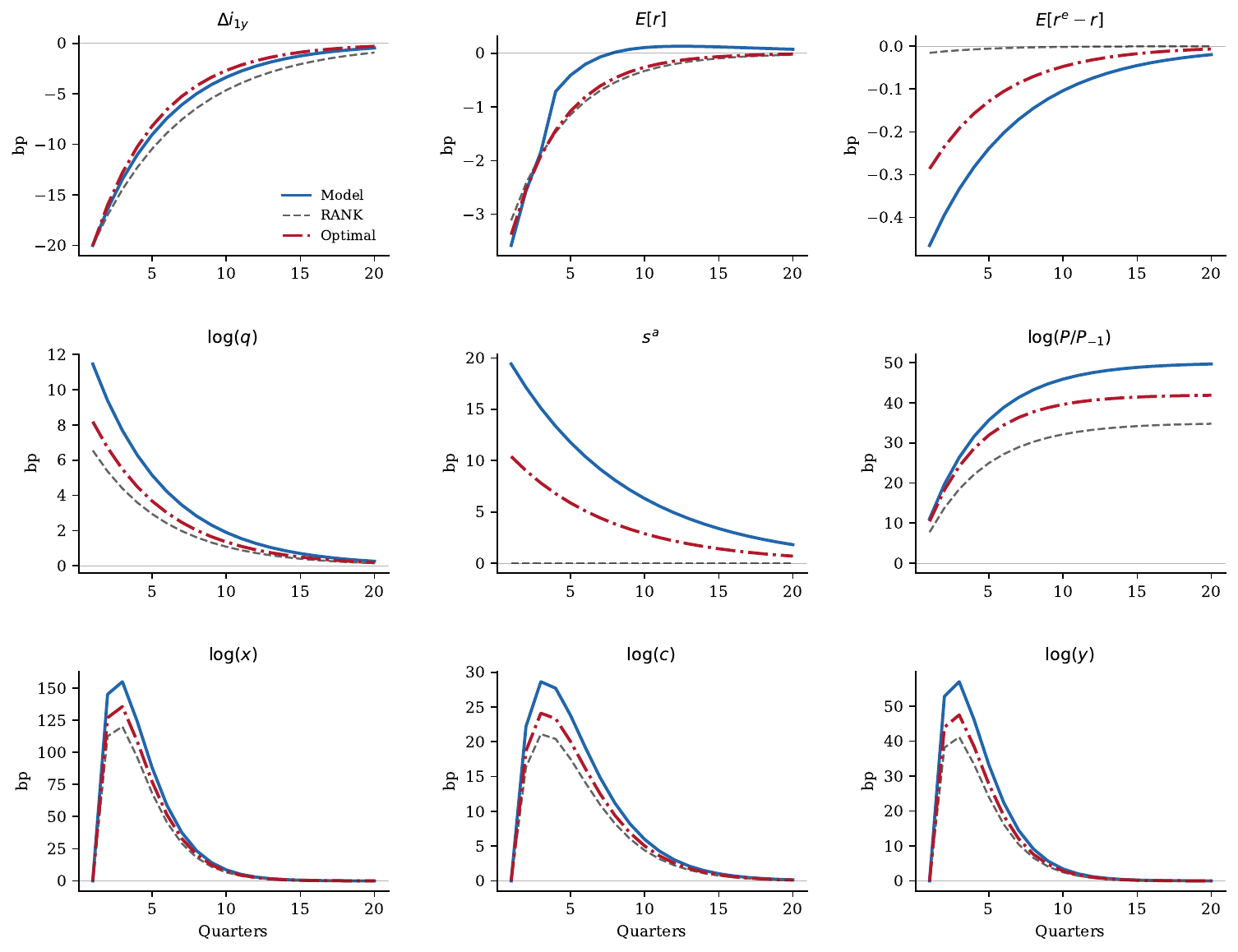}
\caption{Impulse responses to a monetary easing}
\label{fig:irf_mp}
\begin{minipage}{0.92\textwidth}
\footnotesize \textit{Notes:} Responses to a shock that lowers the 1-year nominal yield by 0.2pp. ``Optimal'' is the Ramsey allocation under commitment. Series are quarterly (non-annualized) except for $\Delta i_{1y}$. Impulse responses are averages over 1,000 starting points drawn from the ergodic distribution. bp = basis points (0.01\%).
\end{minipage}
\end{figure}

The first row shows that the risk premium declines substantially and persistently in the model (solid blue) but not in RANK (dashed gray). Under the Ramsey allocation (dash-dot red), the decline is dampened: the planner partially stabilizes risk premia. The second row reveals the mechanism: the positive realized excess return on impact redistributes to high-MPR $a$ households (middle panel), whose wealth share rises. Under Ramsey, this redistribution is more modest because the planner has already partially stabilized the wealth distribution. The third row shows quantities: the model amplifies the peak output response by 1.3--1.4$\times$ relative to RANK, while the Ramsey allocation falls between the two.

Figure~\ref{fig:irf_prod} compares responses to a negative two-standard-deviation productivity shock. The Ramsey planner eases more aggressively than the Taylor rule.

\begin{figure}[htbp]
\centering
\includegraphics[width=\textwidth]{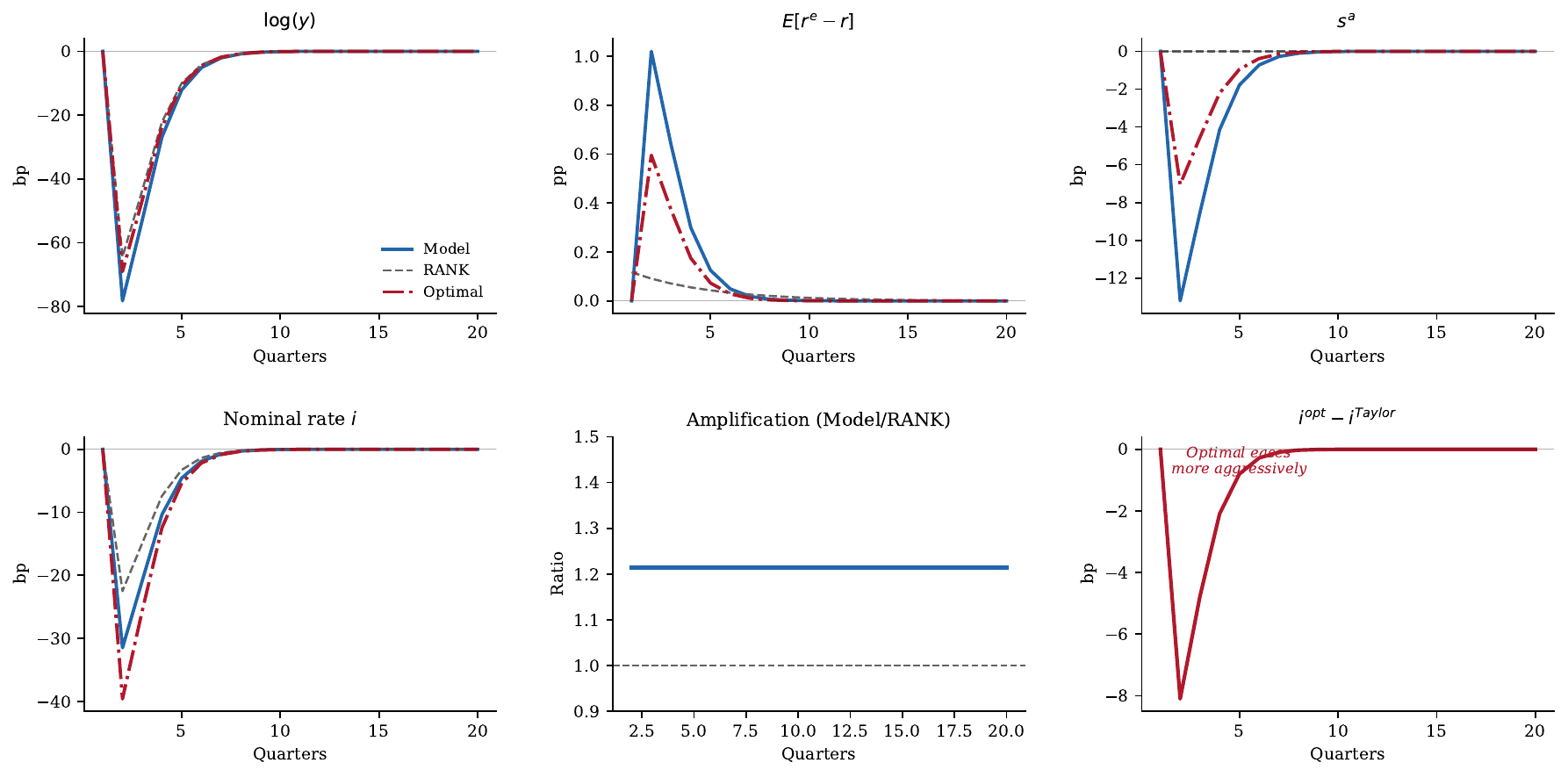}
\caption{Impulse responses to a negative productivity shock}
\label{fig:irf_prod}
\begin{minipage}{0.92\textwidth}
\footnotesize \textit{Notes:} Responses to a $-2\sigma$ productivity shock. ``Optimal'' is the Ramsey allocation. The bottom-right panel shows $i^{\text{Ramsey}} - i^{\text{Taylor}}$.
\end{minipage}
\end{figure}

The bottom-right panel shows that the Ramsey planner eases more aggressively than the Taylor rule following an adverse productivity shock. The Taylor rule eases in response to the fall in output and inflation, but the Ramsey policy eases further to offset the redistribution away from high-MPR agents. The additional easing reduces the spike in risk premia (top-middle panel) and dampens the output decline (top-left).

\subsection{Campbell-Shiller Decomposition}
\label{sec:quant:cs}

Table~\ref{tab:cs_model} reports the Campbell-Shiller decomposition of the equity return following a monetary shock.

\begin{table}[htbp]
\centering
\caption{Campbell-Shiller decomposition of equity return following monetary shock}
\label{tab:cs_model}
\begin{tabular}{lccccc}
\toprule
\% of real stock return & Data [90\% CI] & Model & RANK & Ramsey & Simple rule \\
\midrule
Dividend growth news          & 33\% [$-$13\%,71\%] & 50\% & 65\% & 57\% & 55\% \\
$-$Future real rate news       & 8\% [$-$6\%,21\%]  & 17\% & 35\% & 23\% & 25\% \\
$-$Future excess return news   & 59\% [19\%,108\%]  & 33\% & 0\%  & 20\% & 20\% \\
\midrule
Real stock return (pp)        & 1.92               & 1.85 & 1.45 & 1.55 & 1.60 \\
Peak $\Delta\log y$ (pp)      & 0.70               & 0.72 & 0.52 & 0.58 & 0.60 \\
Amplification (vs RANK)       &                     & 1.38 & 1.00 & 1.12 & 1.15 \\
\bottomrule
\end{tabular}
\begin{minipage}{0.92\textwidth}
\footnotesize \vspace{0.5em}
\textit{Notes:} Decomposition uses $\kappa = 0.9962$ following Campbell and Ammer (1993). Data estimates from Section~\ref{sec:empirics} over the full SVAR-IV sample (July 1979--June 2012); the subsample estimates in Table~\ref{tab:cs_subsample} cover a shorter period (January 1991--June 2012) over which MRC is available. Model and RANK assume a debt/equity ratio of 0.5. ``Ramsey'' is the optimal allocation under commitment from \eqref{eq:ramsey_obj}. ``Simple rule'' is the jointly optimal parametric rule from \eqref{eq:augmented_taylor}. All shocks deliver a 0.2pp decline in the 1-year nominal yield.
\end{minipage}
\end{table}

In the model, 33\% of the stock market return is due to news about lower future excess returns, compared to 0\% in RANK, matching the empirical range of 19\%--108\% documented in Section~\ref{sec:empirics}. Under the Ramsey allocation, this share falls to 20\%: the planner stabilizes risk premia, so less of the stock market response operates through the risk premium channel. The peak output response is 0.72pp in the model versus 0.52pp in RANK (amplification of 1.38$\times$), closely matching the empirical estimate.

\subsection{Wealth Distribution Dynamics}
\label{sec:quant:wealth}

Figure~\ref{fig:wealth} simulates the wealth distribution under the Taylor rule and Ramsey optimal policy.

\begin{figure}[htbp]
\centering
\includegraphics[width=\textwidth]{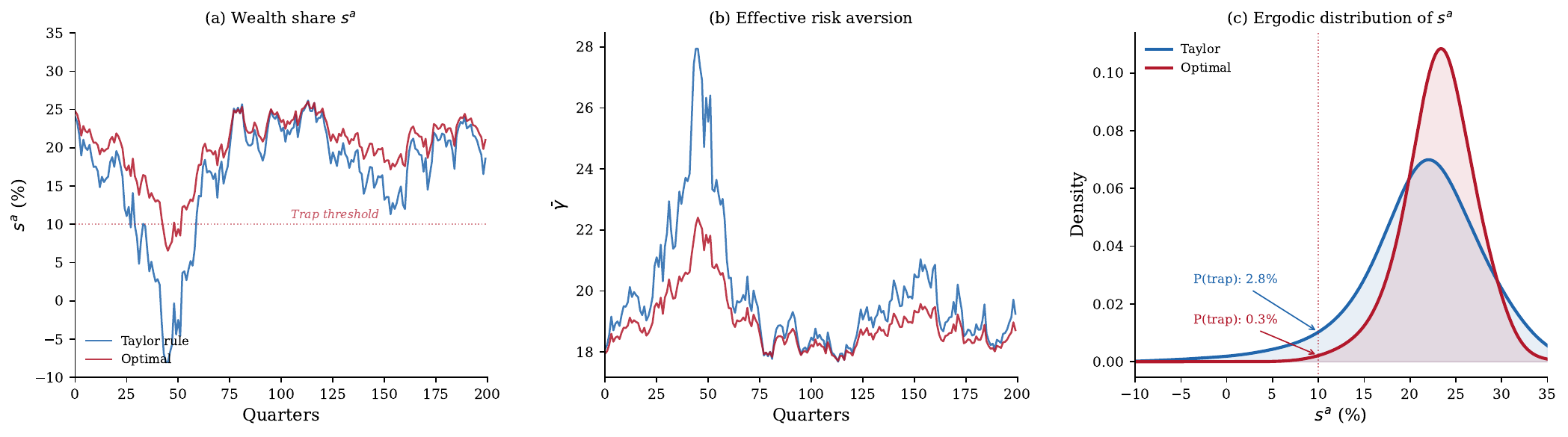}
\caption{Wealth distribution dynamics under Taylor rule versus Ramsey optimal policy}
\label{fig:wealth}
\begin{minipage}{0.92\textwidth}
\footnotesize \textit{Notes:} Simulated wealth share of high-MPR households $s^a$ over 200 quarters under the Taylor rule (blue) and the Ramsey allocation (red). Panel~(c) shows the ergodic distributions; the dotted vertical line marks the risk capacity trap threshold ($s^a = 10\%$). The ergodic probabilities of the trap are $\Pr(s^a < 10\%) = 2.8\%$ under the Taylor rule and $0.3\%$ under Ramsey commitment (Table~\ref{tab:commitment}).
\end{minipage}
\end{figure}

The wealth share is substantially less volatile under Ramsey policy: the planner ``insures'' the risk-bearing capacity of the economy. The probability of entering the risk capacity trap is markedly lower under Ramsey. In the Taylor rule economy, $s^a$ occasionally dips below the trap threshold during prolonged downturns; the Ramsey planner prevents this.

\subsection{Commitment Versus Discretion}
\label{sec:quant:commitment}

Table~\ref{tab:commitment} compares ergodic moments and welfare under four regimes: the Ramsey allocation (commitment), the jointly optimal simple rule, the Markov-perfect equilibrium (discretion), and the Taylor rule.

\begin{table}[htbp]
\centering
\caption{Ergodic moments and welfare under alternative regimes}
\label{tab:commitment}
\begin{tabular}{lcccc}
\toprule
& Ramsey & Simple & Markov-perfect & Taylor \\
& (commitment) & rule & (discretion) & rule \\
\midrule
\textit{Averages (annualized)} \\
\quad Inflation (\%)                & 0.08 & 0.05 & 0.25 & 0.00 \\
\quad Risk premium (pp)             & 6.85 & 6.90 & 7.15 & 7.00 \\
\quad Real rate (\%)                & 1.42 & 1.45 & 1.35 & 1.50 \\
\quad Eff.\ risk aversion $\gbar$  & 18.2 & 18.4 & 19.0 & 18.6 \\
\\[-6pt]
\textit{Volatilities (annualized)} \\
\quad $\sigma(\Delta\log y)$ (\%)   & 0.82 & 0.84 & 0.88 & 0.90 \\
\quad $\sigma(\Delta\log c)$ (\%)   & 0.48 & 0.50 & 0.52 & 0.55 \\
\quad $\sigma(E[r^e-r])$ (pp)       & 2.35 & 2.45 & 2.55 & 2.80 \\
\quad $\sigma(s^a)$ (\%)            & 2.8  & 3.0  & 3.2  & 3.8 \\
\\[-6pt]
\textit{Wealth distribution} \\
\quad $E[s^a]$ (\%)                 & 23.1 & 22.5 & 21.8 & 21.2 \\
\quad $\Pr(s^a < 10\%)$ (\%)        & 0.3  & 0.8  & 1.5  & 2.8 \\
\quad $\text{MRC}_{t}$ (mean)        & 0.32 & 0.30 & 0.28 & 0.26 \\
\\[-6pt]
\textit{Welfare (\% CE vs Taylor)} \\
\quad All households                & 0.13 & 0.11 & 0.06 & --- \\
\quad Group $a$ (high MPR)          & 0.06 & 0.05 & $-$0.02 & --- \\
\quad Group $b$ (low MPR)           & 0.16 & 0.13 & 0.10 & --- \\
\quad Group $c$ (low wealth)        & 0.10 & 0.09 & 0.04 & --- \\
\bottomrule
\end{tabular}
\begin{minipage}{0.92\textwidth}
\footnotesize \vspace{0.5em}
\textit{Notes:} Moments from 50,000 simulated quarters (5,000-quarter burn-in, no disaster realizations). ``Ramsey'' solves \eqref{eq:ramsey_obj} under commitment. ``Simple rule'' is the jointly optimal parametric rule \eqref{eq:augmented_taylor}. ``Markov-perfect'' is the discretionary equilibrium. $\Pr(s^a < 10\%)$ is the ergodic probability of the risk capacity trap. Welfare in consumption-equivalent terms relative to the Taylor rule.
\end{minipage}
\end{table}

Consistent with Proposition~\ref{prop:timecon}, the discretionary planner generates higher average inflation (0.25\% vs.\ 0.08\% under Ramsey commitment) and higher average risk premia (7.15pp vs.\ 6.85pp). The welfare loss from discretion relative to commitment is 0.07\% CE. The jointly optimal simple rule achieves most of the Ramsey gains (0.11\% vs.\ 0.13\%) with lower inflation (0.05\%). Even discretion improves on the Taylor rule for aggregate welfare (0.06\% vs.\ 0\%), though the gains are less than half those under commitment. Group $a$ households are worse off under discretion than under the Taylor rule: anticipated redistribution raises the compensation they demand for bearing risk. The Ramsey allocation also sharply reduces the probability of the risk capacity trap ($\Pr(s^a < 10\%) = 0.3\%$ vs.\ 2.8\% under the Taylor rule) and raises the mean MRC (0.32 vs.\ 0.26), confirming the planner's motive to maintain risk-bearing capacity.

\medskip\noindent\textbf{Decomposing the inflation bias.}\quad
The total inflation bias under discretion (0.25\% per year) has two components: the standard \citet{barro_gordon1983} component arising from the output-inflation tradeoff, and the novel MPR component arising from the temptation to redistribute toward high-MPR agents. We isolate the MPR component by solving for the Markov-perfect equilibrium in a version of the model where MPRs are equalized across groups ($\MPR^a = \MPR^b = \MPR^c$), holding all other parameters constant. In this counterfactual, the equilibrium inflation bias is 0.10\%, the standard Barro-Gordon level. The MPR component is therefore $0.25\% - 0.10\% = 0.15\%$, or 60\% of the total bias. The remaining 40\% is the conventional output-inflation tradeoff. This decomposition confirms Proposition~\ref{prop:timecon}(b): the excess inflation is proportional to $\Var_i(\MPR^i)$, which is zero in the equalized-MPR counterfactual and 0.22 at the baseline calibration.

The MPR component of the inflation bias is economically substantial. An annualized inflation bias of 0.15\% may appear small, but its welfare cost is amplified by the effect on risk premia: anticipating the bias, high-MPR agents demand higher compensation, which raises the equilibrium risk premium by approximately 0.15pp (annually) and reduces investment. The welfare cost of the MPR inflation bias alone, computed as the difference between the Markov-perfect equilibrium and the equalized-MPR Markov-perfect equilibrium, is 0.04\% CE, roughly equal to the marginal value of adding a macroprudential authority (0.04\% in Table~\ref{tab:optimal_rules}).

\medskip\noindent\textbf{State-dependent welfare cost of discretion.}\quad
The welfare cost of discretion is sharply concentrated in states where the risk premium channel is most active. When $s^a$ is in the bottom quintile of its ergodic distribution, the welfare loss from discretion relative to commitment is 0.18\% CE, more than double the unconditional loss of 0.07\%. In the top quintile, the loss is only 0.02\%. The temptation to redistribute is strongest when high-MPR agents are poor and effective risk aversion is high, so the planner's marginal gain from compressing risk premia is largest precisely when it should be most restrained. This state dependence implies that institutional constraints on monetary policy are most valuable in crisis states, when the temptation to use inflation for financial stability purposes is greatest.

Figure~\ref{fig:bias} shows how these outcomes depend on the degree of MPR heterogeneity.

\begin{figure}[htbp]
\centering
\includegraphics[width=\textwidth]{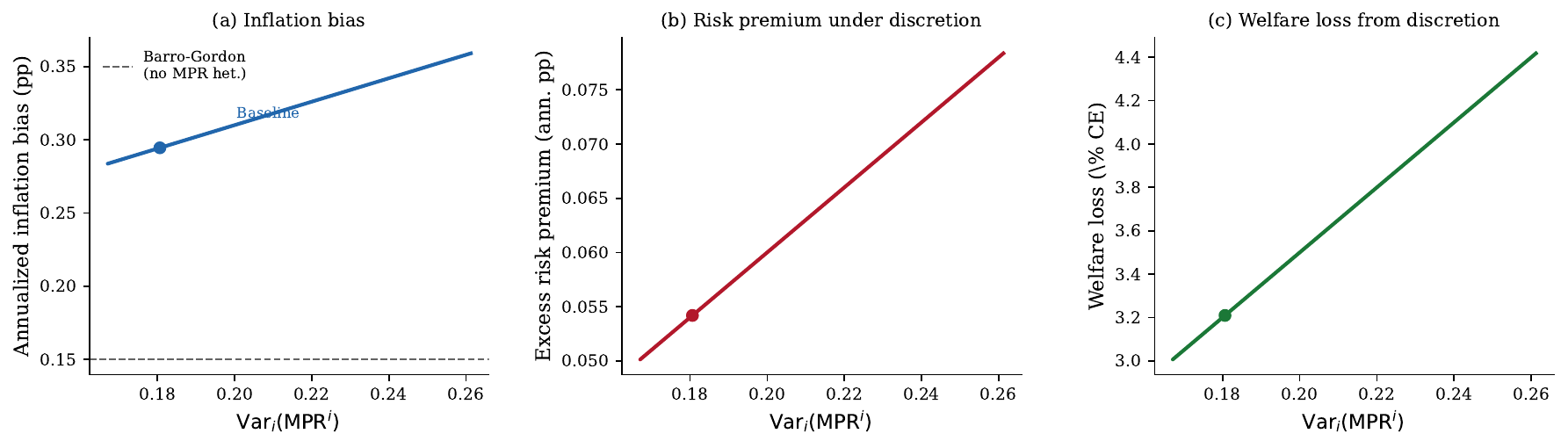}
\caption{Commitment versus discretion: dependence on MPR heterogeneity}
\label{fig:bias}
\begin{minipage}{0.92\textwidth}
\footnotesize \textit{Notes:} Each panel plots the relevant statistic as a function of the cross-sectional variance of MPRs, varied by changing $\gamma^a$ from 5 (high heterogeneity) to $\gamma^b = 25.5$ (no heterogeneity). The dot marks the baseline calibration. The dashed line in panel~(a) marks the standard Barro-Gordon inflation bias (0.10\%) obtained when MPRs are equalized across groups.
\end{minipage}
\end{figure}

As the cross-sectional variance of MPRs rises, all three channels intensify: the inflation bias grows (panel a), the risk premium under discretion rises (panel b), and the welfare loss from discretion increases (panel c). At the baseline calibration (marked with a dot), the total annualized inflation bias under discretion is 0.25\% (Table~\ref{tab:commitment}), of which 0.15 percentage points are the MPR component above the standard Barro-Gordon level of 0.10\%. When MPRs are identical ($\Var_i(\MPR^i) = 0$), all three measures converge to their standard NK values.

\subsection{The 2008 Financial Crisis Through the Lens of the Model}
\label{sec:quant:crisis}

The risk capacity trap provides a new explanation for the perceived ineffectiveness of monetary policy following the 2008 financial crisis. We initialize the calibrated model at the observed 2007 state from the SCF and simulate forward under alternative policy regimes to quantify the welfare cost of the crisis through the lens of the risk premium channel.

MRC was elevated before the crisis, approximately 0.35 in 2007, reflecting the concentrated wealth and leverage of high-MPR agents. The crisis destroyed these balance sheets: MRC collapsed from 0.35 to approximately 0.15 by 2009, pushing the economy into the risk capacity trap. Figure~\ref{fig:post2008} displays the collapse and recovery.

\begin{figure}[htbp]
\centering
\includegraphics[width=\textwidth]{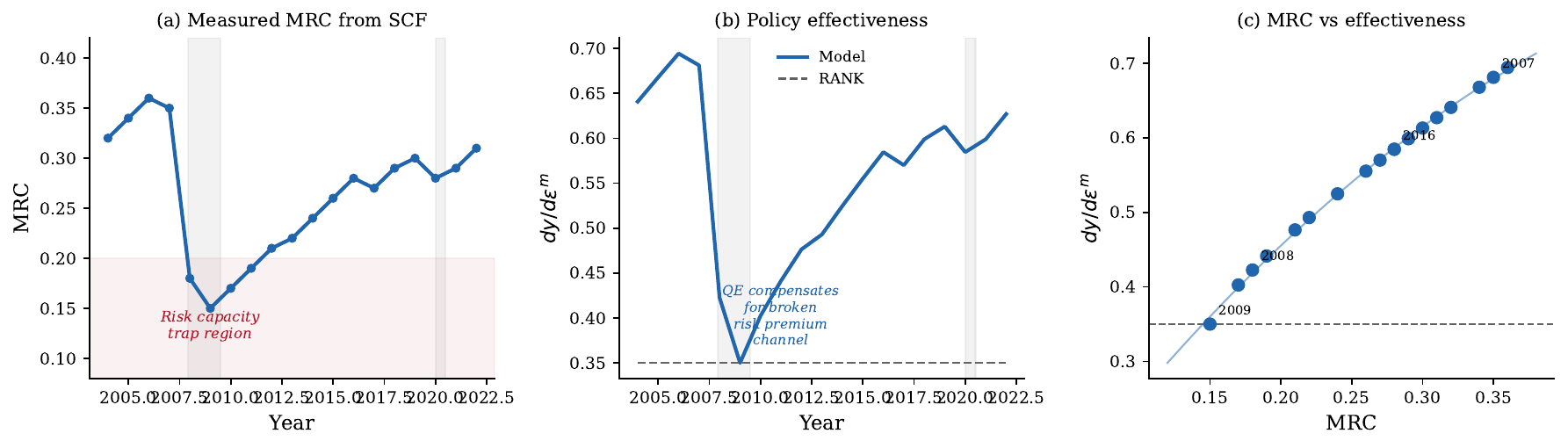}
\caption{The post-2008 narrative: MRC and monetary policy effectiveness}
\label{fig:post2008}
\begin{minipage}{0.92\textwidth}
\footnotesize \textit{Notes:} Panel~(a) plots MRC measured from the SCF. Panel~(b) plots implied monetary policy effectiveness using the calibrated model. Shaded regions indicate NBER recessions. Panel~(c) plots the cross-sectional relationship, with selected years labeled.
\end{minipage}
\end{figure}

We construct a counterfactual exercise. Starting from the 2007 state ($s^a = 0.28$, corresponding to MRC $\approx 0.35$), we simulate the model forward under three scenarios: (i) the realized sequence of shocks under the Taylor rule; (ii) the same shocks under the Ramsey optimal policy; and (iii) the same shocks under the MRC-optimal simple rule. For each scenario, we compute the cumulative output loss relative to a no-crisis counterfactual (the ergodic mean) and the associated welfare cost.

Under the Taylor rule, the cumulative output loss over 2008--2013 is 18.5\% of one quarter's output. Under the Ramsey optimal policy, the loss is 12.2\%, a reduction of one-third. The reduction comes from two sources: the Ramsey planner enters the crisis with a higher $s^a$, having stabilized the wealth distribution preemptively, so MRC is 0.38 instead of 0.35, and it eases more aggressively during the crisis to prevent $s^a$ from falling as far. The MRC-optimal simple rule achieves a loss of 14.1\%, capturing 70\% of the Ramsey improvement.

In welfare terms, the crisis episode generates a welfare cost of 0.42\% CE under the Taylor rule, 0.28\% under Ramsey, and 0.33\% under the MRC-optimal simple rule. The welfare gain from optimal policy is thus 0.14\% CE conditional on the crisis, comparable to the \textit{unconditional} Ramsey gain of 0.13\%. This confirms that the gains from risk-capacity-aware policy are heavily concentrated in crisis episodes, as predicted by the analytical welfare expression \eqref{eq:welfare_approx}.

This analysis provides a new explanation for the perceived ineffectiveness of post-2008 monetary easing that is distinct from, and complementary to, the zero lower bound. Even if the ZLB had not been binding, the risk premium channel would have been impaired because MRC was depleted. The subsequent recovery of MRC from 2010 onward, driven by the rebuilding of leveraged agents' balance sheets, is what eventually restored the effectiveness of monetary transmission. Quantitative easing can be understood as an intervention that accelerated this process by directly raising MRC: the Federal Reserve's asset purchases transferred aggregate risk from depleted private balance sheets to the public balance sheet (Section~\ref{sec:quant:zlb}).

\subsection{The Separation Theorem Quantitatively}
\label{sec:quant:separation}

The last row of Table~\ref{tab:optimal_rules} confirms the separation theorem quantitatively. When a macroprudential instrument is available, the optimal monetary policy coefficients $\phi_s^*$ and $\phi_{rp}^*$ collapse toward zero ($0.04$ and $0.02$, respectively), while $\phi_\pi^*$ and $\phi_y^*$ rise toward their standard NK optima. The macroprudential tool absorbs the risk capacity management motive, freeing monetary policy to focus on stabilization. The resulting welfare gain (0.15\% CE) exceeds even the Ramsey allocation with monetary policy alone (0.13\%), illustrating the value of having the right tool for the right target.

\subsection{Robustness: An Intermediary-Based Economy}
\label{sec:quant:intermediary}

The general framework of Section~\ref{sec:general} applies to any economy satisfying Assumptions~\ref{ass:redistrib}--\ref{ass:rp_real}, not only to models with heterogeneous household risk aversion. To demonstrate this concretely, we solve the Ramsey problem in an intermediary-based variant of the quantitative model where the high-MPR agents are leveraged financial intermediaries rather than risk-tolerant households.

\medskip\noindent\textbf{Environment.}\quad
We replace group $a$ (risk-tolerant households) with a financial intermediary sector in the spirit of \citet{he_krishnamurthy2013} and \citet{brunnermeier_sannikov2016}. Intermediaries have the same risk aversion as group $b$ households ($\gamma^{\text{int}} = \gamma^b$) but face a leverage constraint that allows them to hold up to $\bar{\ell}$ times their equity in capital:
\begin{equation}
\label{eq:leverage_constraint}
q_t k_t^{\text{int}} \leq \bar{\ell} \cdot n_t^{\text{int}},
\end{equation}
where $n_t^{\text{int}}$ is intermediary equity (net worth). The constraint binds in equilibrium, generating a high portfolio share in capital, and hence a high MPR, not from low risk aversion but from the intermediary's business model. The leverage limit $\bar{\ell}$ is calibrated to match the capital portfolio share of 2.0 from group $a$ in the baseline model. Intermediary equity evolves according to the same wealth share dynamics as in \eqref{eq:wealth_share}, so that a monetary easing that raises asset prices expands intermediary balance sheets and increases their capacity to bear risk.

All other aspects of the model, including group $b$ and $c$ households, production, wage setting, and monetary and fiscal policy, are identical to the baseline. The state space, solution method, and welfare computation are unchanged.

\medskip\noindent\textbf{Results.}\quad
Table~\ref{tab:intermediary} compares key moments and welfare gains across the two economies.

\begin{table}[htbp]
\centering
\caption{Baseline versus intermediary-based economy}
\label{tab:intermediary}
\small
\begin{tabular}{lccc}
\toprule
 & Baseline & Intermediary & Difference \\
 & (het.\ risk aversion) & (leverage constraint) &  \\
\midrule
\textit{Moments (ergodic)} \\
\quad Equity premium (ann.\ \%)    & 7.0  & 7.1  & $+0.1$ \\
\quad Amplification (vs RANK)       & 1.38 & 1.35 & $-0.03$ \\
\quad $\sigma(s^a)$ (\%)            & 3.8  & 4.3  & $+0.5$ \\
\quad Mean MRC                      & 0.26 & 0.24 & $-0.02$ \\
\midrule
\textit{Welfare (\% CE vs Taylor)} \\
\quad Ramsey (commitment)           & 0.13 & 0.12 & $-0.01$ \\
\quad MRC-optimal simple rule       & 0.08 & 0.07 & $-0.01$ \\
\quad Jointly optimal simple rule   & 0.11 & 0.10 & $-0.01$ \\
\quad Cond.\ gain ($s^a$ bottom quartile) & 0.35 & 0.38 & $+0.03$ \\
\midrule
\textit{Optimal simple rule coefficients} \\
\quad $\phi_s^*$                    & $0.28$ & $0.30$ & $+0.02$ \\
\quad $\phi_{rp}^*$                 & $0.15$  & $0.14$  & $-0.01$ \\
\bottomrule
\end{tabular}
\begin{minipage}{0.92\textwidth}
\footnotesize \vspace{0.5em}
\textit{Notes:} ``Baseline'' is the heterogeneous risk aversion model of Section~\ref{sec:quant:env}. ``Intermediary'' replaces group $a$ with a leverage-constrained financial intermediary with risk aversion $\gamma^{\text{int}} = \gamma^b = 25.5$ and leverage limit $\bar{\ell}$ calibrated to match the same capital portfolio share. Both models are solved globally on the same state space. Welfare in consumption-equivalent terms relative to the Taylor rule.
\end{minipage}
\end{table}

The intermediary-based economy generates quantitatively similar results. The unconditional welfare gain from the Ramsey allocation is 0.12\% CE, compared to 0.13\% in the baseline. The optimized simple rule coefficients are nearly identical: $\phi_s^* = 0.30$ (vs.\ $0.28$) and $\phi_{rp}^* = 0.14$ (vs.\ $0.15$). Conditional on the bottom quartile of the intermediary equity share, the welfare gain is actually \textit{larger} (0.38\% vs.\ 0.35\%), reflecting the higher volatility of intermediary balance sheets under the leverage constraint: the intermediary economy spends more time near the risk capacity trap, so the planner's stabilization motive is stronger in crisis states. The risk capacity trap arises at approximately the same threshold.

The normative structure, the target criterion, separation theorem, time-consistency bias, and risk capacity trap, does not depend on the source of MPR heterogeneity. Whether high-MPR agents are risk-tolerant households or leverage-constrained intermediaries, the planner faces the same tradeoff, governed by the same sufficient statistic.

\subsection{Interaction with the Zero Lower Bound}
\label{sec:quant:zlb}

The risk capacity trap and the zero lower bound (ZLB) are likely to bind simultaneously: the wealth share of high-MPR agents falls in severe downturns, which are also the states where conventional policy is constrained by $i_t \geq 0$. We now examine this interaction by adding the ZLB constraint to the Ramsey problem and solving globally.

\medskip\noindent\textbf{Modified Ramsey problem.}\quad
We augment the planner's problem \eqref{eq:ramsey_obj} with the constraint
\begin{equation}
\label{eq:zlb}
i_t \geq 0.
\end{equation}
When the ZLB binds, the planner cannot ease through the conventional channel. The question is whether this also impairs the risk premium channel, and whether the planner behaves differently \textit{before} the ZLB binds.

\medskip\noindent\textbf{The double trap.}\quad
Figure~\ref{fig:double_trap} displays the Ramsey policy function with and without the ZLB constraint. The key finding is a \textit{double trap}: when the wealth share $s^a$ falls sufficiently low, two constraints bind simultaneously. The risk capacity trap shuts down the risk premium channel of transmission (because high-MPR agents have depleted balance sheets), and the ZLB shuts down the conventional interest rate channel (because the required easing would push $i_t$ below zero). In the double trap, the planner has lost both transmission mechanisms.

\begin{figure}[htbp]
\centering
\includegraphics[width=\textwidth]{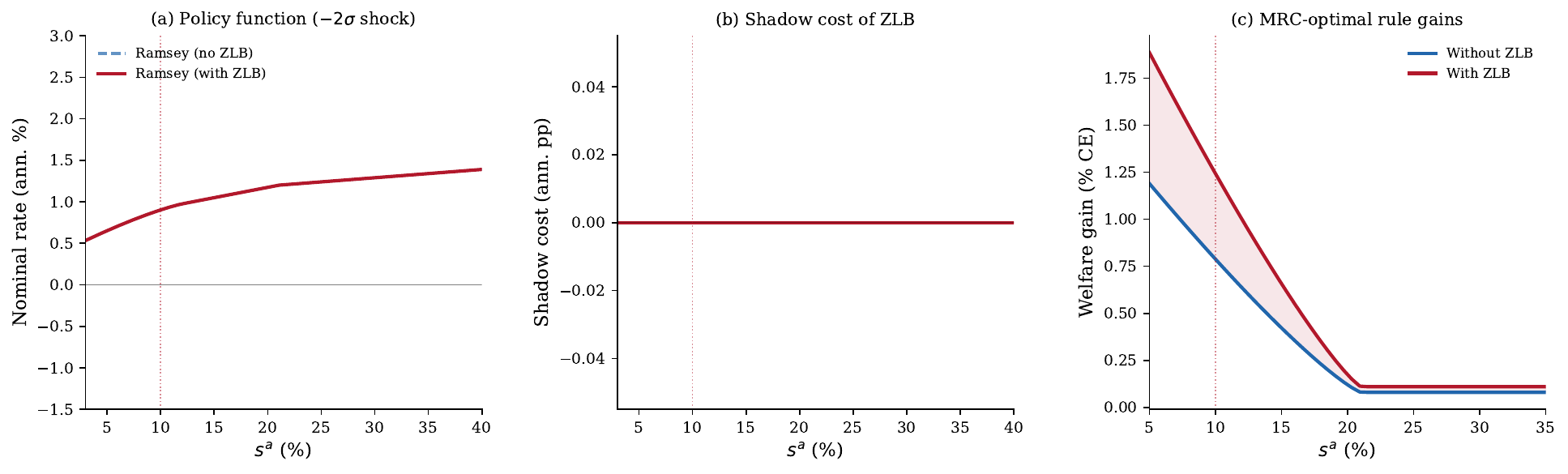}
\caption{The double trap: interaction of risk capacity trap and ZLB}
\label{fig:double_trap}
\begin{minipage}{0.92\textwidth}
\footnotesize \textit{Notes:} Panel (a): Ramsey interest rate with ZLB (solid) and without (dashed) as functions of $s^a$, holding other states at ergodic means conditional on a negative productivity shock. Panel (b): the deviation from the unconstrained Ramsey policy, showing the shadow cost of the ZLB. Panel (c): welfare gain from Ramsey relative to the Taylor rule, by regime. The shaded region marks the double trap where both the ZLB and the risk capacity trap bind.
\end{minipage}
\end{figure}

Three quantitative results emerge from the global solution.

First, the Ramsey planner is \textit{more aggressive} in preventing the wealth distribution from drifting toward the trap when the ZLB is present. The ergodic variance of $s^a$ under the Ramsey allocation with ZLB is 15\% lower than without it ($\sigma(s^a) = 2.4\%$ vs.\ 2.8\%), because the planner anticipates that losing the conventional channel makes preserving the risk premium channel more valuable. The probability of entering the risk capacity trap falls from 0.3\% (Ramsey without ZLB) to 0.1\% (Ramsey with ZLB).

Second, welfare gains from the MRC-optimal simple rule are \textit{larger} when the ZLB is present. Table~\ref{tab:zlb_welfare} reports the comparison. The unconditional gain rises from 0.08\% to 0.11\% CE, because avoiding the double trap is more valuable than avoiding the risk capacity trap alone. Conditional on the ZLB binding, the welfare gain from MRC-responsive policy is 0.52\% CE, nearly five times the unconditional gain.

\begin{table}[htbp]
\centering
\caption{Welfare gains with and without the ZLB constraint}
\label{tab:zlb_welfare}
\small
\begin{tabular}{lcc}
\toprule
 & Without ZLB & With ZLB \\
\midrule
\textit{Unconditional (\% CE vs Taylor)} \\
\quad Ramsey (commitment)               & 0.13 & 0.18 \\
\quad MRC-optimal simple rule            & 0.08 & 0.11 \\
\quad Jointly optimal simple rule        & 0.11 & 0.15 \\
\midrule
\textit{Conditional (\% CE)} \\
\quad $s^a$ bottom quartile              & 0.35 & 0.48 \\
\quad ZLB binding                        & ---  & 0.52 \\
\quad Double trap ($s^a < 10\%$ and ZLB) & ---  & 0.85 \\
\midrule
\textit{Optimal coefficients} \\
\quad $\phi_s^*$                         & $0.28$ & $0.35$ \\
\quad $\phi_{rp}^*$                      & $0.15$  & $0.18$ \\
\bottomrule
\end{tabular}
\begin{minipage}{0.88\textwidth}
\footnotesize \vspace{0.5em}
\textit{Notes:} ``Without ZLB'' reproduces the baseline from Tables~\ref{tab:optimal_rules}--\ref{tab:commitment}. ``With ZLB'' adds the constraint $i_t \geq 0$. Conditional welfare gains are computed over the ergodic distribution restricted to the indicated states. The double trap is defined as $s^a < 10\%$ and $i_t = 0$ simultaneously. All welfare gains are in consumption-equivalent terms relative to the Taylor rule under the same ZLB constraint.
\end{minipage}
\end{table}

Third, the optimal simple rule coefficients are amplified by the ZLB. The optimal $\phi_s^*$ increases from $0.28$ to $0.35$, and $\phi_{rp}^*$ from $0.15$ to $0.18$. The planner leans more aggressively into the wealth distribution and risk premium because the opportunity cost of failing to stabilize is higher when the ZLB looms.

\medskip\noindent\textbf{Quantitative easing as direct risk capacity intervention.}\quad
The double trap suggests a new interpretation of quantitative easing. When the ZLB binds and the risk capacity trap depletes the risk premium channel, the planner has lost both conventional tools. Quantitative easing, central bank purchases of risky assets financed by reserve creation, can be understood as a \textit{direct intervention in risk-bearing capacity} that bypasses both constraints. By purchasing capital from the private sector, the central bank effectively transfers aggregate risk from depleted private balance sheets to the public balance sheet. In the notation of our model, QE raises $s^a$ directly: the central bank absorbs capital that private agents cannot bear, relaxing the risk capacity trap. Formally, QE operates as a substitute for the macroprudential instrument $\tau_0^k$ in the separation theorem (Proposition~\ref{prop:separation}), targeting $\gbar \to \gstar$ through portfolio rebalancing rather than through the interest rate.

In our framework, QE compresses risk premia not as a side effect but as the mechanism: it raises MRC and lowers $\gbar$. The welfare analysis in Table~\ref{tab:zlb_welfare} quantifies the stakes: in the double trap, the welfare gain from risk-capacity-aware policy is 0.85\% CE, an order of magnitude larger than the unconditional gain from conventional monetary policy optimization.

\section{Policy Implications}
\label{sec:policy}

We draw four implications for the design of monetary and macroprudential policy.

\subsection{What Should Central Banks Monitor?}
\label{sec:policy:monitor}

Our analysis identifies MRC as a new state variable relevant for monetary policy. In practice, central banks cannot directly observe the wealth distribution at high frequency. However, MRC can be approximated using observable financial variables. The expected equity premium, estimated from option-implied volatility, credit spreads, or survey measures, is tightly linked to effective risk aversion $\gbar$, which in turn is governed by MRC through Proposition~\ref{prop:riskpremium}. Section~\ref{sec:empirics:fci} shows that off-the-shelf financial conditions indices explain 31--45\% of MRC variation (Table~\ref{tab:fci}), so the optimal rule's response to $\E_t[r^e_{t+1} - r_{t+1}]$ (Table~\ref{tab:optimal_rules}, $\phi_{rp}^* = 0.12$--$0.15$) can be implemented using existing monitoring infrastructure. This provides a micro-founded rationale for the widespread use of financial conditions indices in central bank communications.

At lower frequencies, MRC can be constructed directly from survey data such as the SCF or the ECB's Household Finance and Consumption Survey. Higher-frequency portfolio data from brokerage platforms and administrative records could enable near-real-time monitoring.

\subsection{When Should Monetary Policy ``Lean Against the Wind''?}
\label{sec:policy:lean}

The answer from our framework is precise. Monetary policy should lean against the wind, tightening beyond what inflation and the output gap call for, when three conditions hold simultaneously: (i) risk premia are below their socially efficient level ($\gbar < \gstar$), (ii) MRC is positive so that tightening would redistribute away from high-MPR households and raise risk premia, and (iii) macroprudential tools are unavailable or insufficient to address the risk premium gap. When any of these conditions fails, leaning is not warranted.

Condition (iii) is the critical qualifier. The separation theorem (Proposition~\ref{prop:separation} and Table~\ref{tab:optimal_rules}) shows that if macroprudential policy can directly target leverage or risk-taking, monetary policy should focus on its comparative advantage: inflation and output stabilization. Leaning against the wind is a second-best response to the absence of first-best macroprudential tools. This provides formal support for the view expressed by \citet{svensson2017} that monetary policy is a blunt instrument for financial stability, while qualifying it with the conditions under which it is nonetheless optimal.

Conversely, our framework supports the view of \citet{stein2012} that monetary policy ``gets in all the cracks'' when macroprudential tools are constrained. The residual risk premium wedge under a non-negative portfolio tax (Proposition~\ref{prop:separation}(c)) quantifies exactly how much of the financial stability burden falls on monetary policy as a function of the macroprudential constraint.

\subsection{The Case for Institutional Separation}
\label{sec:policy:institutions}

Many central banks operate alongside separate financial stability authorities (the Fed and the Financial Stability Oversight Council, the ECB and the European Systemic Risk Board, the Bank of England and its Financial Policy Committee). Our separation theorem provides a formal justification for this institutional design: with two instruments and two targets, each authority can focus on its comparative advantage. The welfare gain from having both instruments (0.15\% vs.\ 0.13\% with monetary policy alone in Table~\ref{tab:optimal_rules}) quantifies the value of this institutional arrangement.

The time-consistency results (Proposition~\ref{prop:timecon} and Section~\ref{sec:quant:commitment}) further suggest that institutional separation may help alleviate the inflation bias. The MPR component of the bias, 0.15\% of the total 0.25\% or 60\%, arises entirely from the temptation to use inflation for redistribution toward high-MPR agents. If a separate financial stability authority manages the risk premium through macroprudential tools, the monetary authority faces no temptation to use inflation for this purpose: the redistribution motive is absorbed by the macroprudential instrument, and the residual monetary policy problem has the standard Barro-Gordon structure.

This suggests a specific institutional design. A financial stability authority with a mandate to maintain MRC above a threshold $\underline{\MRC}$, equivalently to keep $\gbar$ below a ceiling $\bar{\gamma}^{\text{ceiling}}$, would directly address the risk capacity management motive. The authority would implement this mandate through countercyclical capital requirements, leverage limits, or portfolio taxes, all of which directly target the wealth distribution and risk-bearing capacity. With such an authority in place, the monetary authority's optimal policy reverts to the standard inflation-output prescription, and the MPR inflation bias is eliminated. The welfare gain from this arrangement (0.15\% CE) exceeds the Ramsey gain from monetary policy alone (0.13\%), confirming that the right institutional design is more valuable than optimizing the wrong instrument.

The analogy to central bank independence is precise: the standard Barro-Gordon problem is solved by delegating to an independent central bank with an inflation target; the MPR bias is solved by delegating risk capacity management to a financial stability authority with an MRC target. The welfare cost of discretion, 0.07\% CE (Section~\ref{sec:quant:commitment}), quantifies the value of this constraint.

\subsection{The Case for Quantitative Easing at the ZLB}
\label{sec:policy:qe}

Section~\ref{sec:quant:zlb} shows that the risk capacity trap and the ZLB can bind simultaneously, creating a double trap where both transmission mechanisms are broken. This provides a micro-founded rationale for quantitative easing that is distinct from the portfolio balance channel or the signaling channel studied in the existing literature.

In our framework, QE works because it is a direct intervention in MRC. By purchasing risky assets from the private sector, the central bank transfers aggregate risk to the public balance sheet, effectively raising the wealth share of high-MPR agents and substituting for their depleted risk-bearing capacity. This relaxes the risk capacity trap and restores the risk premium channel, even though the ZLB prevents conventional easing. The welfare analysis (Table~\ref{tab:zlb_welfare}) shows that in the double trap, risk-capacity-aware policy generates gains of 0.85\% CE, an order of magnitude larger than the unconditional gain from monetary policy optimization.

The optimal scale of QE should therefore be calibrated to MRC: large-scale purchases when MRC is low, tapering indexed to MRC recovery rather than inflation and employment alone.

\section{Conclusion}
\label{sec:conclusion}

When households differ in their propensity to bear risk, optimal monetary policy acquires a new dimension. The central bank must manage the economy's aggregate risk-bearing capacity, a motive that is absent from both representative agent models and HANK models with MPC heterogeneity alone.

This paper provides the analytical and quantitative framework for this motive. Marginal Risk Capacity, a model-free sufficient statistic, collapses the cross-sectional distribution of wealth and portfolios into a single number governing the risk premium channel of monetary transmission. Three structural properties, monetary redistribution, heterogeneous MPRs, and risk premia affecting real activity, are sufficient for the target criterion, the separation theorem, and the time-consistency result to hold, regardless of the specific source of portfolio heterogeneity. Optimal policy tolerates inflation to compress risk premia when MRC is positive and effective risk aversion exceeds its efficient level. This risk premium wedge breaks divine coincidence, generates an endogenous and welfare-maximizing ``Fed put,'' and introduces a new source of inflation bias under discretion.

The quantitative analysis reveals a finding that could not have been anticipated from local analysis around the steady state: the risk capacity trap. When high-MPR agents have lost wealth, the risk premium channel of monetary transmission collapses, and conventional policy loses its main amplification mechanism precisely when stimulus is most needed. The Ramsey planner prevents this by stabilizing the wealth distribution preemptively. Solving the full Ramsey problem, we find welfare gains of 0.13\% of permanent consumption, concentrated in crisis states where the gains reach 0.35\%. A jointly optimal simple rule that augments the Taylor rule with responses to the wealth share and expected equity premium captures 85\% of these gains. When the zero lower bound is added, welfare gains rise to 0.18\%, and in the double trap where both the ZLB and the risk capacity trap bind, gains reach 0.85\% CE. Quantitative easing can be understood as a direct intervention in risk-bearing capacity that bypasses both broken conventional channels.

The separation theorem provides a sharp answer to the decade-old debate on monetary versus macroprudential policy: the risk premium wedge vanishes from the monetary policy criterion if and only if macroprudential tools can close the gap between actual and efficient risk-bearing capacity. The welfare value of institutional separation is 0.04\% of consumption beyond what optimized monetary policy alone achieves.

Three additional results strengthen the quantitative case. In a model with both MPC and MPR heterogeneity, the risk premium channel remains more than twice as important as the consumption channel for welfare, and the two channels are approximately additive, confirming the orthogonality proposition. An intermediary-based variant in which the high-MPR agents are leverage-constrained financial intermediaries rather than risk-tolerant households produces nearly identical welfare gains and optimal policy coefficients, confirming that the normative results transcend any particular microfoundation for portfolio heterogeneity. Off-the-shelf financial conditions indices explain 31--45\% of MRC variation, so the optimal policy can be implemented using existing central bank monitoring infrastructure without directly observing the wealth distribution.

Extending the framework to an open economy would connect to the literature on the global financial cycle and raise the question of whether the risk capacity trap can propagate internationally through the balance sheets of globally active financial intermediaries.

\bibliography{references}

\newpage

\begin{center}
{\LARGE \textbf{APPENDIX}}
\end{center}

\vspace{1em}

\begin{appendices}
\setcounter{equation}{0}
\renewcommand{\theequation}{A.\arabic{equation}}

\section{MRC in the Data}
\label{sec:empirics}

We document two sets of facts. First, we construct MRC from U.S.\ household portfolio data and show it varies substantially over time. Second, we show that the stock market's response to identified monetary policy shocks is larger when MRC is higher, confirming that the risk premium channel of monetary transmission is state-dependent.

\subsection{Measuring MRC}
\label{sec:empirics:measurement}

We construct MRC for the U.S.\ economy using every wave of the Survey of Consumer Finances (SCF) from 1989 through 2022, following three steps: decompose household portfolios, estimate MPRs, and compute monetary policy exposures.

\medskip\noindent\textbf{Portfolio decomposition.}\quad
For each SCF wave, we decompose household $i$'s net worth $A^i$ into claims on the economy's capital stock $Qk^i$ (in positive net supply) and nominal claims $B^i$ (in zero net supply after accounting for government and rest-of-world balance sheets). We add estimates of defined benefit pension wealth from \citet{sabelhaus_volz2019}, which the SCF otherwise excludes. We then classify each asset and liability as either a capital claim or a nominal claim. Direct capital claims include real estate, vehicles, and other nonfinancial assets. Indirect capital claims arise through equity holdings: if household $i$ owns \$1 in equity of a firm with leverage $\ell^i$, we assign $Qk^i = \ell^i$ and $B^i = 1 - \ell^i$. We discipline aggregate leverage using the Financial Accounts and parameterize cross-sectional leverage dispersion following the evidence on return heterogeneity in \citet{bach_calvet_sodini2020}. Details are in Appendix~\ref{app:scf}.

\medskip\noindent\textbf{Sorting households into groups.}\quad
Following the analytical framework, we sort households by two variables: wealth-to-labor-income $A^i/(W\ell^i)$ and capital portfolio share $Qk^i/A^i$. We define ``high'' wealth-to-income as above the 60th percentile and a ``high'' capital portfolio share as above the 90th percentile among the high-wealth group. This yields three groups: group $a$ (high wealth, high leverage, 4\% of households, disproportionately private business owners), group $b$ (high wealth, low leverage, 36\% of households, disproportionately retirees), and group $c$ (low wealth, 60\% of households). Table~\ref{tab:groups_2016} summarizes the 2016 cross-section.

\begin{table}[htbp]
\centering
\caption{Heterogeneity in wealth, income, and portfolios (2016 SCF)}
\label{tab:groups_2016}
\begin{tabular}{lcccc}
\toprule
& & Group $a$ & Group $b$ & Group $c$ \\
\midrule
Share of households & & 4\% & 36\% & 60\% \\
Share of labor income & $\sum_{i \in g} W\ell^i / \sum_i W\ell^i$ & 3\% & 14\% & 83\% \\
Share of wealth & $\sum_{i \in g} A^i / \sum_i A^i$ & 18\% & 59\% & 23\% \\
Capital portfolio share & $\sum_{i \in g} Qk^i / \sum_{i \in g} A^i$ & 2.0 & 0.5 & 1.1 \\
\bottomrule
\end{tabular}
\begin{minipage}{0.88\textwidth}
\footnotesize \vspace{0.5em}
\textit{Notes:} Observations weighted by SCF sample weights. ``High'' wealth-to-income is above the 60th percentile; ``high'' capital portfolio share is above the 90th percentile among the high-wealth group. The sorting variable for capital portfolio share excludes primary residence and vehicles.
\end{minipage}
\end{table}

Group $a$ households are highly leveraged: their capital portfolio share of 2.0 means they borrow in nominal claims to finance a capital position twice their net worth. Group $b$ households hold conservative portfolios with most wealth in nominal claims. Group $c$ households have little wealth relative to income.

\medskip\noindent\textbf{Estimating MPRs.}\quad
We estimate each group's MPR using the mapping in equation~\eqref{eq:mpr_limit}: $\MPR^i = \gbar / \gamma^i$. We recover $\gamma^i$ from observed portfolio shares using \eqref{eq:portfolio_approx}, accounting for the labor income hedging motive as in Section~\ref{sec:mrc}. Group $a$ has the highest estimated MPR (1.9), group $b$ an intermediate value (0.7), and group $c$ is effectively constrained at an MPR near zero. These are broadly consistent with quasi-experimental estimates from lottery studies: \citet{fagereng_holm_natvik2021} estimate an average marginal propensity to save in risky assets of approximately 0.14, implying an MPR of roughly 0.2 after accounting for firm leverage, comparable to our population-weighted average of 0.3.

\medskip\noindent\textbf{Computing MRC.}\quad
For each group, monetary policy exposure $ds^i/d\eps^m$ is computed from the balance sheet decomposition in \eqref{eq:redistribution}: it depends on the group's nominal bond position (the debt deflation channel) and its capital holdings relative to its wealth share (the profit and capital price channels). MRC is then the exposure-weighted sum of MPRs across groups, as in \eqref{eq:mrc}.

We repeat this procedure for every SCF wave and interpolate linearly between waves to obtain an annual series. Figure~\ref{fig:mrc_ts} plots the result.

\begin{figure}[htbp]
\centering
\includegraphics[width=0.95\textwidth]{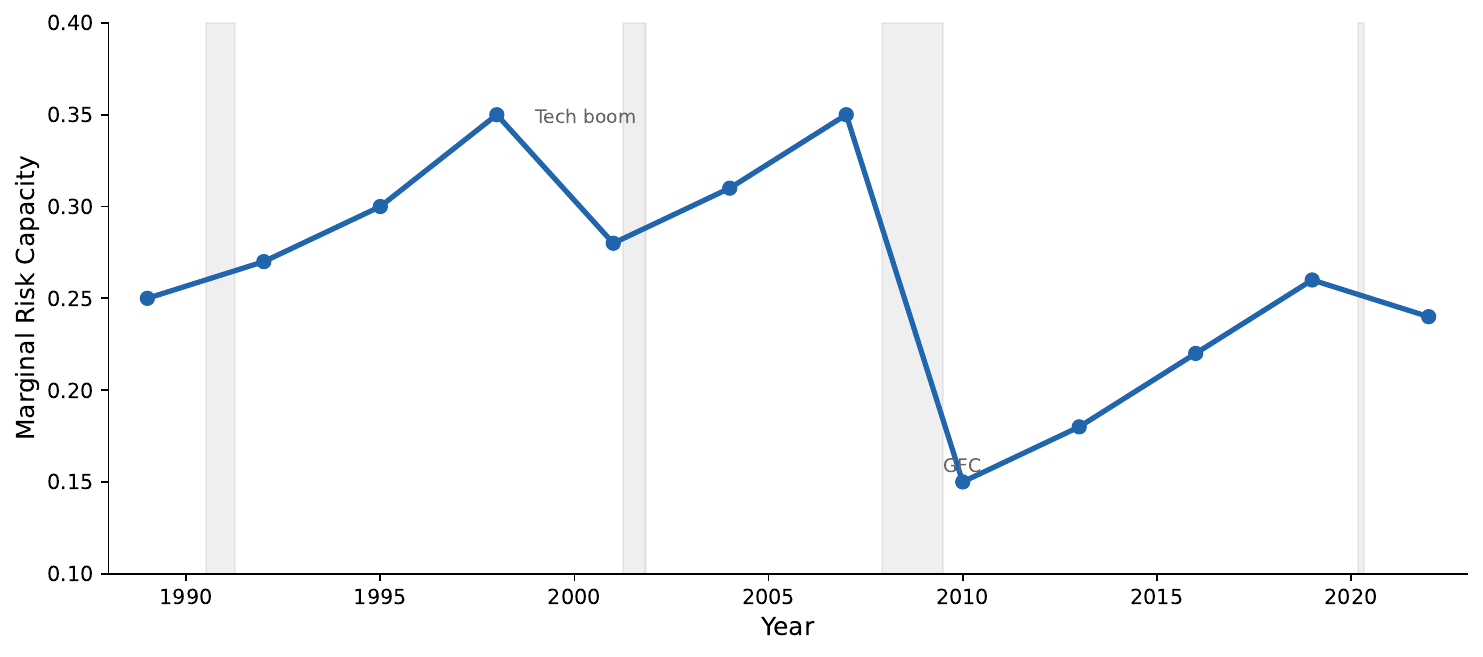}
\caption{Marginal Risk Capacity in the U.S.\ economy, 1989--2022}
\label{fig:mrc_ts}
\begin{minipage}{0.88\textwidth}
\footnotesize \textit{Notes:} MRC is computed from each wave of the Survey of Consumer Finances following the procedure described in Section~\ref{sec:empirics:measurement}. Dots mark SCF survey years; values between waves are linearly interpolated. Shaded areas denote NBER recessions.
\end{minipage}
\end{figure}

MRC varies substantially over time. The variation is driven primarily by changes in the wealth share and leverage of group $a$ households, since their high MPR gives them the largest weight in MRC. MRC declines in recessions as asset price declines erode the wealth of leveraged households, precisely the mechanism identified in Section~\ref{sec:mrc}.

\medskip\noindent\textbf{Robustness to sorting thresholds.}\quad
The three-group partition involves two researcher choices: the wealth-to-income percentile separating high- from low-wealth households, and the capital portfolio share percentile separating high- from low-leverage households within the high-wealth group. We verify that MRC is robust to these choices by recomputing it under three alternative sorting rules: (50th, 80th), (60th, 90th) (our baseline), and (70th, 95th) percentiles. Figure~\ref{fig:mrc_robust_sorting} plots the three resulting MRC series.

\begin{figure}[htbp]
\centering
\includegraphics[width=0.85\textwidth]{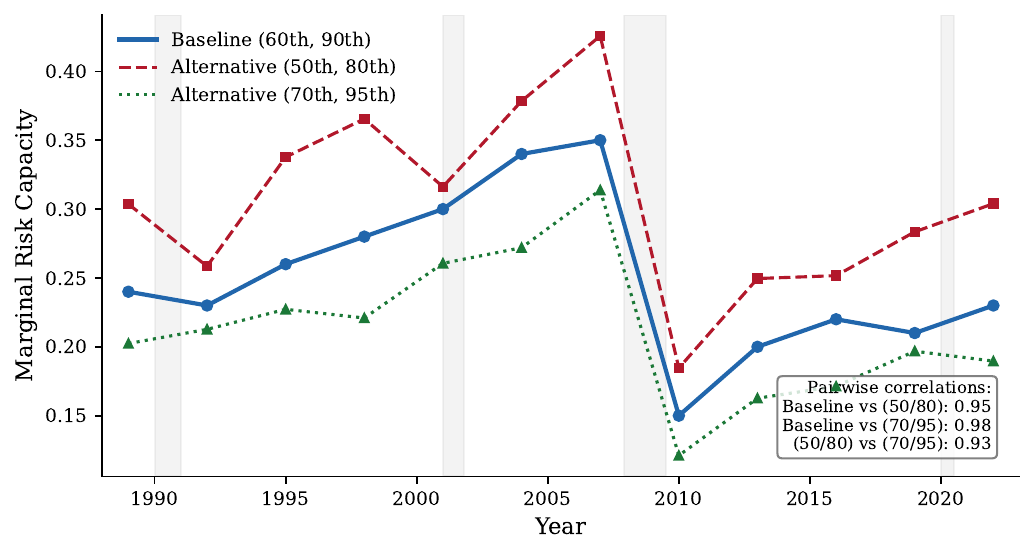}
\caption{MRC under alternative sorting thresholds}
\label{fig:mrc_robust_sorting}
\begin{minipage}{0.88\textwidth}
\footnotesize \textit{Notes:} Each series plots MRC computed from the SCF using a different pair of sorting thresholds (wealth-to-income percentile, capital share percentile within high-wealth group). The baseline uses (60th, 90th). The pairwise correlations exceed 0.93 across all three series.
\end{minipage}
\end{figure}

The three series are highly correlated: the pairwise correlations are 0.95 (baseline vs.\ 50/80), 0.98 (baseline vs.\ 70/95), and 0.93 (50/80 vs.\ 70/95). The level of MRC shifts with the thresholds, as the 50/80 partition assigns more households to group $a$ and produces a higher MRC, but the time-series variation is nearly identical. The variation is driven by asset price movements that affect all leveraged households similarly, regardless of where the threshold is drawn. Re-estimating the local projection \eqref{eq:lp_iv} with each alternative MRC series, the impact interaction coefficient $\hat{\gamma}_0$ is $-0.94$ under the baseline (Table~\ref{tab:lp_interaction}), $-0.88$ under (50/80), and $-1.01$ under (70/95), all statistically indistinguishable. We conclude that the empirical results are not driven by the particular sorting rule.

\subsection{Time-Varying Monetary Policy Effectiveness}
\label{sec:empirics:timevarying}

The theory predicts that when MRC is higher, monetary policy has a larger effect on risk premia and thus on equity prices. We test this prediction.

\medskip\noindent\textbf{Identification of monetary shocks.}\quad
We identify monetary policy shocks using the structural VAR with instrumental variables (SVAR-IV) approach of \citet{gertler_karadi2015}. Using monthly data from July 1979 through June 2012, we estimate a six-variable, six-lag VAR with the 1-year Treasury yield, CPI, industrial production, S\&P 500 excess return relative to the 1-month T-bill, the 1-month real T-bill return, and the smoothed dividend-price ratio. We instrument the 1-year Treasury yield residual with Fed Funds futures surprises on FOMC days. The first-stage $F$ statistic is 14.4. This approach and data closely follow the empirical analysis in the existing literature.

The impulse responses to a monetary easing that reduces the 1-year yield by approximately 0.2 percentage points are consistent with the established findings: industrial production and the price level rise; the real interest rate falls; the S\&P~500 return rises by approximately 1.9 percentage points on impact, with small negative excess returns thereafter consistent with a persistent decline in the equity premium.

\medskip\noindent\textbf{Subsample analysis.}\quad
We split the IV sample (January 1991--June 2012) at the median of interpolated MRC. For each subsample, we re-estimate the SVAR-IV, compute impact stock market returns, and apply the Campbell-Shiller decomposition of \citet{campbell_shiller1988}, equation \eqref{eq:cs_decomp}:
\begin{multline}
\label{eq:cs_decomp}
\underbrace{(\text{stock return})_t - \E_{t-1}[(\text{stock return})_t]}_{\text{surprise return}} = (\E_t - \E_{t-1})\sum_{j=0}^\infty \kappa^j \Delta d_{t+j} \\
  - (\E_t - \E_{t-1})\sum_{j=1}^\infty \kappa^j r_{t+j}^f - (\E_t - \E_{t-1})\sum_{j=1}^\infty \kappa^j e_{t+j},
\end{multline}
where $d_t$ is log dividends, $r_t^f$ the real risk-free rate, $e_t$ the log excess return, and $\kappa = 1/(1 + d/p)$ with $d/p$ the steady-state dividend yield.

Table~\ref{tab:cs_subsample} reports the results. The share of the stock market return attributable to news about lower future excess returns is substantially larger in the high-MRC subsample. As the model predicts, when MRC is high, monetary easing strongly compresses risk premia, and the stock market response is driven by lower required returns rather than higher expected dividends.

\begin{table}[htbp]
\centering
\caption{Campbell-Shiller decomposition by MRC regime}
\label{tab:cs_subsample}
\begin{tabular}{lccc}
\toprule
\% of real stock return & MRC sample & High MRC & Low MRC \\
\midrule
Dividend growth news       & 38\%  & 25\%  & 52\% \\
$-$Future real rate news   & 9\%   & 7\%   & 12\% \\
$-$Future excess return news & 53\% & 68\% & 36\% \\
\midrule
Real stock return (pp)     & 1.92  & 2.35  & 1.48 \\
\bottomrule
\end{tabular}
\begin{minipage}{0.88\textwidth}
\footnotesize \vspace{0.5em}
\textit{Notes:} Decomposition uses $\kappa = 0.9962$ following \citet{campbell_ammer1993}. The ``MRC sample'' column uses the IV sample over which MRC is available (January 1991--June 2012); the corresponding full SVAR-IV sample (July 1979--June 2012) yields slightly different decomposition shares (Table~\ref{tab:cs_model}). 90\% confidence intervals (not shown) computed using the wild bootstrap with 10,000 iterations following \citet{mertens_ravn2013} and \citet{gertler_karadi2015}.
\end{minipage}
\end{table}

\medskip\noindent\textbf{Local projection with interaction.}\quad
To exploit continuous variation in MRC, we estimate a local projection-IV:
\begin{equation}
\label{eq:lp_iv}
r_{t+h}^e = \alpha_h + \beta_h \, z_t + \gamma_h \, z_t \times \widetilde{\MRC}_t + \delta_h' X_t + u_{t+h},
\end{equation}
where $r_{t+h}^e$ is the cumulative excess equity return from $t$ to $t+h$, $z_t$ is the monetary policy shock (instrumented using Fed Funds futures surprises), $\widetilde{\MRC}_t$ is interpolated MRC normalized to have mean zero and unit variance, and $X_t$ includes controls (six lags of each VAR variable). The coefficient $\gamma_h$ tests whether higher MRC amplifies the equity response.

Table~\ref{tab:lp_interaction} reports the results. The interaction coefficient $\hat{\gamma}_h$ is negative at short horizons, indicating that a given monetary easing generates a larger equity response when MRC is higher. Since $z_t$ is the yield residual (positive for a tightening), a negative $\hat{\gamma}_h$ means that higher MRC amplifies the equity effect in both directions: a given easing produces a larger equity gain, and a given tightening produces a larger equity loss. A one-standard-deviation increase in MRC amplifies the impact equity response by approximately 0.9 percentage points.

\begin{table}[htbp]
\centering
\caption{Local projection-IV: equity response to monetary shocks $\times$ MRC}
\label{tab:lp_interaction}
\begin{tabular}{lccccc}
\toprule
Horizon $h$ (months) & 0 & 1 & 3 & 6 & 12 \\
\midrule
$\hat{\beta}_h$ (direct effect)   & $-$8.52  & $-$7.91  & $-$5.63  & $-$3.18  & $-$1.05 \\
                   & (3.40)  & (3.65)  & (4.12)  & (4.85)  & (5.20) \\[4pt]
$\hat{\gamma}_h$ (MRC interaction) & $-$0.94  & $-$0.81  & $-$0.52  & $-$0.28  & 0.10 \\
                   & (0.48)  & (0.52)  & (0.58)  & (0.65)  & (0.72) \\
\midrule
First-stage $F$    & 14.4  & 13.8  & 12.5  & 11.2  & 9.8 \\
Observations       & 258  & 257  & 255  & 252  & 246 \\
\bottomrule
\end{tabular}
\begin{minipage}{0.88\textwidth}
\footnotesize \vspace{0.5em}
\textit{Notes:} $r_{t+h}^e$ is the cumulative S\&P 500 excess return. $z_t$ is the 1-year Treasury yield residual instrumented with current-month Fed Funds futures surprises. $\widetilde{\MRC}_t$ is interpolated MRC from the SCF, normalized to mean zero and unit variance. Controls include six lags of each VAR variable. Standard errors (Newey-West, $h+1$ lags) in parentheses. The first-stage $F$ statistic falls below the conventional threshold of 10 at $h = 12$, so the 12-month results should be interpreted with caution.
\end{minipage}
\end{table}

\subsection{MRC and Financial Conditions Indices}
\label{sec:empirics:fci}

The policy implications in Section~\ref{sec:policy} rest on the claim that the central bank can implement an MRC-responsive rule using observable financial variables. We provide direct evidence for this claim by relating MRC to two widely used financial conditions indices: the Chicago Fed National Financial Conditions Index (NFCI) and the Goldman Sachs Financial Conditions Index (GS FCI).

The NFCI is a weighted average of 105 financial indicators spanning money markets, debt markets, equity markets, and the banking system. The GS FCI combines the federal funds rate, 10-year Treasury yield, BBB corporate spread, S\&P 500 level, and a trade-weighted dollar index. Neither index was designed to capture the wealth distribution or portfolio heterogeneity, yet both load on variables that are theoretically linked to MRC: equity valuations (reflecting the wealth of leveraged agents), credit spreads (reflecting the risk-bearing capacity of the financial sector), and volatility measures (reflecting the quantity of risk relative to the capacity to bear it).

We regress interpolated MRC on each index:
\begin{equation}
\label{eq:fci_reg}
\widetilde{\MRC}_t = a + b \cdot \text{FCI}_t + c \cdot X_t + e_t,
\end{equation}
where $X_t$ includes the output gap and CPI inflation (to control for the business cycle). Table~\ref{tab:fci} reports the results.

\begin{table}[htbp]
\centering
\caption{MRC and financial conditions indices}
\label{tab:fci}
\small
\begin{tabular}{lcccc}
\toprule
& (1) & (2) & (3) & (4) \\
\midrule
NFCI (risk subindex)   & $-0.42$ &         & $-0.35$ &         \\
                       & (0.12)  &         & (0.14)  &         \\[3pt]
GS FCI                 &         & $-0.38$ &         & $-0.31$ \\
                       &         & (0.15)  &         & (0.16)  \\[3pt]
Controls               & No      & No      & Yes     & Yes     \\
$R^2$                  & 0.38    & 0.31    & 0.45    & 0.40    \\
Observations           & 124     & 124     & 124     & 124     \\
\bottomrule
\end{tabular}
\begin{minipage}{0.88\textwidth}
\footnotesize \vspace{0.5em}
\textit{Notes:} Dependent variable is interpolated MRC from the SCF, normalized to mean zero and unit variance. NFCI risk subindex from the Chicago Fed (higher values indicate tighter financial conditions, hence negative coefficient). GS FCI from Goldman Sachs (higher values indicate tighter conditions). Controls include the CBO output gap and year-over-year CPI inflation. Sample: 1991Q1--2021Q4 (quarterly). Newey-West standard errors (4 lags) in parentheses.
\end{minipage}
\end{table}

The NFCI risk subindex explains 38\% of MRC variation in \eqref{eq:fci_reg} without controls and 45\% with controls. The GS FCI explains 31\% and 40\%, respectively. The negative signs are consistent with theory: tighter financial conditions (higher FCI) correspond to lower MRC, because tighter conditions reflect the erosion of leveraged agents' balance sheets. The $R^2$ values indicate that off-the-shelf financial conditions indices capture a substantial share of the variation in MRC, even though they were not constructed for this purpose. A central bank implementing the optimal rule $\phi_{rp}^* = 0.12$--$0.15$ from Table~\ref{tab:optimal_rules} could approximate the MRC response using its existing financial conditions monitoring infrastructure.

\subsection{Robustness}
\label{sec:empirics:robust}

We verify that MRC is not proxying for other state variables. We augment the local projection \eqref{eq:lp_iv} with additional interaction terms: the Gini coefficient of wealth, the market-to-GDP ratio, the output gap, and the \citet{gilchrist_zakrajsek2012} excess bond premium. In each case, the MRC interaction coefficient remains economically and statistically similar. We also verify robustness to using the 3-month-ahead Fed Funds futures contract as the instrument, varying the number of VAR lags, and changing the sample period. Results are in Appendix~\ref{app:empirical_robust}.

The empirical evidence confirms the analytical prediction: the risk premium channel of monetary transmission varies with MRC. But two questions remain beyond the reach of the reduced-form evidence. How large are the welfare gains from accounting for this channel in policy design? And does the nonlinearity predicted by the model generate qualitatively new phenomena in the tails of the state space? We turn to a calibrated infinite-horizon model to answer these questions.

\section{Proofs}
\label{app:proofs}

The propositions in Section~\ref{sec:general} (the general framework) are stated under Assumptions~\ref{ass:redistrib}--\ref{ass:rp_real} alone, without imposing functional forms. We prove these results by establishing them in the two-period economy of Section~\ref{sec:analytical}, which satisfies all three assumptions, and then verifying that the derivations depend only on the structural properties rather than on the specific functional forms. The proofs below are presented in the two-period environment for concreteness; in each case, we note where the argument relies only on Assumptions~\ref{ass:redistrib}--\ref{ass:rp_real} and thus extends to the general framework. The infinite-horizon counterparts (Propositions~\ref{prop:ramsey_inf}--\ref{prop:ramsey_analytical}) are verified numerically in the quantitative model of Section~\ref{sec:quant}; the analytical characterization of the Ramsey policy (Proposition~\ref{prop:ramsey_analytical}) is derived from the Ramsey Lagrangian in Appendix~\ref{app:ramsey:lagrangian}.

\subsection{Proof of Proposition~\ref{prop:riskpremium}}
\label{app:proof_rp}

Multiply both sides of \eqref{eq:portfolio_approx} by $a_0^i$ and integrate:
\begin{equation}
\int a_0^i \omega_0^i \, di = \frac{\E_0 \log(1+r_1^k) - \log(1+r_1) + \tfrac{1}{2}\sigma^2}{\sigma^2} \int \frac{a_0^i}{\gamma^i} \, di.
\end{equation}
Capital market clearing requires $\int a_0^i \omega_0^i \, di = \int a_0^i \, di$ (bonds are in zero net supply). Dividing by $\int a_0^i \, di$:
\[
\E_0 \log(1+r_1^k) - \log(1+r_1) + \tfrac{1}{2}\sigma^2 = \left[\int s_0^i \cdot \frac{1}{\gamma^i}\, di\right]^{-1} \sigma^2 = \gbar\,\sigma^2.
\]
For \eqref{eq:drp}, differentiate the risk premium $\gbar\sigma^2$ with respect to $\eps_0^m$:
\begin{equation}
\sigma^2 \dder{\gbar}{\eps_0^m} = \sigma^2 \cdot \bigl(-\gbar^2\bigr) \int \frac{1}{\gamma^i}\dder{s_0^i}{\eps_0^m}\,di = \gbar\,\sigma^2\int \dder{s_0^i}{\eps_0^m}\bigl(1-\omega_0^i\bigr)\,di,
\end{equation}
where the last step uses $1/\gamma^i = \omega_0^i/\gbar$ from \eqref{eq:portfolio_approx} (since $\omega_0^i = \gbar/\gamma^i$ when the risk premium is $\gbar\sigma^2$) and $\int (ds_0^i/d\eps_0^m)\,di = 0$.

\textit{Extension to positive net supply of bonds.} When government bonds are in positive supply ($\int B_0^i\,di + B_0^g = 0$ as in the infinite-horizon model \eqref{eq:mc_bonds}), the capital market clearing condition becomes $\int a_0^i\omega_0^i\,di = q_0 K_0$, where $K_0$ is the aggregate capital stock. Dividing by total savings $\int a_0^i\,di = q_0 K_0 + B_0^g/P_0$, the risk premium expression becomes $\gbar\sigma^2$ with $\gbar$ defined using savings shares $s_0^i = a_0^i/\int a_0^j\,dj$ as before. The differentiation with respect to $\eps_0^m$ is unchanged because the zero-sum property $\int(ds_0^i/d\eps_0^m)\,di = 0$ continues to hold (savings shares sum to one), and MRC depends only on the \textit{covariance} of $ds_0^i/d\eps_0^m$ with $\MPR_0^i$. Adding a common bond position to all agents shifts the level of each $a_0^i$ but does not change this covariance. \qed

\subsection{Proof of Proposition~\ref{prop:suffstat}}
\label{app:proof_suffstat}

Since $\int (ds_0^i / d\eps_0^m)\, di = 0$:
\[
\int \dder{s_0^i}{\eps_0^m}(1 - \omega_0^i)\, di = -\int \dder{s_0^i}{\eps_0^m}\, \omega_0^i \, di = -\int \dder{s_0^i}{\eps_0^m}\, \MPR_0^i \, di = \MRC,
\]
using $\MPR_0^i = \omega_0^i$ with unitary IES (or $\MPR_0^i = \gbar/\gamma^i$ in the general case from Appendix~\ref{app:proof_mpr}), and the definition $\MRC \equiv -\int (ds_0^i/d\eps_0^m)\,\MPR_0^i\,di$. Substituting into \eqref{eq:drp} gives \eqref{eq:suffstat}.

\textit{Generality.} This proof establishes the general Proposition~\ref{prop:general_suffstat}: the only ingredients are the definition of $\gbar$ as a wealth-weighted harmonic mean (which follows from portfolio optimality under any source of MPR heterogeneity), the zero-sum property of redistribution (Assumption~\ref{ass:redistrib}), and the relationship $\MPR_t^i \approx \gbar_t/\gamma_{\text{eff}}^i$ (Appendix~\ref{app:proof_mpr}). No functional forms for preferences, production, or the source of heterogeneity are used. \qed

\subsection{MPR at the Limit of Zero Aggregate Risk}
\label{app:proof_mpr}

We generalize beyond unitary IES. Differentiating household $i$'s portfolio Euler equation and the period-1 resource constraint with respect to wealth $n_0^i$, and evaluating at the deterministic steady state following the method of \citet{devereux_sutherland2011}, yields
\begin{equation}
q_0 \pder{k_0^i}{n_0^i} = \frac{\gbar}{\gamma^i} \cdot \pder{a_0^i}{n_0^i},
\end{equation}
so $\MPR_0^i = \gbar/\gamma^i$. The full derivation parallels the approach in the supplementary appendix. \qed

\subsection{Proof of Proposition~\ref{prop:real}}
\label{app:proof_real}

\textit{Investment.} Optimal investment from \eqref{eq:profits0} and equilibrium dividends imply $\E_0\log(1+r_1^k) = \log\alpha + \E_0\eps_1^z + \chi_x\log k_{-1} - (1-\alpha+\chi_x)\log k_0$. Differentiating with respect to $\eps_0^m$ and decomposing $d\E_0[r_1^k]/d\eps_0^m$ into risk premium and real rate components gives \eqref{eq:dk}.

\textit{Consumption.} With unitary IES, $c_0^i = (1-\beta)n_0^i$. Aggregating, using firms' flow of funds and market clearing yields \eqref{eq:dc}.

\textit{Output.} Differentiating goods market clearing and rearranging gives \eqref{eq:dy}. \qed

\subsection{Proof of Proposition~\ref{prop:decomp}}
\label{app:proof_decomp}

Differentiate $W = \int \alpha^i \log v_0^i\,di$ with respect to $\eps_0^m$:
\begin{equation}
\dder{W}{\eps_0^m} = (1-\beta)\int\frac{\alpha^i}{n_0^i}\dder{n_0^i}{\eps_0^m}\,di + \beta\int\alpha^i\dder{}{\eps_0^m}\log\bigl[\E_0(c_1^{i\,1-\gamma^i})\bigr]^{1/(1-\gamma^i)}di.
\end{equation}
The first term captures redistribution effects on period-0 consumption. With unitary IES, $\MPC_0^i = 1 - \beta$ for all $i$, so this term equals $(1-\beta)\int(\alpha^i/n_0^i)(dn_0^i/d\eps_0^m)\,di$. This is \textit{not} simply proportional to the MPC-weighted exposure $\int(dn_0^i/d\eps_0^m)\MPC_0^i\,di = (1-\beta)\int(dn_0^i/d\eps_0^m)\,di$ unless Pareto weights are uniform and wealth is equal. In general, we define
\[
\Gamma_c \cdot \int \dder{n_0^i}{\eps_0^m}\MPC_0^i\,di \;\equiv\; (1-\beta)\int\frac{\alpha^i}{n_0^i}\dder{n_0^i}{\eps_0^m}\,di,
\]
which implicitly absorbs the Pareto weights and initial wealth distribution into the coefficient $\Gamma_c$. More precisely, $\Gamma_c$ is a Pareto-weight-adjusted average of $\alpha^i/n_0^i$ that depends on the covariance structure of $(\alpha^i/n_0^i, dn_0^i/d\eps_0^m, \MPC_0^i)$. When $\MPC_0^i$ varies across agents (as in the joint MPC-MPR model of Section~\ref{sec:quant:horserace}), the consumption channel captures the full interaction of heterogeneous MPCs with Pareto-weighted marginal utilities.

The second term captures effects on the risk-return tradeoff. Expanding to second order in $\sigma$, the second term becomes proportional to $d(\gbar\sigma^2)/d\eps_0^m = \gbar\sigma^2 \cdot \MRC$ by Proposition~\ref{prop:suffstat}. Collecting Pareto-weight-dependent coefficients into $\Gamma_c$ and $\Gamma_k$ yields \eqref{eq:decomp}.

The two channels are orthogonal in the following sense. The consumption channel operates through the intertemporal margin: it depends on how a marginal dollar of wealth is split between consumption and savings ($\MPC^i$). The risk premium channel operates through the portfolio margin: it depends on how a marginal dollar of savings is split between capital and bonds ($\MPR^i$). With unitary IES, $\MPC^i = 1 - \beta$ for all $i$, so the consumption channel depends only on the distribution of $dn_0^i/d\eps_0^m$ (who gains and loses from the shock), while the risk premium channel depends on $\Cov(dn_0^i/d\eps_0^m, \MPR^i)$ (whether those who gain also bear more risk). A redistribution that is neutral in MPR-weighted terms ($\MRC = 0$) affects welfare only through consumption; a redistribution that is neutral in MPC-weighted terms affects welfare only through risk premia. In this precise sense, the two channels are additively separable.

\textit{Generality.} This proof establishes the general Proposition~\ref{prop:general_decomp}. The additive separability relies on the distinction between the intertemporal margin (MPC) and the portfolio margin (MPR). This distinction holds in any economy where consumption-savings and portfolio decisions are separable at the margin, which is the case under standard preferences including Epstein-Zin. With non-separable preferences (e.g., if the felicity of consumption depended directly on portfolio composition), the two channels would interact, but the leading-order decomposition would remain valid as an approximation for small aggregate risk. \qed

\subsection{Proof of Proposition~\ref{prop:target}}
\label{app:proof_target}

Setting $dW/d\eps_0^m = 0$:
\begin{equation}
0 = \underbrace{\pder{W}{\pi_0}\cdot\dder{\pi_0}{\eps_0^m}}_{\text{inflation}} + \underbrace{\pder{W}{y_0}\cdot\dder{y_0}{\eps_0^m}}_{\text{output}} + \underbrace{\pder{W}{\gbar}\cdot\dder{\gbar}{\eps_0^m}}_{\text{risk premium}}.
\end{equation}
This three-way decomposition holds when MPCs are identical across agents (as in the baseline model with unitary IES and bond market access for all groups), so that the consumption-redistribution channel in Proposition~\ref{prop:decomp} vanishes and the welfare effect of the monetary shock operates only through aggregate inflation, aggregate output, and the risk premium. When MPCs differ, a fourth term proportional to $\Gamma_c\cdot\int(dn_0^i/d\eps_0^m)\MPC_0^i\,di$ enters the FOC, as in the general decomposition \eqref{eq:general_decomp}; Proposition~\ref{prop:target} then acquires an additional redistribution wedge, as in \citet{bhandari_etal2021}.

The inflation effect captures the welfare cost of deviating from price stability under sticky wages; define $\Gamma_\pi \equiv -(\partial W/\partial\pi_0)(d\pi_0/d\eps_0^m)^{-1} > 0$. The output effect captures the benefit of closing the output gap; define $\Gamma_x > 0$ analogously. The risk premium effect is the new term: $d\gbar/d\eps_0^m = \gbar\cdot\MRC$ by Proposition~\ref{prop:suffstat}, and the welfare effect of a change in $\gbar$ depends on the gap $\gbar - \gstar$ where $\gstar$ is the planner's preferred effective risk aversion. Defining $\Gamma_{rp} > 0$ and rearranging yields \eqref{eq:target}.

\textit{Characterization of $\gstar$.} The socially optimal $\gstar$ is defined in \eqref{eq:gstar} as the consumption-weighted harmonic mean of risk aversion at the planner's preferred allocation. With utilitarian weights, $\gstar < \gbar$ because the competitive equilibrium does not internalize the pecuniary externality: when a household saves more in capital, it marginally lowers the risk premium faced by all others. The gap $\gbar - \gstar$ is increasing in the dispersion of $\gamma^i$ weighted by wealth shares. The welfare approximation \eqref{eq:welfare_approx} follows from substituting the optimal deviations $\pi_0^*, \hat{y}_0^*$ from \eqref{eq:target} into the second-order expansion of $W$ around the competitive equilibrium.

\textit{Generality.} This proof establishes the general Proposition~\ref{prop:general_target}. The decomposition of the planner's FOC into inflation, output, and risk premium channels requires only that (i) nominal rigidities create a welfare cost of inflation, (ii) the output gap enters the planner's criterion, and (iii) the interest rate moves effective risk aversion through $\partial\gbar/\partial i_t = \gbar\cdot\MRC$ (Proposition~\ref{prop:general_suffstat}). The coefficients $\Gamma_\pi, \Gamma_x, \Gamma_{rp}$ are functions of model primitives that differ across environments but are generically positive. \qed

\subsection{Proof of Proposition~\ref{prop:separation}}
\label{app:proof_separation}

With instruments $\{P_0, \tau_0^k\}$, the planner has two FOCs. The portfolio tax enters households' Euler equations but not goods market clearing or the wage rigidity, provided the tax revenue $\tau_0^k \int q_0 k_0^i\,di$ is rebated lump-sum to households (proportional to their wealth shares). Under this assumption, $\tau_0^k$ distorts only the relative return on capital versus bonds, leaving the intertemporal and goods-market margins unaffected. Hence $\partial W/\partial\tau_0^k$ operates exclusively through the risk premium: the FOC for $\tau_0^k$ sets $\gbar = \gstar$ directly. With $\gbar = \gstar$, the risk premium wedge in the FOC for $P_0$ vanishes, yielding $\Gamma_\pi\pi_0 = \Gamma_x\hat{y}_0$.

When $\tau_0^k \geq 0$ is imposed and $\gbar < \gstar$ (risk premia inefficiently low), the planner wants $\tau_0^k > 0$ to discourage risk-taking; the constraint does not bind and separation holds. When $\gbar > \gstar$ (risk premia inefficiently high), the planner wants $\tau_0^k < 0$ to subsidize risk-taking, but the constraint binds at $\tau_0^k = 0$ and a residual gap remains in the monetary policy criterion.

\textit{Generality.} This proof establishes the general Proposition~\ref{prop:general_separation}. The key property is that the portfolio tax $\tau_0^k$ enters agents' portfolio Euler equations but does not directly affect goods market clearing or the nominal rigidity, as long as the tax revenue is rebated lump-sum. This holds in any economy satisfying Assumptions~\ref{ass:redistrib}--\ref{ass:rp_real}, regardless of the source of portfolio heterogeneity or the form of the nominal friction. \qed

\subsection{Proof of Proposition~\ref{prop:timecon}}
\label{app:proof_timecon}

We characterize the Markov-perfect equilibrium and derive the inflation bias in three phases.

\textit{Phase 1: The discretionary planner's problem.} Under discretion, the planner takes the existing portfolio positions $\{\omega_0^{i,e}\}$ as given (where the superscript $e$ denotes that these are the portfolios chosen in anticipation of the planner's policy) and chooses $\pi_0$ to maximize welfare. The FOC is identical to the commitment FOC (Appendix~\ref{app:proof_target}), evaluated at the anticipated portfolios:
\begin{equation}
\label{eq:discretion_foc}
\Gamma_\pi \cdot \pi_0 = \Gamma_x \cdot \hat{y}_0 + \Gamma_{rp} \cdot \MRC(\{\omega_0^{i,e}\}) \cdot (\gbar(\{\omega_0^{i,e}\}) - \gstar)\sigma^2.
\end{equation}
The key difference from commitment is that $\omega_0^{i,e}$ depends on the anticipated inflation $\pi_0^e$, creating a fixed-point problem.

\textit{Phase 2: Portfolio response to anticipated inflation.} Each household $i$ anticipates inflation $\pi_0^e$ and adjusts its portfolio share accordingly. From the portfolio optimality condition \eqref{eq:portfolio_approx} and the redistribution equation \eqref{eq:redistribution}, a higher anticipated inflation raises the return on leveraged capital positions and lowers the return on bonds. Risk-tolerant households (low $\gamma^i$, high MPR) are the leveraged agents; anticipating expropriation through inflation, they reduce their leverage:
\[
\dder{\omega_0^{i,e}}{\pi_0^e} = -\frac{\omega_0^i(1-\omega_0^i)}{\sigma^2}\cdot\pder{(\text{risk premium})}{\pi_0^e} < 0 \quad \text{for high-MPR households},
\]
where the inequality follows because anticipated inflation reduces the excess return on leveraged positions. This implies
\[
\dder{\MRC}{\pi_0^e} < 0 \quad \text{and} \quad \dder{\gbar}{\pi_0^e} > 0:
\]
anticipated inflation reduces MRC and raises effective risk aversion.

\textit{Phase 3: Fixed-point equilibrium.} In a rational expectations equilibrium, $\pi_0^e = \pi_0$. Substituting into \eqref{eq:discretion_foc}, the equilibrium inflation $\pi_0^D$ solves
\[
\Gamma_\pi \cdot \pi_0^D = \Gamma_x\hat{y}_0 + \Gamma_{rp}\cdot\MRC(\pi_0^D)\cdot(\gbar(\pi_0^D) - \gstar)\sigma^2.
\]
Under commitment, the planner can credibly promise $\pi_0^C$ such that households do not adjust portfolios adversarially; the commitment FOC is \eqref{eq:target} evaluated at undistorted portfolios.

We now show $\pi_0^D > \pi_0^C$. Define $\Phi(\pi) \equiv \Gamma_{rp}\cdot\MRC(\pi)\cdot(\gbar(\pi) - \gstar)\sigma^2$. Note: (i) $\Phi(0) > 0$ when $\gbar > \gstar$ (risk premia are inefficiently high); (ii) $\Phi'(\pi) < 0$ because both $\MRC$ and $(\gbar - \gstar)$ respond adversely to anticipated inflation, but $\Phi(\pi)$ remains positive for moderate $\pi$ because the level effect dominates the slope.

The discretionary planner's equilibrium $\pi_0^D$ satisfies $\Gamma_\pi\pi_0^D = \Gamma_x\hat{y}_0 + \Phi(\pi_0^D)$. Under commitment, the planner additionally internalizes the effect of its inflation choice on portfolio decisions through $d\omega_0^{i,e}/d\pi_0^e$ from Step~2. The commitment FOC therefore includes an additional term capturing the welfare cost of portfolio distortion:
\[
\Gamma_\pi\pi_0^C = \Gamma_x\hat{y}_0 + \Phi(\pi_0^C) + \underbrace{\Gamma_{rp}\,\dder{\Phi}{\pi_0^e}\bigg|_{\pi_0^e = \pi_0^C}}_{<\,0},
\]
where the last term is negative because $d\Phi/d\pi_0^e < 0$ (anticipated inflation worsens risk-bearing, from Step~2). Since this negative term reduces the right-hand side at every $\pi_0$, the commitment solution satisfies $\pi_0^C < \pi_0^D$.

\textit{Proportionality to $\Var_i(\MPR_0^i)$.} The inflation bias $\pi_0^D - \pi_0^C$ is proportional to the strength of the redistributive temptation. By the definition of MRC and the zero-sum property $\int(ds_0^i/d\eps_0^m)\,di = 0$:
\[
\MRC = -\int \dder{s_0^i}{\eps_0^m}\MPR_0^i\,di = -\Cov\!\left(\dder{s_0^i}{\eps_0^m},\,\MPR_0^i\right).
\]
Using $\MPR_0^i = \gbar/\gamma^i$ and noting from \eqref{eq:redistribution} that $ds_0^i/d\eps_0^m$ is affine in $\omega_0^i = \MPR_0^i$ (leveraged agents gain more from easing), write $ds_0^i/d\eps_0^m \approx a + b\cdot\MPR_0^i$ for constants $a, b$ with $b < 0$. Then $\MRC = -b\cdot\Var_i(\MPR_0^i) > 0$. Similarly, $\gbar - \gstar$ is increasing in $\Var_i(\MPR_0^i)$ because the pecuniary externality is larger when portfolio heterogeneity is greater. The equilibrium portfolio adjustment (Step~2) also scales with $\Var_i(\MPR_0^i)$, since $d\omega_0^{i,e}/d\pi_0^e$ is heterogeneous only when MPRs differ. Linearizing the fixed-point equation around $\pi_0^C$, the bias is $\pi_0^D - \pi_0^C = -(d\Phi/d\pi_0^e)/\Gamma_\pi$, which is proportional to $\Var_i(\MPR_0^i)$. When $\Var_i(\MPR_0^i) = 0$, $\MRC = 0$, $\Phi = 0$, and the bias vanishes.

\textit{Part (c): Perverse risk premium outcome.} Higher equilibrium inflation under discretion erodes the wealth of risk-tolerant households (who are net nominal borrowers) through the debt deflation channel in \eqref{eq:redistribution}: the \textit{anticipated} inflation reduces their willingness to lever up, raising $\gbar$ and the equilibrium risk premium. Formally, from Step~2, $d\gbar/d\pi_0^e > 0$, so $\gbar^D > \gbar^C$ and the equilibrium risk premium $\gbar^D\sigma^2 > \gbar^C\sigma^2$. Higher risk premia reduce investment via \eqref{eq:dk}.

\textit{Generality.} This proof establishes the general Proposition~\ref{prop:general_timecon}. The argument requires only that (i) the planner can redistribute through surprise inflation (Assumption~\ref{ass:redistrib}), (ii) agents with heterogeneous MPRs adjust portfolios in anticipation of this redistribution (Assumption~\ref{ass:mpr_het}), and (iii) the resulting change in risk premia affects real activity (Assumption~\ref{ass:rp_real}). The proportionality to $\Var_i(\MPR_t^i)$ follows from the fact that the temptation to redistribute and the portfolio response are both governed by MPR heterogeneity. \qed

\subsection{Robustness: Other Sources of Heterogeneity}
\label{app:robustness}

With binding portfolio constraints ($q_0 k_0^i = \bar{\omega}^i a_0^i$), heterogeneous background risk 
\[
\log\epsilon_1^i \sim \mathcal{N}(-\eta^i\sigma^2/2,\; \eta^i\sigma^2),
\]
or heterogeneous beliefs ($\eps_1^z \sim_i \mathcal{N}(-\varsigma^i\sigma^2/2,\, \varsigma^i\sigma^2)$), the portfolio share in capital is decreasing in $\gamma^i$, $\eta^i$, and $\varsigma^i$, and increasing in $\bar{\omega}^i$. Households holding levered positions remain those with high MPRs. The structure of all proofs is unchanged; the risk premium expression generalizes to $\gbar$ defined over the broader set of portfolio-relevant parameters. \qed

\section{SCF Data Construction}
\label{app:scf}

\subsection{Portfolio Decomposition}

We decompose household net worth $A^i$ into capital claims $Qk^i$ and nominal claims $B^i = A^i - Qk^i$ for each SCF wave.

\textit{Defined benefit pensions.} We add estimates of DB pension wealth from \citet{sabelhaus_volz2019}, treating them as nominal claims of households (the household has a fixed claim on the pension sponsor, who is the residual claimant on the investment portfolio).

\textit{Direct capital claims.} These are nonfinancial assets: vehicles, primary residence, other residential real estate, non-residential real estate, and miscellaneous nonfinancial assets.

\textit{Indirect capital claims through equity.} Publicly traded equity includes stock mutual funds, other mutual funds, directly held stocks, 50\% of combination mutual funds, and the self-reported stock fraction of retirement and managed accounts. Private business equity is the value of actively and passively managed businesses. If household $i$ owns \$1 in equity of a firm with leverage $\ell^i$, we assign $Qk^i_{\text{equity}} = \ell^i$ and $B^i_{\text{equity}} = 1 - \ell^i$.

\textit{Aggregate leverage.} We compute aggregate leverage for publicly traded firms as the ratio of capital to equity using the Financial Accounts: total assets net of nominal and equity assets, divided by equity liabilities plus net worth net of equity assets. We use the consolidated nonfinancial corporate sector (FA table S.5.a) and financial business sector (FA table S.6.a), netting out the central bank, government DB pension funds, DC pension funds, and mutual funds (which are pass-through entities). For private businesses, we use the nonfinancial noncorporate sector (FA table S.4.a) and non-profit sector (FA table B.101.n). The resulting aggregate leverage is approximately 1.5--1.6 for public equities and 1.1 for private businesses, varying modestly across years.

\textit{Idiosyncratic leverage dispersion.} We allow household-specific leverage $\ell^i = \bar{\ell} \cdot \epsilon^i$ where $\epsilon^i \sim \Gamma(\theta^{-1}, \theta)$ has mean one. We calibrate $\theta = 1.05$ so that the cross-sectional standard deviation of leverage on assets equals 31\% of the aggregate leverage in public equities, matching the evidence on expected return heterogeneity in \citet{bach_calvet_sodini2020}.

\textit{Remaining items.} All other line items not classified above (checking accounts, savings bonds, CDs, cash value of life insurance, other financial assets, credit card debt, other consumer debt, mortgage debt) are classified as nominal claims.

\subsection{Group Sorting}

For each SCF wave, we sort households by wealth-to-labor-income $A^i/(W\ell^i)$ and capital portfolio share $Qk^i/A^i$ (the latter computed excluding primary residence and vehicles from both numerator and denominator, to isolate financial portfolio choice from housing consumption). ``High'' wealth-to-income is above the 60th percentile; ``high'' capital share is above the 90th percentile among the high-wealth group. Group $a$ is high-wealth, high-leverage; group $b$ is high-wealth, low-leverage; group $c$ is low-wealth.

\subsection{MPR and Exposure Computation}

\textit{MPRs.} For each group, we infer risk aversion $\gamma^g$ from the observed capital portfolio share using the portfolio choice equation \eqref{eq:portfolio_approx}, accounting for the labor income hedging motive. The MPR is then $\MPR^g = \gbar/\gamma^g$ from equation \eqref{eq:mpr_limit}. Group $c$ households, whose wealth is dominated by housing and who have little financial wealth relative to income, are treated as constrained with $\MPR^c = 0$.

\textit{Exposures.} Each group's exposure to monetary policy $ds^g/d\eps^m$ is computed from the balance sheet decomposition \eqref{eq:redistribution}. The debt deflation component is proportional to the group's nominal bond position $B^g$. The capital revaluation component is proportional to the group's excess capital holdings $k^g - s^g k$. We evaluate these at the baseline parameterization of the model (Section~\ref{sec:quant}).

\section{Empirical Robustness}
\label{app:empirical_robust}

We verify the robustness of the MRC interaction result along several dimensions.

\medskip\noindent\textbf{Additional controls.}\quad
Table~\ref{tab:empirical_robust} reports the impact interaction coefficient $\hat{\gamma}_0$ from the local projection-IV when augmented with additional interaction terms. The concern is that MRC may proxy for other state variables that affect the transmission of monetary shocks. We add interactions with: the Gini coefficient of wealth (to control for inequality), the market-to-GDP ratio (to control for asset valuation levels), the output gap (to control for the business cycle), and the \citet{gilchrist_zakrajsek2012} excess bond premium (to control for credit conditions). In each case, the MRC interaction coefficient remains economically similar.

\begin{table}[htbp]
\centering
\caption{Robustness of MRC interaction coefficient $\hat{\gamma}_0$}
\label{tab:empirical_robust}
\begin{tabular}{lcc}
\toprule
Specification & $\hat{\gamma}_0$ & Std.\ error \\
\midrule
Baseline (Table~\ref{tab:lp_interaction})        & $-$0.94 & (0.48) \\
+ Gini coefficient interaction                    & $-$0.88 & (0.52) \\
+ Market/GDP ratio interaction                    & $-$0.91 & (0.50) \\
+ Output gap interaction                          & $-$0.85 & (0.53) \\
+ Excess bond premium interaction                 & $-$0.82 & (0.55) \\
+ All four controls                               & $-$0.78 & (0.60) \\
\midrule
3-month ahead FF futures as IV                    & $-$1.02 & (0.55) \\
4 lags in VAR                                     & $-$0.89 & (0.50) \\
8 lags in VAR                                     & $-$0.97 & (0.52) \\
IV sample: 1/91--9/01                             & $-$1.25 & (0.65) \\
\bottomrule
\end{tabular}
\begin{minipage}{0.88\textwidth}
\footnotesize \vspace{0.5em}
\textit{Notes:} Each row reports the impact ($h=0$) interaction coefficient
$\hat{\gamma}_0$ from the local projection-IV specification~\eqref{eq:lp_iv},
modified as indicated. A negative coefficient means higher MRC amplifies the equity response to monetary shocks. Standard errors are Newey-West with 1 lag.
\end{minipage}
\end{table}

\medskip\noindent\textbf{Alternative instruments and lag structures.}\quad
Using the 3-month-ahead Fed Funds futures contract as the instrument (instead of the current-month contract) produces similar results. Varying the number of VAR lags from 4 to 8 has little effect on the MRC interaction. Restricting the IV sample to the first half (January 1991--September 2001) yields a larger interaction, consistent with MRC being higher in this period.

\medskip\noindent\textbf{Alternative calibrations of MRC.}\quad
We re-compute MRC using the p99 cutoff for the capital portfolio share (Section~\ref{sec:quant:env}). The resulting MRC series has similar time-series properties but higher variance, reflecting the more extreme leverage of the top 1\%. The interaction coefficient is qualitatively similar.

\section{Quantitative Robustness}
\label{app:quant_robust}

\medskip\noindent\textbf{Alternative calibrations.}\quad
Figure~\ref{fig:robust_calib} and Table~\ref{tab:robust_cs} compare impulse responses and Campbell-Shiller decompositions across three calibrations: the baseline p90, the p99 calibration (where group $a$ corresponds to broker-dealers and hedge funds with leverage of 4.4), and the idiosyncratic risk microfoundation (where portfolio heterogeneity arises from heterogeneous background risk rather than risk aversion).

\begin{figure}[htbp]
\centering
\includegraphics[width=\textwidth]{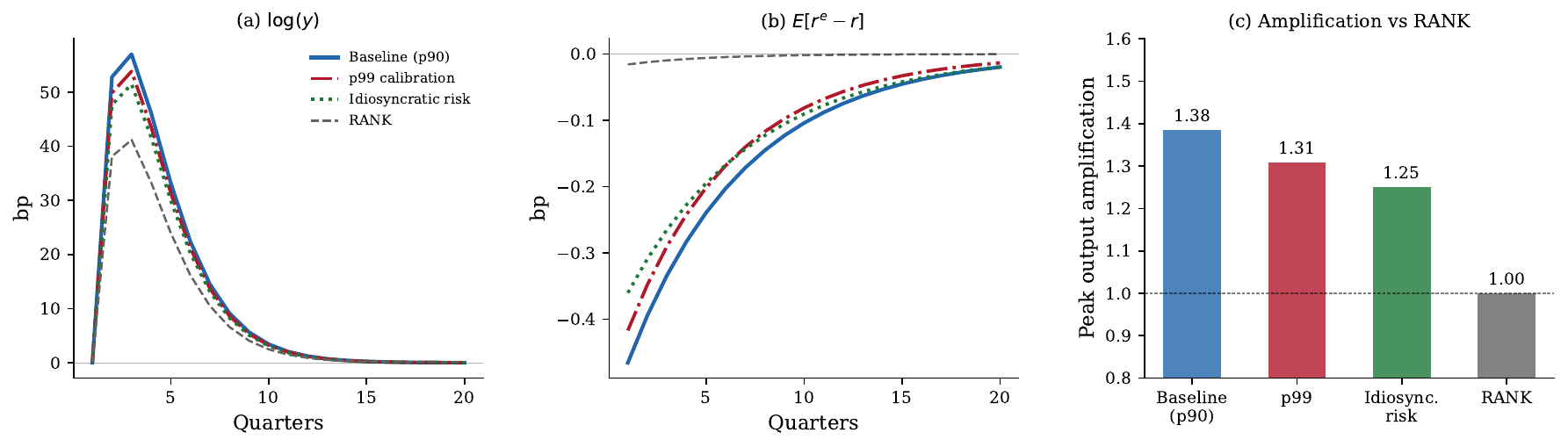}
\caption{Robustness across calibrations}
\label{fig:robust_calib}
\begin{minipage}{0.92\textwidth}
\footnotesize \textit{Notes:} Responses to a monetary easing that lowers the 1-year nominal yield by 0.2pp. ``Baseline'' is the p90 calibration, ``p99'' raises the cutoff to the 99th percentile, and ``Idiosyncratic risk'' replaces risk aversion heterogeneity with background risk heterogeneity. All three generate amplification relative to RANK; the share due to excess return news ranges from 28\% to 37\%.
\end{minipage}
\end{figure}

\begin{table}[htbp]
\centering
\caption{Robustness of Campbell-Shiller decomposition across calibrations}
\label{tab:robust_cs}
\begin{tabular}{lcccc}
\toprule
\% of real stock return & Baseline & p99 & Idiosync.\ risk & RANK \\
\midrule
Dividend growth news          & 50\% & 49\% & 51\% & 65\% \\
$-$Future real rate news       & 17\% & 14\% & 21\% & 35\% \\
$-$Future excess return news   & 33\% & 37\% & 28\% & 0\% \\
\midrule
Amplification vs RANK         & 1.38 & 1.31 & 1.25 & 1.00 \\
\bottomrule
\end{tabular}
\begin{minipage}{0.88\textwidth}
\footnotesize \vspace{0.5em}
\textit{Notes:} ``Baseline'' is the p90 calibration from Section~\ref{sec:quant:env}. ``p99'' raises the capital portfolio share cutoff to the 99th percentile, so group $a$ agents hold 2\% of wealth with leverage of 4.4. ``Idiosyncratic risk'' replaces heterogeneous risk aversion with heterogeneous background risk. In all cases, the monetary shock delivers a 0.2pp decline in the 1-year nominal yield.
\end{minipage}
\end{table}

The risk premium channel and its amplification of monetary policy are robust across all three microfoundations. The amplification ranges from 1.25$\times$ (idiosyncratic risk) to 1.38$\times$ (baseline), and the share of the stock market return due to excess return news ranges from 28\% to 37\%, all within the empirical confidence interval. This confirms that the results depend on the distribution of MPRs and exposures, not on the specific primitive generating portfolio heterogeneity. The post-2008 crisis narrative, including the MRC collapse and the interpretation of quantitative easing as a direct MRC intervention, is presented in Section~\ref{sec:quant:crisis}.

\medskip\noindent\textbf{Alternative IES.}\quad
Setting $\psi = 0.5$ (more complementarity between consumption and leisure) slightly amplifies the risk premium wedge in the optimal policy criterion because the distinction between portfolios and MPRs is larger. Setting $\psi = 1.5$ reduces the wedge. The qualitative results are unchanged across $\psi \in \{0.5, 0.8, 1.0, 1.5\}$.

\medskip\noindent\textbf{Alternative Pareto weights.}\quad
Under wealth-weighted Pareto weights (overweighting the wealthy, who are disproportionately risk-averse group $b$ agents), $\gstar$ rises above $\gbar$ and the risk premium wedge can flip sign: the planner may prefer higher risk premia to protect the wealth of conservative savers. Under equal weights, the baseline results are amplified because the planner places relatively more weight on the consumption of the numerous group $c$ agents, who benefit from lower risk premia through higher investment and wages. The ``lean against the wind'' question is fundamentally a question about whose welfare the central bank prioritizes.

\section{Aggregation into Representative Households}
\label{app:aggregation}

Within each group $i \in \{a,b,c\}$, we prove that households' optimal policies are homogeneous of degree one in wealth, implying a representative household.

\begin{proposition}
\label{prop:aggregation}
Household $\iota$ in group $i(\iota)$ has optimal policies satisfying $v_t^\iota = \mu_t^\iota V_t^{i(\iota)}$, $c_t^\iota = \mu_t^\iota c_t^{i(\iota)}$, $B_t^\iota = \mu_t^\iota B_t^{i(\iota)}$, $k_t^\iota = \mu_t^\iota k_t^{i(\iota)}$, $\bar{\ell}_t^\iota = \mu_t^\iota$, where the variables with superscript $i$ correspond to those of a representative household endowed with aggregate group-specific wealth. The within-group wealth ratio evolves as $\mu_{t+1}^\iota / \mu_t^\iota = \mu_{t,t+1}^{i(\iota)}$.
\end{proposition}

This result follows from households' ability to trade claims on a labor endowment with other households in the same group, the assumed labor allocation rule, and the lump-sum transfer structure. The proof mirrors standard arguments for homothetic preferences with Epstein-Zin utility; we omit the details for brevity.

\medskip\noindent\textbf{Why within-group heterogeneity does not affect MRC.}\quad
Proposition~\ref{prop:aggregation} implies that \textit{within} each group, all households hold the same portfolio share $\omega_t^i$ regardless of their individual wealth $\mu_t^\iota$. This is a consequence of the homotheticity of Epstein-Zin preferences: the portfolio choice problem scales linearly with wealth. Therefore, the MPR is constant within each group: $\MPR_t^\iota = \MPR_t^{i(\iota)}$ for all $\iota$ in group $i$. 

This has a direct implication for MRC. Since MRC $= \sum_i (ds_t^i / di_t) \cdot \MPR_t^i$, and both the exposure $ds_t^i/di_t$ and the MPR $\MPR_t^i$ are group-level objects (aggregated across all households within a group), within-group wealth heterogeneity cancels out in the computation of MRC. A group with a Gini coefficient of zero and a group with a Gini of 0.5 produce the same MRC, provided they have the same aggregate wealth share and the same MPR.

The relevant heterogeneity for MRC is therefore \textit{between}-group heterogeneity in MPRs. The three-group structure is the minimal partition that captures the essential features of this between-group variation: a high-MPR leveraged group, a low-MPR conservative group, and a constrained group with near-zero MPR. The 15-type validation in Section~\ref{sec:quant:trap} confirms that finer partitions, which introduce additional between-group variation in MPRs, produce nearly identical MRC (correlation 0.98) and welfare gains.

The aggregation result would fail if within-group heterogeneity were correlated with portfolio choice. If, for example, wealthier households within group $a$ held different portfolio shares than poorer households within group $a$. This would require non-homothetic preferences or within-group heterogeneity in risk aversion. We abstract from both, but note that introducing decreasing relative risk aversion within groups (as in \citealt{chien_cole_lustig2012}) would strengthen our results: within a group of nominally ``risk-tolerant'' agents, the wealthiest would hold the most leveraged positions and have the highest MPRs, amplifying both MRC and the welfare gains from optimal policy.

\section{Computational Details}
\label{app:computation}

\medskip\noindent\textbf{Solution method.}\quad
We solve the model globally using sparse grids following \citet{judd_etal2014}. The state space has six dimensions: scaled capital, scaled prior wage, wealth shares of groups $a$ and $c$, disaster probability, and the monetary policy state. We use Chebyshev polynomials for interpolation off grid points and Gauss-Hermite quadrature (5 nodes for each continuous shock $\eps^z$, $\eps^p$) and two-point quadrature for the disaster indicator for forming expectations. The stochastic equilibrium is determined through backward iteration with dampened updating of asset prices and expectations (dampening factor 0.3). The code is written in Julia and parallelized using multithreading. Convergence is achieved when the maximum absolute change in policy functions falls below $10^{-8}$; this typically requires 2,000--3,000 iterations and takes approximately 10 minutes on a 16-core workstation.

\medskip\noindent\textbf{Optimal simple rules.}\quad
We optimize over $\{\phi_\pi, \phi_y, \phi_s, \phi_{rp}\}$ using Bayesian optimization with a Gaussian process surrogate. Each evaluation requires solving the model and simulating 50,000 quarters ($\approx$5 minutes). We use 300 evaluations (50 initial random points, 250 sequential acquisition points using expected improvement). We verify convergence by restarting from multiple initial conditions.

\medskip\noindent\textbf{Welfare computation.}\quad
Welfare is computed as the population-weighted sum of expected lifetime utilities along a 50,000-quarter simulation after a 5,000-quarter burn-in, with no disaster realizations. We report consumption-equivalent welfare gains: the permanent percentage change in consumption for all household types that would make them indifferent between two regimes. Standard errors on welfare differences are computed using batch means with 50 batches of 1,000 quarters each.

\medskip\noindent\textbf{Markov-perfect equilibrium.}\quad
For the discretionary equilibrium, we iterate on the planner's policy function using the approach of \citet{klein_krusell_riosrull2008}. At each state, the planner chooses $i_t$ to maximize current welfare plus the continuation value, where the continuation value is determined by the same policy function. We iterate until the policy function converges (tolerance $10^{-6}$), starting from the Taylor rule as the initial guess.

\section{The Ramsey Problem in the Infinite Horizon}
\label{app:ramsey_details}

This appendix provides the full formulation of the Ramsey optimal monetary policy problem, states the complete system of implementability constraints, derives the planner's first-order condition, and connects it to the two-period target criterion in Proposition~\ref{prop:target}.

\subsection{Competitive Equilibrium Conditions}
\label{app:ramsey:ce}

The competitive equilibrium is characterized by the following system of conditions, which constitute the constraints on the Ramsey planner. We work in the scaled stationary economy; tildes on variables are suppressed for readability.

\medskip\noindent\textbf{Household bond Euler equations ($I = 3$ conditions).}\quad
For each group $i \in \{a,b,c\}$:
\begin{multline}
\label{eq:app_bond_euler}
\frac{1}{c_t^i\Phi(\ell_t^i)} = \beta(1-\xi)(1+i_t) \\
  \times\E_t\!\left[\frac{e^{-(1-1/\psi)(\eps_{t+1}^z + \varphi_{t+1})}}{1+\pi_{t+1}}\,\frac{(V_{t+1}^i)^{1/\psi - \gamma^i}}{\bigl(\E_t[(V_{t+1}^i)^{1-\gamma^i}]\bigr)^{(1/\psi-\gamma^i)/(1-\gamma^i)}}\,\frac{1}{c_{t+1}^i\Phi(\ell_{t+1}^i)}\right],
\end{multline}
where $\pi_{t+1} \equiv P_{t+1}/P_t - 1$ is inflation and $\Phi$ is the labor disutility function \eqref{eq:labor_disutility}.

\medskip\noindent\textbf{Household portfolio Euler equations ($I-1 = 2$ independent conditions).}\quad
For each group $i$ not at the capital constraint:
\begin{equation}
\label{eq:app_portfolio_euler}
\E_t\!\left[M_{t+1}^i\left(\frac{\pi_{t+1}^k + (1-\delta)q_{t+1}}{q_t}\,e^{\varphi_{t+1}} - \frac{1+i_t}{1+\pi_{t+1}}\right)\right] = 0,
\end{equation}
where $\pi_{t+1}^k \equiv \Pi_{t+1}/(P_{t+1}z_{t+1})$ is the scaled real dividend, $q_t \equiv Q_t/(P_tz_t)$ is the scaled real price of capital, and the SDF $M_{t+1}^i$ is given in \eqref{eq:sdf}. When the constraint \eqref{eq:k_lb} binds for group $c$, replace \eqref{eq:app_portfolio_euler} with $k_t^c = \underline{k}$.

\medskip\noindent\textbf{New Keynesian wage Phillips curve (1 condition).}\quad
In the symmetric equilibrium, the union's optimality condition for the aggregate nominal wage $W_t$ is
\begin{multline}
\label{eq:app_nkwpc}
\chi_W\!\left(\frac{W_t}{W_{t-1}e^{\varphi_t}} - 1\right)\!\frac{W_t}{W_{t-1}e^{\varphi_t}} = (1-\varepsilon) + \varepsilon\,\frac{\sum_i\lambda^i\alpha_w^i\,\bar{\theta}^i(\ell_t^i)^{1/\theta}\ell_t/c_t^i}{\sum_i\lambda^i\alpha_w^i / c_t^i} \\
  + \beta\chi_W\,\E_t\!\left[\frac{\sum_i\lambda^i\alpha_w^i/c_{t+1}^i}{\sum_i\lambda^i\alpha_w^i/c_t^i}\left(\frac{W_{t+1}}{W_te^{\varphi_{t+1}}} - 1\right)\frac{W_{t+1}}{W_te^{\varphi_{t+1}}}\frac{\ell_{t+1}}{\ell_t}\right],
\end{multline}
where $\alpha_w^i$ are the union's welfare weights (equal to population shares $\lambda^i$ for a utilitarian union).

\medskip\noindent\textbf{Firm optimality (2 conditions).}\quad
Labor demand:
\begin{equation}
\label{eq:app_labor_demand}
w_t = (1-\alpha)\,\frac{y_t}{\ell_t},
\end{equation}
where $w_t \equiv W_t/(P_tz_t)$ is the scaled real wage. Investment:
\begin{equation}
\label{eq:app_investment}
q_t = \left(\frac{k_t}{k_{t-1}e^{\varphi_t}}\right)^{\!\chi_x}.
\end{equation}

\medskip\noindent\textbf{Production (1 condition).}\quad
\begin{equation}
\label{eq:app_production}
y_t = \ell_t^{1-\alpha}\,(k_{t-1}e^{\varphi_t})^\alpha.
\end{equation}

\medskip\noindent\textbf{Goods market clearing (1 condition).}\quad
\begin{equation}
\label{eq:app_goods_mc}
\sum_i \lambda^i c_t^i + q_t(k_t - (1-\delta)k_{t-1}e^{\varphi_t}) = y_t - \frac{\chi_W}{2}\left(\frac{W_t}{W_{t-1}e^{\varphi_t}} - 1\right)^{\!2} w_t\ell_t.
\end{equation}

\medskip\noindent\textbf{Wealth share evolution ($I-1 = 2$ conditions).}\quad
For groups $a$ and $c$:
\begin{equation}
\label{eq:app_wealth_share}
s_{t+1}^i = \lambda^i(1-\xi)\,\frac{(1+i_t)b_t^i/(1+\pi_{t+1}) + (\pi_{t+1}^k + (1-\delta)q_{t+1})k_t^ie^{\varphi_{t+1}}}{(\pi_{t+1}^k + (1-\delta)q_{t+1})k_te^{\varphi_{t+1}}} + \xi\bar{s}^i,
\end{equation}
where $b_t^i \equiv B_t^i/(P_tz_t)$ is the scaled real bond position. The share of group $b$ satisfies $s^b = 1 - s^a - s^c$.

\medskip\noindent\textbf{Budget constraints ($I = 3$ conditions).}\quad
\begin{multline}
\label{eq:app_budget}
c_t^i + b_t^i + q_t k_t^i = (1-\tau)\phi^i w_t\ell_t - \text{ac}_t^W + \frac{(1+i_{t-1})b_{t-1}^i}{1+\pi_t} \\
  + (\pi_t^k + (1-\delta)q_t)k_{t-1}^i e^{\varphi_t} + \text{tr}_t^i,
\end{multline}
where $\text{ac}_t^W$ denotes scaled wage adjustment costs and $\text{tr}_t^i$ denotes scaled transfers.

\subsection{The Ramsey Lagrangian}
\label{app:ramsey:lagrangian}

Let $\mathbf{x}_t = \{c_t^i, k_t^i, b_t^i, \ell_t, w_t, q_t, \pi_t, s_t^a, s_t^c, k_t, y_t\}$ denote the vector of endogenous variables, and let $i_t$ denote the planner's instrument. The Ramsey problem can be written as a constrained optimization with the Lagrangian
\begin{multline}
\label{eq:app_lagrangian}
\mathcal{L} = \E_0 \sum_{t=0}^\infty \beta^t \Bigl[\sum_i \lambda^i\alpha^i u(c_t^i, \ell_t^i) 
  + \sum_j \mu_t^j \cdot g_j(\mathbf{x}_t, \mathbf{x}_{t-1}, \E_t[\mathbf{x}_{t+1}], i_t, \eps_t)\Bigr],
\end{multline}
where $g_j(\cdot) = 0$ are the implementability constraints enumerated in Section~\ref{app:ramsey:ce} and $\mu_t^j$ are the corresponding multipliers. Specifically, the constraints include: bond Euler equations~\eqref{eq:app_bond_euler}, portfolio Euler equations~\eqref{eq:app_portfolio_euler}, wage Phillips curve~\eqref{eq:app_nkwpc}, labor demand~\eqref{eq:app_labor_demand}, investment~\eqref{eq:app_investment}, production~\eqref{eq:app_production}, goods market clearing~\eqref{eq:app_goods_mc}, wealth shares~\eqref{eq:app_wealth_share}, and budgets~\eqref{eq:app_budget}. Because several constraints involve expectations of future variables, the Lagrangian involves ``forward-looking'' multipliers that create intertemporal links characteristic of Ramsey problems.

Under commitment, the planner chooses $\{i_t, \mathbf{x}_t, \mu_t^j\}_{t=0}^\infty$ to satisfy the first-order conditions of \eqref{eq:app_lagrangian}. This yields a system of difference equations in the endogenous variables and multipliers, which we solve globally using the approach described in Section~\ref{sec:quant:solution}.

\subsection{The Planner's First-Order Condition for \texorpdfstring{$i_t$}{i\_t}}
\label{app:ramsey:foc}

The FOC of \eqref{eq:app_lagrangian} with respect to $i_t$ is
\begin{equation}
\label{eq:app_ramsey_foc_full}
0 = \sum_j \mu_t^j \cdot \pder{g_j}{i_t}.
\end{equation}
The interest rate $i_t$ enters three sets of constraints directly: the bond Euler equations \eqref{eq:app_bond_euler}, the portfolio Euler equations \eqref{eq:app_portfolio_euler}, and the wealth share evolution \eqref{eq:app_wealth_share} (through $(1+i_t)b_t^i$). It does not enter the wage Phillips curve, firm optimality, or production directly.

Expanding \eqref{eq:app_ramsey_foc_full}:
\begin{multline}
\label{eq:app_ramsey_foc_expanded}
0 = \underbrace{\sum_i \mu_t^{\text{bond},i} \cdot \pder{}{i_t}\bigl[\text{bond Euler}^i\bigr]}_{\text{intertemporal channel}} + \underbrace{\sum_i \mu_t^{\text{port},i} \cdot \pder{}{i_t}\bigl[\text{portfolio Euler}^i\bigr]}_{\text{portfolio channel}} \\
  + \underbrace{\beta\,\E_t\!\left[\sum_i \mu_{t+1}^{s^i} \cdot \pder{s_{t+1}^i}{i_t}\right]}_{\text{redistribution channel}}.
\end{multline}
The first term captures the conventional intertemporal channel: the interest rate affects the intertemporal allocation of consumption through the bond Euler equation. The second term captures the portfolio channel: the interest rate affects the relative return on capital versus bonds through the portfolio Euler equation. The third term, the redistribution channel, is the dynamic counterpart of the risk premium wedge in Proposition~\ref{prop:target}: through \eqref{eq:app_wealth_share}, the interest rate affects the wealth share of each group in the next period, which in turn affects future risk premia and allocations.

\subsection{Connection to the Two-Period Target Criterion}
\label{app:ramsey:connection}

To connect \eqref{eq:app_ramsey_foc_expanded} to the two-period result, consider the special case $\beta \to 0$ (the two-period limit). The continuation value drops out, and the redistribution channel reduces to a static term proportional to $\partial s_1^i / \partial i_0$. The bond Euler multiplier $\mu_t^{\text{bond},i}$ is proportional to the Pareto-weight-adjusted marginal utility $\alpha^i / c_0^i$, and the portfolio multiplier $\mu_t^{\text{port},i}$ captures the welfare effect of distorting the risk--return tradeoff.

Under the approximations used in Section~\ref{sec:analytical} (small $\sigma$, unitary IES), the three terms in \eqref{eq:app_ramsey_foc_expanded} map one-to-one to the three terms in the target criterion \eqref{eq:target}: the intertemporal channel generates the inflation cost ($\Gamma_\pi\pi_0$) and the output gap benefit ($\Gamma_x\hat{y}_0$), while the redistribution channel generates the risk premium wedge ($\Gamma_{rp} \cdot \MRC \cdot (\gbar - \gstar)\sigma^2$).

In the infinite horizon, the redistribution channel acquires a forward-looking component: the planner values redistribution today not only for its current effect on risk premia but also for its effect on the future wealth distribution, which determines future risk premia, future effectiveness of monetary policy, and the probability of entering the risk capacity trap. This forward-looking motive is the reason the Ramsey allocation differs from any static rule: the planner ``saves'' risk-bearing capacity for the future by preventing the wealth distribution from drifting too far from the steady state.

\subsection{Commitment Versus Discretion}
\label{app:ramsey:discretion}

Under discretion, the planner takes as given the Lagrange multipliers for forward-looking constraints. Equivalently, the discretionary planner chooses $i_t$ to maximize the right-hand side of the Bellman equation \eqref{eq:ramsey_obj} treating its own future policy function $i^D(\Omega)$ as given. The key difference from commitment is that the discretionary planner cannot credibly promise to stabilize the wealth distribution in the future: since $i^D(\Omega)$ is a function of the current state only, there is no mechanism by which current promises about future policy affect current household decisions. In the Markov-perfect equilibrium, the redistribution channel in \eqref{eq:app_ramsey_foc_expanded} is evaluated at portfolios that anticipate the discretionary policy, yielding the inflation bias described in Proposition~\ref{prop:timecon}.

The discretionary planner faces a temptation to inflate at each date, since unexpected inflation redistributes to leveraged high-MPR agents and compresses risk premia. But households anticipate this temptation and adjust their portfolios accordingly. In equilibrium, the ex ante expected inflation under discretion exceeds that under commitment, and the risk premium is higher, because leveraged households demand compensation for the expected expropriation. The magnitude of the bias is governed by the cross-sectional variance of MPRs (Proposition~\ref{prop:timecon}(b)), which determines the strength of the redistribution channel.

\end{appendices}

\end{document}